\documentclass[final,hidelinks,onefignum,onetabnum]{siamart220329}

\usepackage[color]{showkeys}
\usepackage{etex}
\usepackage[utf8]{inputenc}
\usepackage{todonotes}
\usepackage{xspace}
\usepackage{latexsym}
\usepackage{amsmath}
\usepackage{stmaryrd}
\usepackage{ifthen}
\usepackage{supertabular}
\usepackage{tcolorbox}
\usepackage{multicol}
\usepackage{tikz}
\usepackage{algorithmicx}
\usepackage{algpseudocode}
\usepackage{amsfonts}
\usepackage{amssymb}

\usetikzlibrary{snakes,shapes,arrows,automata}

\newcounter{showLineNumbersLeft}
\newcounter{showIndexEntries}
\newcounter{showLabels}

\usepackage{prettyref}

\colorlet{DMnormalbackcolor}{gray!25}
\colorlet{DMlightbackcolor}{gray!10}
\colorlet{DMmediumbackcolor}{gray!20}
\colorlet{DMdarkbackcolor}{gray!30}
\colorlet{DMmediumforecolor}{black!57}
\colorlet{DMlightforecolor}{black!46}

\newlength{\ksize} 

\ifthenelse{\equal{\value{showLabels}}{1}}{	\usepackage[notref,notcite]{showkeys}}{}

\ifthenelse{\equal{\value{showIndexEntries}}{1}}{	\usepackage{showidx}}{}

\newsiamremark{problem}{Problem}
\newsiamremark{examples}{Examples}

\definecolor{darkgreen}{rgb}{0.0,0.7,0.0}
\newenvironment{pss}{\color{red}} {}

\newenvironment{vd}{\color{blue}} {}

\newenvironment{ipo}{\color{magenta}} {}

\newcommand{\set}[2]{\{{#1}\,\mid\,{#2}\}}

\newcommand{\os}[1]{\{{#1}\}}

\newcommand{\dedfin}{De\-de\-kind-finite\xspace}

\newcommand{\enu}{enumerable representation\xspace}\newcommand{\Wlog}{w.l.o.g.\@\xspace}
\newcommand{\WLOG}{W.l.o.g.\@\xspace}
\newcommand{\SNF}{Smith normal form\xspace}\newcommand{\fg}{f.g.\@\xspace}

\newcommand\lcm{\mathop\text{lcm}}
\newcommand{\tra}{transition\xspace}
\newcommand{\tras}{transitions\xspace}
\newcommand{\sztra}{rank-$1$ transition\xspace}
\newcommand{\sztras}{rank-$1$ transitions\xspace}

\newcommand{\invol}{involution\xspace}

\newcommand{\wrt}{with respect to{}\xspace}

\newcommand{\arc}[1]{\overset{#1}\ra}
\newcommand{\larc}[1]{\overset{#1}\longrightarrow}
\newcommand{\Arc}[1]{\overset{#1}\longrightarrow}

\newcommand\lds{,\ldots ,} 
 
\newcommand{\sse}{\subseteq}
\newcommand{\es}{\emptyset}
\newcommand{\sm}{\setminus}

\newcommand{\wt}[1]{\widetilde{#1}}

\newcommand{\homo}{homo\-mor\-phism\xspace}

\newcommand{\iso}{iso\-mor\-phism\xspace}

\newcommand{\epi}{epi\-mor\-phism\xspace}

\newcommand{\BS}{\ensuremath{\mathop\mathrm{BS}}}  

\newcommand\SL{\mathop\mathrm{SL}}
\newcommand\UT{\mathop\mathrm{UT}}

\newcommand\GL{\mathop\mathrm{GL}}
\newcommand\PSL{\mathop\mathrm{PSL}}

\newcommand\SLZ{\SL(2,\Z)}

\newcommand\GLZ{\GL(2,\Z)}

\newcommand\GLQ{\GL(2,\Q)}

\newcommand\MZ{\mathbb{Z}^{2\times 2}}
\newcommand\MQ{\mathbb{Q}^{2\times 2}}
\newcommand\MnR{R^{n\times n}}
\newcommand\MnZ{\Z^{n\times n}}
\newcommand\MnQ{\Q^{n\times n}}

\newcommand{\VDmatrix}[4]{\begin{pmatrix}{#1}&{#2}\\{#3}&{#4}\end{pmatrix}}
\newcommand{\vdmatrix}[4]{\left(\begin{smallmatrix}#1 & #2\\ #3 & #4\end{smallmatrix}\right)}

\newcommand{\N}{\mathbb{N}}
\newcommand{\Z}{\mathbb{Z}}
\newcommand{\Q}{\mathbb{Q}}
\newcommand{\R}{\mathbb{R}}

\newcommand{\DTIME}{\mathrm{DTIME}}
\newcommand{\NTIME}{\mathrm{NTIME}}

\newcommand{\PTIME}{\mathbf{P}}
\renewcommand{\P}{\PTIME}
\newcommand{\NP}{\mathbf{NP}}

\newcommand{\DEXPTIME}{\mathbf{EXPTIME}}\newcommand{\NEXPTIME}{\mathbf{NEXPTIME}}

\newcommand{\Oh}{\mathop{\mathcal{O}}}
 \newcommand{\cA}{\mathcal{A}}

\newcommand{\cB}{\mathcal{B}}
\newcommand{\cC}{\mathcal{C}}

\newcommand{\cL}{\mathcal{L}}

\newcommand{\cP}{\mathcal{P}}
\newcommand{\cQ}{\mathcal{Q}}
\newcommand{\cR}{\mathcal{R}}

\newcommand{\cT}{\mathcal{T}}

\newcommand{\Sig}{\Sigma}
\newcommand{\Del}{\Delta}
\newcommand{\Gam}{\Gamma}

\newcommand{\del}{\delta}
\newcommand{\eps}{\varepsilon}

\newcommand{\alp}{\alpha}
\newcommand{\bet}{\beta}

\newcommand{\ovv}{\overline{\phantom a}}
\newcommand{\ov}[1]{\overline{#1}}
\newcommand{\oi}[1]{{#1}^{-1}}

\newcommand{\id}{\mathrm{id}}

\newcommand{\abs}[1]{\left|\mathinner{#1}\right|}
\newcommand{\Abs}[1]{\left\Vert\mathinner{#1}\right\Vert}
\newcommand{\Abm}[1]{{\Vert\mathinner{#1}\Vert}_{\text{max}}}
\newcommand{\Abbin}[1]{{\Vert\mathinner{#1}\Vert}_{\text{bin}}}

\newcommand{\gen}[1]{\left< \mathinner{#1} \right>}

\newcommand{\Rat}{\operatorname{Rat}}

\newcommand{\Rec}{\operatorname{Rec}}

\newcommand{\RAT}{\operatorname{Rat}}
\newcommand{\FRAT}[2]{{\operatorname{FRat}}(#1,#2)}
\newcommand{\REC}{\operatorname{Rec}}
\newcommand{\REG}{\operatorname{Reg}}

\newcommand\ra{\longrightarrow}

\def\math#1{\ifmmode #1\else \mbox{$#1$} \fi}

\newcommand{\IFF }{if and only if\xspace}

\newcommand{\Schuetz}{Sch{\"u}tz\-en\-ber\-ger\xspace}

\newcommand*{\height}[0]{\operatorname{height}}

\makeatletter
\newcommand*{\TDFTA}{\@ifnextchar{.}{$\downarrow$FTA}{$\downarrow$FTA\@\xspace}}
\newcommand*{\BUFTA}{\@ifnextchar{.}{$\uparrow$FTA}{$\uparrow$FTA\@\xspace}}
\newcommand*{\TDDFTA}{\@ifnextchar{.}{$\downarrow$DFTA}{$\downarrow$DFTA\@\xspace}}
\newcommand*{\BUDFTA}{\@ifnextchar{.}{$\uparrow$DFTA}{$\uparrow$DFTA\@\xspace}}
\newcommand*{\FTA}{\@ifnextchar{.}{FTA}{FTA\@\xspace}}
\newcommand*{\eg}{\@ifnextchar{.}{e.\,g}{e.\,g.\@\xspace}}
\newcommand*{\ie}{\@ifnextchar{.}{i.\,e}{i.\,e.\@\xspace}}
\newcommand*{\OBdA}{\@ifnextchar{.}{O.\,B.\,d.\,A}{O.\,B.\,d.\,A.\@\xspace}}
\newcommand*{\oBdA}{\@ifnextchar{.}{o.\,B.\,d.\,A}{o.\,B.\,d.\,A.\@\xspace}}
\newcommand*{\usw}{\@ifnextchar{.}{usw}{usw.\@\xspace}}
\renewcommand*{\dh}{\@ifnextchar{.}{d.\,h}{d.\,h.\@\xspace}}
\newcommand*{\zB}{\@ifnextchar{.}{z.\,B}{z.\,B.\@\xspace}}
\newcommand*{\idR}{\@ifnextchar{.}{i.\,d.\,R}{i.\,d.\,R.\@\xspace}}
\newcommand*{\bzw}{\@ifnextchar{.}{bzw}{bzw.\@\xspace}}
\newcommand*{\s}{\@ifnextchar{.}{s}{s.\@\xspace}}
\newcommand*{\su}{\@ifnextchar{.}{s.\,u}{s.\,u.\@\xspace}}
\newcommand*{\iZ}{\@ifnextchar{.}{i.\,Z}{i.\,Z.\@\xspace}}
\newcommand*{\iW}{\@ifnextchar{.}{i.\,W}{i.\,W.\@\xspace}}
\newcommand*{\ua}{\@ifnextchar{.}{u.\,a}{u.\,a.\@\xspace}}
\newcommand*{\iA}{\@ifnextchar{.}{i.\,A}{i.\,A.\@\xspace}}
\makeatother

\newcommand{\nd}{non\-deterministic\xspace}
\newcommand{\Ip}{In particular,\xspace}
\newcommand{\ip}{in particular,\xspace}

\renewcommand{\hom}{homomorphism\xspace}

\renewcommand{\phi}{\varphi}

\newsiamremark{remark}{Remark}

\usepackage{prettyref}
\newcommand{\prref}[1]{\prettyref{#1}}
\newrefformat{thm}{Theorem~\ref{#1}}
\newrefformat{lem}{Lemma~\ref{#1}}
\newrefformat{def}{Definition~\ref{#1}}
\newrefformat{cor}{Corollary~\ref{#1}}
\newrefformat{prob}{Problem~\ref{#1}}
\newrefformat{prop}{Proposition~\ref{#1}}
\newrefformat{sec}{Section~\ref{#1}}
\newrefformat{kap}{Chapter~\ref{#1}}
\newrefformat{ex}{Example~\ref{#1}}
\newrefformat{exs}{Examples~\ref{#1}}
\newrefformat{eq}{Equation~(\ref{#1})}
\newrefformat{rem}{Remark~\ref{#1}}
\newrefformat{fig}{Figure~\ref{#1}}
\newrefformat{par}{Paragraph~\ref{#1}}
\newrefformat{clm}{Claim~\ref{#1}}

\author{Volker Diekert\thanks{Formale Methoden der Informatik, Universit\"at Stuttgart, Germany\\ (\email{diekert@fmi.uni-stuttgart.de}).}
\and Igor Potapov\thanks{Department of Computer Science, Ashton Building, Ashton Street, University of Liverpool, UK (\email{potapov@liverpool.ac.uk}).}
\and Pavel Semukhin\thanks{Department of Computer Science, James Parsons Building, Liverpool John Moores University, UK (\email{p.semukhin@ljmu.ac.uk}).}}

\title{Decidability of membership problems for flat rational subsets of $\GL(2, \Q)$ and singular matrices\thanks{This manuscript substantially extends the following three conference papers~\cite{DiekertPS20,PS_SODA,PS_MFCS2017}.}}

\headers{Decidability of membership problems for flat rational subsets}{V. Diekert, I. Potapov and P. Semukhin}

\begin{document}
\maketitle

\begin{abstract}
We consider membership problems for rational subsets of the semigroup of $2\times 2$ matrices over~$\mathbb{Q}$. For a semigroup $M$, the rational subsets $\Rat(M)$ are defined as the sets accepted by NFAs whose transitions are labeled by elements of $M$. In general, it is undecidable on inputs $m\in M$ and $R\in \Rat(M)$ whether $m$ belongs to $R$.
    Therefore, we restrict our attention to the family $\FRAT{M}{S}$ of flat rational subsets of $M$ over $S$, where $S$ is a subsemigroup of $M$. It consists of finite unions of the form $g_0L_1g_1 \cdots  L_tg_t$, where $L_i\in \Rat(S)$ and $g_i\in M$.
    Assuming that the membership for $\RAT(S)$ is decidable, we prove various results when the membership for $\FRAT{M}{S}$ is decidable.

    If $H$ is a subgroup of a group $G$, then we provide a rather general condition when $\FRAT{G}{H}$ is an (effective) relative Boolean algebra. This leads to one of our main results that 
    the emptiness problem for Boolean combinations of sets in 
    $\FRAT{\mathrm{GL}(2,\mathbb{Q})}{\mathrm{GL}(2,\mathbb{Z})}$ is decidable.
    It is possible that such a strong decidability result cannot be pushed any further for groups sitting between 
    $\mathrm{GL}(2,\mathbb{Z})$ and $\mathrm{GL}(2,\mathbb{Q})$. To support this possibility, we prove the following dichotomy: if $G$ is a finitely generated group such that $\mathrm{GL}(2,\mathbb{Z}) < G < \mathrm{GL}(2,\mathbb{Q})$, then either $G\cong \mathrm{GL}(2,\mathbb{Z})\times \mathbb{Z}^k$ or $G$ contains an extension of the Baumslag-Solitar group $\mathop\mathrm{BS}(1,q)$ of infinite index. It is open whether the membership for rational subsets is decidable in the latter case.
For singular matrices, we will show that 
the membership problem for $\FRAT{\MQ}{S}$ is decidable in doubly exponential time,
where $S$ is the monoid generated by $\mathrm{GL}(2,\mathbb{Z})\cup \set{r\in \Q}{r>1} \cup \os{0,\vdmatrix{1}000}$.
\end{abstract} 



\section{Introduction}\label{sec:intro}
Many computational problems in matrix theory are inherently difficult to solve
even for $2\times 2$ matrices, and most them 
are  undecidable in a higher dimension. One of these problems is the \emph{semigroup membership problem} over some fixed commutative ring $R$:
given a sequence $A, A_{1} \lds  A_{m}$ in $R^{n\times n}$, determine whether $A$  belongs to the semigroup generated by the $A_i$'s. In other words, determine whether there exist an integer $k \ge 1$ and $i_1,\ldots ,i_k \in \{1,\ldots ,m\}$ such that
$A=A_{i_1}  \cdots  A_{i_k}$. Here and in the following 
$R^{n\times n}$ denotes the multiplicative monoid of $n \times n$ matrices with coefficients in $R$, and $\GL(n,R)$ denotes its 
group of units which consists of the matrices that are invertible in $R^{n\times n}$. We also use $\SL(n,R)$ to denote the subgroup of $\GL(n,R)$ of matrices with determinant~one.
The semigroup membership problem has been intensively studied since 1947 when Markov showed in~\cite{Markov47} that this problem is undecidable for matrices in~$\mathbb{Z}^{6 \times 6}$.
A special case is the \emph{mortality problem} where the target matrix $A$ is the 
the zero matrix. The mortality problem is undecidable for $\Z^{3 \times 3}$ by Paterson~\cite{Paterson70}. For  $\Z^{2 \times 2}$
it is unknown whether the mortality problem is decidable. 
If $A$ and the $A_i$'s are in $\GL(n,R)$, then the \emph{subgroup membership problem} asks whether $A$  belongs to the matrix semigroup which is generated by the $A_i$'s and $\oi{A}_i$'s. The subgroup membership problem is undecidable for $\GL(4,\Z)$ by Mihailova~\cite{Mihailova58}. It is unknown whether the subgroup membership is decidable for $\GL(3,\Z)$.
Even significantly restricted cases of these membership problems turn out to be undecidable for high dimensional matrices over the integers~\cite{BHHKP08,KoeLoZe16},
and very few cases are known to be decidable, see~\cite{BabaiBCIL96,BPS_MFCS2019,Charlier_Honkala_2014_Inf_and_Comp_237}. The decidability of many of these problems remains open even for $2 \times 2$ matrices over integers~\cite{CHHN14,ColcombetOSW19,Harju09,KNPicalp2018,PotapovDagstuhl19}.

A natural and important generalization of the semigroup membership problem is the \emph{membership problem for rational subsets} of a semigroup $M$: given an element $a\in M$ and a rational subset $L \sse M$, decide whether $a$ belongs to $L$. The family of \emph{rational subsets} of $M$ is denoted by $\Rat(M)$, and it has various equivalent definitions: 
homomorphic images of regular subsets of \fg free semigroups, 
regular expressions over $M$, or acceptance by $M$-NFAs. An 
$M$-NFA is a  non-deterministic finite automaton $\cA$ whose \tras are labeled by elements in $M$. The label of a directed path is the directed product over its labels, and the accepted language is the set $L(\cA)\sse M$ of 
labels of directed paths from initial to final states. 
Using $M$-NFAs allows for a graphical representation and is typically a more concise notation than using regular expressions. Thus, 
$M$-NFAs are our preferred way of defining sets in $\Rat(M)$.

It is well-known that the group $\mathrm{SL}(2,\Z)$ has a free subgroup of rank $2$ of index $12$ by~\cite{Newman62}. Hence, both
$\mathrm{GL}(2,\Z)$ and $\mathrm{SL}(2,\Z)$ are finitely generated virtually free groups, and
the families of their rational subsets form  effective Boolean algebras~\cite{sen96actainf,Silva02}.
\Ip the membership problem for rational subsets in 
$\mathrm{GL}(2,\mathbb{Z})$ and in $\mathrm{SL}(2,\mathbb{Z})$ is decidable. This is no longer the case in higher dimensions. For example, 
in dimension four, $\RAT(\mathrm{SL}(4,\Z))$ is not even closed under finite intersections, and therefore it is not a Boolean algebra. However, 
this is still open for $\mathrm{SL}(3,\Z)$, see \prref{rem:sl34z}.
 
Two previous results that extended the decidability of the semigroup membership problem beyond $\mathrm{GL}(2,\Z)$ are~\cite{PS_SODA,PS_MFCS2017}. 
The present paper pushes the frontier of decidability even further. 
First of all, we consider membership problems for $2\times 2$ matrices over the rationals, whereas~\cite{PS_SODA,PS_MFCS2017} only dealt with integer matrices. Since the rational subset membership problem is known to be decidable for $\mathrm{GL}(2,\Z)$, we focus on finitely generated  sub\-groups~$G$ of $\mathrm{GL}(2,\Q)$ which contain $\mathrm{GL}(2,\Z)$. Also, in contrast to~\cite{PS_SODA,PS_MFCS2017}, we give concrete complexity bounds: all complexities are in deterministic doubly exponential time (or better) for a natural binary encoding of the inputs.

In order to provide an essentially self-contained exposition of the main results, we combine a number of auxiliary results in Sections~\ref{sec:pre}, \ref{sec:tfl}, and~\ref{sec:SNFC}. In \prref{sec:pre}, we characterize recognizable and rational sets in semigroups and highlight essential properties of (relative) Boolean algebras. In \prref{sec:tfl}, we show so-called \emph{Fatou property} for groups. It states that if $G$ is a group and $H$ is its subgroup, then $L\sse H$ and $L\in \RAT(G)$ implies that $L\in \RAT(H)$ (see \prref{thm:silva} and \prref{cor:silvana}).
We also provide techniques for transferring results for rational subsets in group extensions of finite index (\prref{cor:finext}). In \prref{sec:SNFC}, we describe a cubic procedure for computing a \emph{\SNF} of a non-zero matrix in $\MQ$, and we discuss properties of \emph{commensurators}, a notion borrowed from geometric group theory.

In \prref{sec:algglq}, we prove our first main result which is \prref{thm:nofinext}.
It states a dichotomy for a finitely generated (\emph{\fg} for short) subgroup $G$ sitting strictly between $\mathrm{GL}(2,\Z)$ and $\mathrm{GL}(2,\Q)$.
In the first case of the dichotomy, $G$ is generated by 
$\mathrm{GL}(2,\Z)$ and \emph{finitely many} nonsingular central matrices $\vdmatrix r 00 {r}$. In this case, $G$ 
is isomorphic to $\mathrm{GL}(2,\Z)\times \Z^k$ for $k\geq 1$, in which the membership problem for rational subsets is known to be decidable.

This is the best we can hope for groups sitting strictly between $\GL(2,\Z)$ and $\GL(2,\Q)$ in the general case. 
Indeed, our dichotomy states that if such a \fg~group $G$ is not isomorphic to $\GL(2,\Z)\times \Z^k$, then $G$ contains an extension of infinite index of a Baumslag-Solitar group $\mathop\mathrm{BS}(1,q)$ for some $q\geq 2$.
The Baumslag-Solitar groups $\mathop\mathrm{BS}(p,q)$ are defined by two generators $a$ and $t$ with the  defining relation
$t a^p \oi t = a^q$.  They were introduced in~\cite{baumslag62some} and have been widely studied since then.
As we see in the proof of \prref{thm:nofinext}, $\mathop\mathrm{BS}(1,q)$ cannot appear as a subgroup in $\GL(2,\Z)\times \Z^k$, which implies that the two cases of the dichotomy are mutually exclusive.
The group $\mathop\mathrm{BS}(1,q)$ is metabelian, and subgroup membership is decidable for \fg~metabelian groups by~\cite{Romanovskii74}.
Actually, a stronger result is known for Baumslag-Solitar groups: 
the membership problem for rational subsets of $\mathop\mathrm{BS}(1,q)$ is decidable for all $q\geq 2$
by Cadilhac, Chistikov, and Zetzsche~\cite{CadilhacCZ2020ICALP}.
However, it is not clear how to generalize this result to extensions of $\BS(1,q)$ of infinite index.

Motivated by the above results and observations, we introduce in \prref{sec:FRAT} the notion of \emph{flat rational sets} of a semigroup $M$ over its subsemigroup $S$. We denote them by $\FRAT M S$. 
In the terminology of \Schuetz \cite{sch76}, $\FRAT M S$ is the \emph{polynomial closure} of $\RAT(S)$ in~$M$. More precisely, 
a subset $L\sse M$ is a flat rational set \IFF it can be written as a finite union of languages $L_0m_1L_1\cdots m_kL_k$, where 
the $m_i$'s belong to $M$ and $L_i\in \Rat(S)$ for $1\leq i \leq k$. 

We are mainly interested in the study of $\FRAT{\mathrm{GL}(2,\mathbb{Q})}{S}$. Since $\mathrm{GL}(2,\mathbb{Q})$ is not finitely generated, the family $\FRAT{\mathrm{GL}(2,\mathbb{Q})}{S}$ is never a Boolean algebra because $\mathrm{GL}(2,\mathbb{Q})\notin \FRAT{\mathrm{GL}(2,\mathbb{Q})}{S}$. 
One of our main results about flat rational sets (\prref{thm:lea}) shows that under some natural assumptions on
a group $G$ and its subgroup $H$, the family $\FRAT G H$ forms an effective relative Boolean algebra (see \prref{def:releba}). As an application of this result, we will show that we can decide the emptiness of finite Boolean combinations of flat rational sets of $\GL(2,\Q)$ over $\GL(2,\Z)$ (\prref{cor:leas}).
In \prref{thm:genfratmon}, we provide an alternative intrinsic description of flat rational sets. 
In the rest of \prref{sec:FRAT}, we show a reduction of the membership problem for $\FRAT{M}{G}$ to that of $\FRAT{M}{H}$, where $M$ is a monoid, $G$ is a subgroup of its group of units, and $H$ is a finite index subgroup of $G$ (\prref{thm:cAm} and \prref{cor:Pred}). 

In the remaining three sections, we prove new decidability results for flat rational sets that contain matrices from $\MQ$. In \prref{sec:membfrat}, we show that the membership problem for flat rational sets of $\mathrm{GL}(2,\mathbb{Q})$ over $\mathrm{GL}(2,\mathbb{Z})$ is decidable in exponential time (\prref{thm:red}). We then prove various generalizations of this result, although with a worse complexity bound. 
For example, we show that the membership problem for $\FRAT{\mathrm{GL}(2,\mathbb{Q})}{S}$ is decidable in doubly exponential time, where $S=\mathrm{GL}(2,\mathbb{Z})\cup \set{g \in \GL(2,\Q)}{\,|\det (g)|> 1}$ (\prref{thm:nsmflat}).

If the target is a non-zero singular matrix, then 
we show in \prref{sec:ginger} that the membership problem for $\FRAT{\MQ}{S}$ is decidable in doubly exponential time
for the monoid~$S$ which is generated by
$\mathrm{GL}(2,\mathbb{Z})\cup \set{r \in \Q}{r> 1} \cup \os{\vdmatrix 1000}$ (\prref{thm:sing}).
However, we  prove a better complexity bound for the \emph{mortality problem}.
Namely, we show that mortality for $\FRAT{\MQ}{S}$ is decidable in exponential time
for the monoid~$S$ which is generated by $\mathrm{GL}(2,\mathbb{Z})\cup \Q \cup \os{\vdmatrix 1000}$ (\prref{thm:mort}).
In \prref{sec:conclusion}, we discuss potential directions for future research and list several open problems in this field.

\section{Notation and preliminaries}\label{sec:pre}
An \emph{\invol} of a set $S$ is a mapping $x\mapsto \ov x$ such that 
$\ov{\ov x}=x$ for all $x\in S$. 
An {\invol} of a semigroup $(S,\cdot)$ is an \invol~$\ovv$ of $S$ such that $\ov {x\cdot y} =  \ov {y}\cdot \ov {x}$. 
A monoid $M$ is a semigroup $(M,\cdot)$ with a neutral element~$1$. 
Typically, we write commutative monoids like $\N$, $\Z$, $\Z/n\Z$ or~$\Q$  with an additive notation. If we use a multiplicative notation, then~$1$ denotes the neutral element of a monoid. \Ip the empty word in free monoids is denoted by $1$. It is also custom to write $xy$ instead of $x\cdot y$. 
A \emph{zero} in a semigroup $(M,\cdot)$ is an element $0$ such that $x\cdot 0 =0\cdot x = 0$ for all $x\in M$. 
If $(M,\cdot)$ is a semigroup with \invol and with a zero $0$, then $\ov 0=0$. If $(M,\cdot)$ is a monoid with \invol, then  $\ov 1=1$.
If $\Gam$ is a set with an \invol~$\ovv$, then $\Gam^+$ 
(resp.,~$\Gam^*$) is a free semigroup (resp., free monoid) with \invol, where the \invol is defined on by $\Gam^*$ by extending it from~$\Gam$ by using
the law $\ov{uv} = \ov{v}\; \ov u$.

If $G$ is a group, then it is a monoid with an \invol~$\ovv$ defined by $\ov g=\oi g$ for all $g\in G$.  The identity mapping is an \invol for commutative semigroups.

In commutative monoids without a zero-element, we might use an additive operation~$+$, and then the neutral element is denoted as~$0$. There will be no risk of confusion.

For a subset $L\sse M$ of a semigroup $M$, the set $L^+$ denotes the subsemigroup of $M$ generated by $L$. If $M$ is a monoid, then 
the submonoid generated by $L$ is $L^*=L^+\cup \os 1$. It is called the \emph{Kleene-star} of $L$. We also use ``\fg''~as an abbreviation for ``finitely generated''. Hence,  
a semigroup (resp., monoid) is \fg if it is a homomorphic image of a \fg
free semigroup (resp., monoid).

The group of \emph{units} of $M$ is the submonoid of invertible elements, denoted henceforth by $U(M)$. It is the set consisting of all $x\in M$ such that there is some $\ov x\in M$ with $\ov x x = x\ov x=1$. If $x\in U(M)$, then we also write 
$\oi x$ instead of~$\ov x$. 
If $x$ is a unit of $M$, then 
$x^\Z$ denotes the set $\set{x^n}{n\in \Z}$, which is the subgroup generated by~$x$. By $Z(M)$ we denote the \emph{center} of $M$, that is, the set of elements which commute with all elements in $M$. 
 We write
 $S\leq M$ if~$S$ is a subsemigroup  of $M$, and $S< M$ if $S\leq M$ but $S \neq M$.

A subsemigroup $I$ of a monoid $M$ is an \emph{ideal} if $M\, I \, M\sse I$. The empty set~$\es$ is an ideal. If $M$ contains a zero $0$, then $\os{0}$ is the least nonempty ideal. If an ideal $I$ contains an element of $U(M)$, then $I=M$. Thus, if $I\neq M$, then $I$ is contained in $M\sm U(M)$. In general, $M\sm U(M)$ is not an ideal (see an example in \prref{rem:erata}).

A group $G$ is called \emph{virtually free} if it contains
a free group of finite index. 
A group is finitely generated as a group \IFF it is finitely generated as a semigroup.

By  $\MnR$ we denote the ring of $n\times n$ matrices over a commutative ring~$R$,  
and we let \mbox{$\det: \MnR \to R$} be the determinant function. The units of $R$ are denoted by $R^*$. 
We view $R$ as a subring of $\MnR$ by identifying $r\in R$ with the matrix
$r = rI_n$, where $I_n$ is the $n$-dimensional identity matrix. Hence, we may write $1= I_n$ and~$-1=-I_n$.
By $\GL(n,R)$ we mean the group of invertible
matrices, that is, the matrices $g\in \MnR$ such that~$\det(g)\in R^*$ is a unit. For $n\geq 2$, the center of $\GL(n,R)$ is $R^*=\set{rI_n}{r\in R^*}$.

When we consider a matrix ring $\MnR$, a semigroup in $\MnR$ refers to a subsemigroup in the multiplicative 
monoid $(\MnR,\cdot)$. 

By $\SL(n,R)$ we denote the \emph{special linear group} $\oi{\det}(1) = \set{g\in \GL(n,R)}{\det(g)=1}$. It is a normal subgroup of $\GL(n,R)$. 
The structure of $\SL(2,\Z)$ is well-understood.\footnote{A discussion  about $\SL(2,\Z)$ including the computation of normal forms is, for example, in~\cite[Sec.~8.12]{edam16}.}
The groups $\SL(2,\Z)$ and  $\GL(2,\Z)$
are \fg~virtually free groups.
\begin{remark}\label{rem:glzvf}
It is shown in 
\cite{Newman62} that the \emph{projective linear group} 
$\PSL(2,\Z)=\SL(2,\Z)/\os{\pm 1}$
has a free subgroup of rank~$2$ and of index~$6$. Hence, $\SL(2,\Z)$ has a free subgroup of rank~$2$ of index~$12$. Therefore, 
$\GL(2,\Z)$ has a free subgroup of rank~$2$ which has index~$24$. Actually, 
the free subgroup of index~$24$ and rank~$2$ can be chosen to be the commutator subgroup of~$\SL(2,\Z)$.  
Possible generators for $\SL(2,\Z)$ are the matrices $R=\vdmatrix{0}{-1}{1}{1}$ of order~$6$ and $S=\vdmatrix{0}{1}{-1}{0}$ of order~$4$. Possible generators for $\GL(2,\Z)$ are $S$, $R$, and $\vdmatrix{1}{0}{0}{-1}$.
\Ip since $\SL(2,\Z)$ and $\GL(2,\Z)$ are generated by elements of finite order, none of the virtually free groups $\PSL(2,\Z)$, $\SL(2,\Z)$, or $\GL(2,\Z)$ is free. \hspace*{\fill}$\diamond$\end{remark}
\subsection{Recognizable and rational sets in semigroups}\label{sec:ratrec}
Throughout this subsection $M=(M,\cdot)$ denotes a semigroup.\footnote{We call it $M$ because in most of our cases the semigroup $M$ is a monoid.}
We recall some classical facts as they can be found with their proofs in the classical textbook of Eilenberg~\cite{eil74} or in the \emph{$q$-Book} of Rhodes and Steinberg~\cite{rs09qtheory} as well as in \cite{edam16}.   
  
\begin{definition}\label{def:recM}
A subset $L\sse M$ belongs to the family of \emph{recognizable sets} $\REC(M)$ if
there exists a \hom $\phi:M\to N$ of $M$ to a finite semigroup $N$ such that $L=\oi \phi(\phi(L))$. We also say that $\phi$ (resp.,~$N$) 
\emph{recognizes} $L$. 
\end{definition}
Note that the canonical \hom of $M$ to the trivial monoid 
$\os 1$ recognizes~$\es$ and $M$. \begin{proposition}\label{prop:recbool}
If $\phi_i:M\to N_i$ recognizes subsets $L_i\sse M$ for $i=1,2$, then 
the \hom $M\to N_1\times N_2,\; m\mapsto \big(\phi_1(m), \phi_2(m)\big)$ recognizes $L_1\cap L_2$ and $M\sm L_i$ for $i=1,2$. \Ip 
$\Rec(M)$ is a Boolean algebra (in the sense of \prref{def:releba} below).
\end{proposition}
\begin{definition}\label{def:ratM}
The family of \emph{rational sets} $\RAT(M)$ has the following definition
using \emph{regular}  (aka \emph{rational}) expressions.
It is the least family such that:
\begin{enumerate}
\item $|L| < \infty, \,  L \sse M \implies L \in \RAT(M)$.
\item $L_1,\, L_2 \in \RAT(M) \implies L_1 \cup L_2,\; L_1 \cdot L_2,\, \text{and } L_1^+ \in \RAT(M).$
\end{enumerate}
\end{definition}
Note that the definition of $\RAT(M)$ is intrinsic without reference to any generating set. Moreover, let $M$ be a monoid and $L\in \RAT(M)$, then $L^*\in \RAT(M)\iff L^+\in \RAT(M)$ because $L^*=L^+\cup \os 1$ and $L^+=L\cdot L^*$.
\begin{remark}\label{rem:recinG}
Let $G$ be a  group. Then $L\sse G$ is recognizable \IFF
there is normal subgroup $N$ of finite index and a finite subset 
$\os{g_1\lds g_k}\sse G$ such that $L=\bigcup\set{g_iN}{1\leq i \leq k}$. \Ip if $G$ is infinite, then no finite subset of $G$ is recognizable. A subgroup $H$ belongs to $\RAT(G)$  \IFF $H$ is \fg by~\cite{AnisimovS75}. 
This does not hold for submonoids:
the standard example is the additive group $\Z\times \Z$. It contains the submonoid 
$\set{(m,n)\in \N\times \N}{m=0\vee n\geq 1} = \{(0,0)\} \cup \big((0,1) + \N\times \N\big)$ which is rational but not finitely generated, see \cite[Sec.~8.9]{edam16}. \hspace*{\fill}$\diamond$\end{remark}
\begin{proposition}\label{prop:freekle}
Let $h:M \to M'$ be any  \hom of semigroups. Then the following assertions hold. 
\begin{itemize}
\item If $L'\in \REC(M')$, then $\oi h(L')\in \REC(M)$. 
\item If $L\in \RAT(M)$, then $h(L)\in \RAT(M')$. 
\item If $L \in \REC(M)$ and $K \in \RAT(M)$, then $L \cap K \in \RAT(M)$.
\item \emph{Kleene's Theorem,~\cite{kle56}:} If $M$ is either a \fg~free monoid or a \fg~free semigroup, then $\REC(M) = \RAT(M)$.
\item Let $L\in \Rat(M)$, then $L$ is contained in a \fg subsemigroup of $M$. \Ip $M\in \Rat(M)$ implies that $M$ is finitely generated.\footnote{McKnight's Theorem~\cite{mck64} is slightly more general. It  states that $M$ is finitely generated \IFF $M\in \RAT(M)$ \IFF  $\REC(M) \sse \RAT(M)$. 
}
\item Let $H$ be a subgroup of a group $G$. Then $H\in \Rec(G)$ \IFF 
the index $[G:H]$ is finite, see~\cite{AnisimovS75}.
\end{itemize}
\end{proposition}
For a \fg~free semigroup (or a \fg~free monoid) $M$, the family of \emph{regular languages} $\REG(M)$ is defined as $\REG(M) = \RAT(M) = \REC(M)$, where the last equality holds thanks to Kleene's Theorem stated in \prref{prop:freekle}
above. Henceforth, if we use the term ``regular language'', then we always refer to a rational subset in some finitely generated free semigroup or monoid.  For other monoids, we frequently have $\REC(M) \neq \RAT(M)$. 
This happens, for example, if $M$ is an infinite group.
Moreover, if $M$
contains a direct product $\os{a,c}^* \times \os{b}^*$, then, in contrast to 
$\REC(M)$, the family   
$\RAT(M)$ is not closed under finite intersection, see for example~\cite[Ex.~7.21]{edam16}.
\begin{definition}\label{def:NFAmon}
Let $M$ be a semigroup and $T\sse M$.
A \emph{nondeterministic finite automaton over $T$} (or a \emph{$T$-NFA} for short) is a tuple $\cA=(Q,\del,I,F)$, where~$Q$ is a set of \emph{states} with subsets $I,F\sse Q$ and 
$\del\sse Q\times T  \times Q$ is a finite set of \emph{\tra{s}}. The set $I$ (resp.,~$F$) is called the set of \emph{initial} (resp.,~\emph{final}) states. 
A \tra $(p,s,q)\in \del$ is also written as $p \arc s q$; and we say that 
$s\in T$ is its \emph{label}.
If $M$ has a neutral element $1$, then an \emph{$\eps$-\tra}\footnote{The notation $\eps$ is used because it frequently appears in the literature on NFAs, where $\eps$ denotes the empty word.} is a \tra with label $1\in M$.

A path (of \tras) of length $n\in \N$ is a sequence 
$q_0,a_1,q_1, \ldots,{a_n},q_n$ such that
$(q_{i-1},a_i,q_i)\in \del$ for all $1\leq 1 \leq n$.
Paths may also be depicted as 
\begin{align}\label{eq:trapath}
q_0 \arc {a_1} q_1 \quad \cdots  \phantom{q_2}\arc {a_{n-1}} q_{n-1}   \arc {a_n} q_n.
\end{align}
If $q_0\in I$ and $q_n\in F$, then we say that the path
is \emph{accepting} for the element $a_1 \cdots a_n\in M$.
The \emph{accepted language $L(\cA)$} is the set of all
$m\in M$ for which there is a factorization $m=a_1 \cdots a_n$ that has an accepting path of length $n$.

A \emph{subautomaton} of $\cA$ is an NFA $\cA'=(Q',\del',I',F')$ such that $Q'\sse Q$ and $\del'\sse \del$, but there are no restrictions how to choose $I'$ or $F'$. 
\end{definition}
We follow the convention that if a path of length zero accepts $m\in M$, then $M$ is a monoid and $m$ is its neutral element.
An NFA $\cA$ is called 
\emph{trim}, if every
state belongs to some accepting path.
Whenever convenient, we assume that $\cA$ is trim. Note that $L(\cA)$ is contained in the subsemigroup of $M$ which is generated by the finite set of labels of the \tras of $\cA$. This holds whether or not $\cA$ is trim. \begin{proposition}\label{prop:ratnfA}
Let $L\sse M$ be any subset. Then the following assertions are equivalent. 
\begin{itemize}
\item The set $L$ belongs to $\RAT(M)$. \item There is some $M$-NFA $\cA$ such that 
$L= L(\cA)$.
\item The set $L$ is the image $\phi(K)$ of a regular set $K\sse\Sig^+$
under some \hom $\phi: \Sig^+ \to M$.
\end{itemize}
\end{proposition}

The following lemma is used in the proof of \prref{thm:silva} below. Its proof is straightforward. 
\begin{lemma}\label{lem:silva}
Let $\cA$ be an $M$-NFA and $q_0 \arc {a_1} q_1 \quad \cdots  \phantom{q_2}\arc {a_{n-1}} q_{n-1} \arc {a_n} q_n$ denote a path as in (\ref{eq:trapath}) with $n\geq 2$. Then adding (or removing) a \tra $q_{0} \arc {m} q_n$ with label
$m=a_1\cdots a_n$ does not change the accepted language.
\end{lemma}
Note that adding \tras possibly makes accepting paths shorter, whereas removing \tras makes the size of the NFA smaller.

\subsection{The input size of matrices and NFAs over matrices}
\label{sec:inputsize}
We use the following notation.
We let $\log (x)= \max\os{1, \log_2(x)}$. Let $f,g:\N \to \R_{\geq 0}$ be two functions with values in non-negative real numbers.
As usual, we let $f\in {\Oh}(g)$ if there is some $k\in \N$ such that $f(n)\leq k g(n) +k$ for all $n\in \N$.
Sometimes we measure complexities in \emph{soft $\Oh$-notation} $\wt {\Oh}$.
 We write 
$f\in \wt {\Oh}(g)$ if $f\in \Oh(g\cdot \log^k (g))$ for some $k\in \N$. Thus, in soft $\Oh$-notation poly-logarithmic factors are neglected. 

The (bit-)complexity of an algorithm depends on the bit encoding 
of the input. When talking about complexity,
we usually work with NFAs where the labels of \tra{s} are $n\times n$ matrices over~$\Q$, and therefore we define their size.
There are two natural encodings: unary and binary. We will use both of them. For a matrix $A=(a_{ij})$ with integer entries $a_{ij}\in \Z$, we let 
$\Abm {A}= \max\set{\abs{a_{ij}}}{1\leq 1,j\leq n}$. Given ${m}\in \MnQ$, we assume that ${m}$ is written as ${m}= \oi pA$ where~$p$ is the least positive integer such that $pm =A \in \MnZ$. That is,~$p$ is $1$ for the zero matrix, and otherwise~$p$ is the $\lcm$ of the denominators of non-zero entries in $m$. For such a representation ${m}= \oi pA$ as above
we define its \emph{unary size} as 
$\Abm {m}= p \Abm {A}$. It  
does not yield a matrix norm, however, for $n=2$, we have:
\begin{align}\label{eq:Abm}
\Abm {{m}_1\cdots {m}_\ell}\leq 2^{\ell-1} \prod_{i=1}^\ell \Abm {{m}_i}.
\end{align}

Since we are (mainly) interested in the bit complexity, we define 
\emph{binary size} of $m$ as
$\Abbin {{m}} = \log(\Abm {m})$.
Hence, for $a,b,c,d\in \Z$ we have
$\Abbin {\vdmatrix abcd}=\log_2(\max\os{2,\abs a,\abs b,\abs c,\abs d})$. \Ip $\Abbin {\vdmatrix 0000} =\Abbin {\vdmatrix 1001}= \Abbin {\vdmatrix 2222}= 1$.
\begin{lemma}\label{lem:bitp}
Let ${m}={m}_1\cdots {m}_\ell$ be a product of $\ell$ matrices in $\MQ$
such that $\Abm{{m}_i} \leq 2^{k}$ for all $1\leq i\leq \ell$. 
Then we have $\Abbin {{m}} \in \Oh(k\ell)$.
\end{lemma}
\begin{proof}
 This is a direct consequence of the inequality in (\ref{eq:Abm}).
\end{proof}
\begin{definition}\label{def:weightmatrix}
Let $\cA=(Q,\del,I,F)$ be an $\MQ$-NFA, that is, all labels of \tras of $\cA$ are matrices 
in $\MQ$. The \emph{binary and unary sizes}  $\Abbin{\cA}$ and
$\Abm{\cA}$ of the NFA~$\cA$
are  defined as follows: \begin{align}
\label{eq:wnfab}
\Abbin{\cA}&=1 + \abs Q + \abs{\del}+\sum_{(p,m,q)\in \del}\Abbin m,\\
\label{eq:wnmax}
\Abm{\cA}&=1 + \abs Q + \abs{\del}+\sum_{(p,m,q)\in \del}\Abm m.
\end{align}
\end{definition}

\subsection{Reductions and complexity classes}\label{sec:NFAcc}
We follow standard notation in complexity theory as it can be found for example in~\cite{pap94}. \Ip we assume the reader is familiar
with the classes $\P$ and $\NP$ which denote the families of decision problems decidable on a Turing machine in deterministic (resp.,~\nd) polynomial time.

Decision problems are encoded as subsets of $\Del^*$ where $\Del$ is a finite alphabet, for example, $\Del=\os{0,1}$.
We define complexity classes via the notion of reductions,
which are realized by (nondeterministic) Turing machines in the following sense.
Let $\Gam$ and $\Sig$ be finite alphabets, and $f:\N \to \N$ be any function.
Let $T_f$ be a Turing machine with input alphabet $\Gam$
and a separate write-only output-tape, which is initially empty.
We assume that $T_f$ also satisfies the following property: on any input $u\in  \Gam^*$ of length $n$,
every computation of $T_f$ stops after at most $f(n)+1$
 steps
and produces some $v\in \Sig^*$ on the output-tape. We write $v\in T_f(u)$ in this case.
Note that by assumption we have $\abs v\in \Oh(f(n))$.

Let $\cP\sse \Gam^*$ and~$\cQ\sse \Sig^*$ be subsets.
We say that $\cP$ is \emph{$\DTIME(f)$ (resp., $\NTIME(f)$) reducible} to $\cQ$
if there exists a deterministic (resp.,~nondeterministic) Turing machine $T_f$
with the above-mentioned properties such that
\(\forall u\in \Gam^*: (u\in \cP \iff \exists v\in \cQ: v\in T_f(u)).\)
As a consequence, if $\cP$ is $\DTIME(f)$ (resp.,~$\NTIME(f)$) reducible to $\cQ$ and if $\cQ$ is $\DTIME(g)$ (resp.,~$\NTIME(g)$) reducible to $\cR$, then 
$\cP$ is $\DTIME(g\circ f)$ (resp.,~$\NTIME(g\circ f)$) reducible to $\cR$.

A \emph{$\PTIME$-reduction} (resp.,~\emph{$\NP$-reduction}) is a  
 $\DTIME(f)$ (resp.,~$\NTIME(f)$) reduction, where~$f$ is some polynomial. 
A \emph{$\DEXPTIME$-reduction}
 (resp.,~\emph{$\NEXPTIME$-reduction}) is a  
 $\DTIME(f)$ (resp.,~$\NTIME(f)$) reduction, where~$f$ is a function of type $2^p$ where~$p$ is some polynomial. 
 
If a problem $\cP$ is $\DTIME(f)$ (resp.,~$\NTIME(f)$) reducible to a singleton like 
$\os{1}$, then we say that  $\cP$ belongs to the complexity class $\DTIME(f)$ (resp.,~$\NTIME(f)$).
The classes $\P=\DTIME(n^{\Oh(1)})$ and $\NP=\NTIME(n^{\Oh(1)})$ are closed under  $\P$ (resp.,~$\NP$)-reductions.

\subsection{Boolean algebras and relative Boolean algebras}\label{sec:rba}
\begin{definition}\label{def:releba}
Let $U$ be any set and~$\cB$ be a family of subsets of $U$.
\begin{itemize}
\item We say that~$\cB$ is a \emph{Boolean algebra}, 
if~$\cB$ is closed under finite union and complement. 
\item We say that~$\cB$ is a 
\emph{relative Boolean algebra} if~$\cB$ is closed under finite union and relative complement: $L, \, K \in \cB \implies 
L\sm K \in \cB$. 
\item We say that~$\cB$ is an 
\emph{effective relative Boolean algebra} if each $K\in \cB$ has an effective finite description, and there is an algorithm that, given descriptions of  
$K, L \in \cB$, computes descriptions of $L\cup K$ and $L\sm K$ and decides the emptiness of~$K$.\end{itemize}
\end{definition}
Every Boolean algebra is a relative Boolean algebra, and every 
relative Boolean algebra contains the empty set~$\es$. 
A relative Boolean algebra is closed under \emph{nonempty} finite intersection. Indeed, $L\cap K= S\sm ((S\sm L) \cup (S\sm K))$ where 
$S=L\cup K$. A relative Boolean algebra~$\cB\sse 2^{U}$ is a Boolean algebra \IFF $U\in \cB$. 
\begin{examples}\label{exs:class}
Let us list some classical examples of (relative) Boolean algebras.
\begin{enumerate}
\item If $M$ is any \fg~semigroup, then the family of recognizable sets $\REC(M)$ is an effective Boolean algebra. \Ip
$\REG(\Sig^*)$ is an effective Boolean algebra if $\Sig$ is finite.
\item Let $M_1$ and $M_2$ be \fg semigroups and let $M=M_1\ast M_2$ denote their free product. 
If $\RAT(M_i)$ is an (effective) Boolean algebra for $i=1,2$,
then $\RAT(M)$ is an (effective) Boolean algebra, see
\cite{Sak92} and also \cite{LohSen06} for a generalization.
\item\label{ratinfgen} Let $M$ be a commutative semigroup. Rational sets in $M$ are also called \emph{semi-linear}: a semi-linear set is a finite union of \emph{linear} sets, and a linear set (in additive notation)
is a set of the form $c + \N d_1 + \cdots + \N d_t$, where $c,d_1,\ldots,d_t\in M$.
The family of semi-linear sets
forms a relative Boolean algebra by
\cite{es69}. Using Presburger arithmetic, it can be shown that
	$\RAT(\Z^k)$  and $\RAT(\N^k)$ are actually effective Boolean algebras for all $k\in \N$. The decidability of Presburger arithmetic is a classical result due to Moj{\.{z}}esz Presburger~\cite{pres29}.
\item Let~$\Q$ be the additive group of the rational numbers. 
Then~$\Q$ is not \fg~and every \fg~subgroup is isomorphic to 
$\Z$. As a consequence, $\RAT(\Q)$ is an effective relative Boolean algebra, but not a Boolean algebra.
\item If $G$ is a \fg~virtually free group, then the family of rational sets 
$\RAT(G)$ is an effective Boolean algebra. If $G$ is an infinitely generated free group, then $\RAT(G)$ is a relative Boolean algebra, but not a 
Boolean algebra. The special case of \fg~free groups is due to 
Benois~\cite{ben69}. The extension to \fg virtually free groups is in~\cite{Grunschlag1999,sen96actainf,Silva02}.
\end{enumerate}
\end{examples}

\section{The Fatou property and transfer results for rational subsets in groups}\label{sec:tfl}
Let $G$ be a group and $H$ be a subgroup. The aim of \prref{sec:tfl} is to prove \prref{thm:silva}. It states that $L\sse H$ and $L\in \RAT(G)$ implies $L\in \RAT(H)$. This is called the Fatou property (see \prref{rem:sl34z} below for further discussion). This property does not hold for groups with respect to submonoids in general. Indeed, according to \prref{rem:recinG} the \fg group $\Z\times \Z$ contains a rational submonoid~$M$ which is not \fg Hence, we have $M\notin \Rat(M)$.

\prref{thm:silva} holds without any assumption about $G$ and its subgroup $H$: for example, the cardinalities the groups~$G$,~$H$, and the 
set of cosets $G/H$ can be arbitrarily high. However, for effectiveness we need some restrictions. Therefore, we introduce the notion of \emph{\enu} in \prref{def:enu}. It is similar to the notion of ``computably \enu'' as used, for example, in  
\cite[Def.~1.1]{tran2018} for Boolean algebras, but differs from it in the sense that we do not require the equality relation to be computably enumerable. \Ip there are groups with an \enu in which it is undecidable whether a given group element represents the neutral element.

We begin with \prref{prop:GLSL}. It gives a quasi-linear time complexity for deciding whether $L(\cA)\sse \SLZ$ when $\cA$ is a $\mathrm{GL}(2,\mathbb{Z})$-NFA. Its proof also serves as a warm-up example for the general proof strategy used later.
\begin{proposition}\label{prop:GLSL}
Let $\cA=(Q,\del,I,F)$ be a $\mathrm{GL}(2,\mathbb{Z})$-NFA of size $\Abbin{\cA}=n$. Then we can construct in soft linear time with respect to $n$ a $\mathrm{GL}(2,\mathbb{Z})$-NFA $\cA'$ such that:
\begin{itemize}
	\item $L(\cA')=L(\cA)$.
	\item $\Abbin{\cA'} \leq \Abbin{\cA}$ and $\cA'$ has at most~$|Q|$ states.
    \item Moreover, $L(\cA')\sse \SLZ$ \IFF all labels 
of \tras in $\cA'$ have determinant~$1$. \Ip we can decide in time
$\wt\Oh(n)$ whether $L(\cA)\sse \SLZ$.
\end{itemize}
\end{proposition}
\begin{proof}
In the first phase we trim  $\cA$, which can be done by standard algorithms in time $\wt\Oh(n)$ because $|Q|+|\del|<n$. Henceforth, we assume without restriction that $\cA$ is trim. 
In the second phase we mark all states in~$Q$ 
either by $+1$ or by~$-1$. The corresponding states are called positive and  negative respectively. All initial states are marked 
with $+1$, hence they are positive.
As long as there is a \tra $p\arc m q$, where~$p$ is marked
and~$q$ is not marked, we mark~$q$ the same way as~$p$ if  
$\det(m)=1$ and with the opposite marking of~$p$ if~$\det(m)=-1$.
After at most $|\del|$ steps all states are marked.
The marking procedure can also be implemented in time $\wt\Oh(n)$ using the fact that binary integers with $n$ bits can be added and multiplied in $\wt\Oh(n)$ by the classical Sch\"onhage-Strassen algorithm \cite{SchonhageS71c}.
If we find a final state which is negative, then we have detected 
an accepted matrix with determinant~$-1$. Hence $L(\cA)$ is not included in $\SLZ$, and
we let $\cA'=\cA$ since $\cA$ must have a \tra whose label has determinant $-1$. We are done in this case.

Therefore, we may assume without restriction that all initial and final states are positive.
In the third phase we relabel \tras using the matrix
$s_{-1}=\vdmatrix {1}{0}{0}{-1}$ of order two and with determinant $-1$. For that we consider all
\tras $p\arc {m} q \in \del$, one after another in some order. 
We either transform $p\arc {m} q$ into a \tra $p\arc {m'} q$ such that $m'\in \SLZ$ or we detect that $L(\cA)$ is not included in $\SLZ$. Recall that $m=\vdmatrix abcd$ with $a,b,c,d\in \Z$ and $\det(m)= \pm 1$. 
We make the following case distinction. 
\begin{enumerate}
\item If both $p$ and~$q$ are positive, then either we have~$\det(m)=1$ and we let 
$m'=m\in \SLZ$ or, if $\det(m)\neq1$, we exit with an error message. 
\item If~$p$ is positive and~$q$ is negative, then we have~$\det(m)=-1$,  and we let 
$m'=m s_{-1}= \vdmatrix {a}{-b}{c}{-d}\in \SLZ$ or, if $\det(m)\neq-1$, we exit with an error message.
\item If~$p$ is negative and~$q$ is positive, then either we have~$\det(m)=-1$  and we let 
$m'=s_{-1}m = \vdmatrix {-a}{-b}{c}{d}\in \SLZ$ or, if $\det(m)\neq-1$,
we exit with an error message.
\item If~$p$ and~$q$ are negative, then either we have~$\det(m)=1$, and we let 
$m'=s_{-1}m s_{-1}= \vdmatrix {-a}{b}{c}{-d}\in \SLZ$ or, if $\det(m)\neq1$, we exit with an error message.
\end{enumerate}

Since $\cA$ is trim, an error message tells us that $\cA$ accepts 
a matrix with determinant $-1$, which implies that $L(\cA)\sm \SLZ\neq \es $. To see this, recall the marking procedure. 
As noted above, we can assume without restriction that the initial and final states are positive. Since the NFA is trim, for every state 
$p\in Q$, there are matrices $f_p$ and $g_p$ such that $f_p$ labels the path defined by the marking procedure from an initial state 
to $p$ and $g_p$ labels \emph{any} path from~$p$ and to some final state. The marking procedure tells us that $\det(f_p)$ is positive \IFF the state $p$ is positive. 
Assume that $\det(f_p)\neq \det(g_p)$. Then $\cA$ accepts $f_pg_p$ 
with $\det(f_p)\det(g_p)=-1$, and we know that $L(\cA)\sm L(\cA)\neq \es$.
So, if $L(\cA)\sse \SLZ$, then $\det(f_p) = \det(g_p)$ for every $p\in Q$.
Now, consider any \tra $p\arc h q$ in $\cA$, then $hg_q$ labels a path from $p$ to some final state, and we conclude $f_phg_q\in L(\cA)$. Therefore, $L(\cA)\sse \SLZ$ implies $\det(h)=\det(f_p) \det(g_q)= \det(f_p) \det(f_q)$. Now, $\det(f_p) \det(f_q)\neq \det(h)$ is exactly the situation when an error message occurs. Thus, an error message implies that $L(\cA)$ is not contained in $\SLZ$.
In this case we stop and let $\cA'=\cA$.

Finally, assume there was no error message. In this case, the construction is finished, and it produces an $\SLZ$-NFA $\cA'=(Q,\del',I,F)$.
It remains to verify $L(\cA')=L(\cA)$. To see this, we first show that $L(\cA) \sse L(\cA')\sse \SLZ$. Consider a path~$\pi$ in $\cA$ which begins in some initial state $\iota\in I$ and which ends in some final state $t\in F$ and which accepts $m\in \mathrm{GL}(2,\mathbb{Z})$. After the 
transformation we obtain a path~$\pi'$ having all labels in $\SLZ$. 

We claim that the path~$\pi'$ accepts the same matrix $m$ as before the transformation. Thus, the claim implies that $m\in \SLZ$ and $L(\cA) \sse L(\cA')\sse \SLZ$. We prove the claim by induction on the number of negative states on that path. 
Both states $\iota\in I$ and $t \in F$ are positive. If all states are positive, then the claim holds. Otherwise, the path~$\pi$ contains a subpath $p_0\arc {m_1'} p_1 \arc {m_2'} \cdots p_{k-1} \arc {m_k'}p_k$ with $k\geq 2$ where 
$p_0$ and $p_k$ are positive, but $p_1\lds  p_{k-1}$ are negative. 
This subpath corresponds to a path  $p_0\arc {m_1} p_1 \arc {m_2} \cdots p_{k-1} \arc {m_k}p_k$ in $\cA$ such that 
$\det(m_1) = \det(m_k)=-1$ and $\det(m_i) =1$ for 
$2\leq i\leq k-1$. By definition of $\cA'$ this yields
$m_1'=m_1s_{-1}$, $m_k'=s_{-1}m_k$, and $m_i'=s_{-1}m_is_{-1}$ for 
$2\leq i\leq k-1$. Hence, $m_{1,k}= m_1\cdots m_k= m'_1\cdots m'_k\in \SLZ$. 
By \prref{lem:silva}, we can temporally add in $\cA$ and  $\cA'$ the same \tra
$p_0\arc {m_{1,k}} p_k$ between positive states without changing the accepted language.
Note that adding these \tras does not affect the above procedure because both $p_0$ and $p_k$ are positive and $\det(m_{1,k})=1$.
However with the new \tras we obtain shorter accepting paths which visit less negative states. We are done by induction on the number of negative states on the accepting path~$\pi$. Removing the newly added \tra brings us back to the NFAs $\cA$ and $\cA'$.  This shows the claim, which implies $L(\cA)\sse L(\cA')$. 

To see that we also have the inclusion $L(\cA')\sse L(\cA)$, consider an accepting path $\pi'$ in $\cA'$ that accepts a matrix $m'\in \SLZ$. Since there is a one-to-one correspondence between the \tras of $\cA$ and $\cA'$, we can construct and accepting path $\pi$ in $\cA$, which accepts some matrix $m$, such that $\pi$ is transformed into $\pi'$ by the above procedure. It follows by the previous argument that $m=m'$ and hence $L(\cA')\sse L(\cA)$.
\prref{fig:SLZGLZ} illustrates this transformation:
the states $p,v,w,t$ are positive and denoted as $p_+,v_+,w_+,t_+$
and $q,u$ are negative and denoted as $q_-, u_-$. The labels before the transformation are in the upper line. The new labels are in the lower line. 
\end{proof}
\begin{figure}[h]
\begin{center}
\begin{tikzpicture}[xscale=1.95,yscale=1.25]
\draw (-1,0) node (p0) {$I\ni \iota_+$};
\draw (0,0) node (p1) {$p_+$};
\draw (1,0) node (p2) {$q_-$};
\draw (2,0) node (p3) {$u_-$};
\draw (3,0) node (p4) {$v_+$};
\draw (4,0) node (p5) {$w_+$};
\draw (5,0) node (p6) {$t_+\in F$};

\draw (-1,1) node (op0) {$I\ni \iota_+$};
\draw (0,1) node (op1) {$p_+$};
\draw (1,1) node (op2) {$q_-$};
\draw (2,1) node (op3) {$u_-$};
\draw (3,1) node (op4) {$v_+$};
\draw (4,1) node (op5) {$w_+$};
\draw (5,1) node (op6) {$t_+\in F$};

\draw[dotted,->,>=latex] (op0) to node [above] {$\ast$} (op1);
\draw[->,>=latex] (op1) to node [above] {$m_1$} (op2);
\draw[->,>=latex] (op2) to node [above] {$m_2$} (op3);
\draw[->,>=latex] (op3) to node [above] {$m_3$} (op4);
\draw[->,>=latex] (op4) to node [above] {$m_4$} (op5);
\draw[dotted,->,>=latex] (op5) to node [above] {$\ast$} (op6);

\draw[dotted,->,>=latex] (p0) to node [above] {$\ast$} (p1);
\draw[->,>=latex] (p1) to node [above] {$m_1s_{-1}$} (p2);
\draw[->,>=latex] (p2) to node [above] {$s_{-1}m_2 s_{-1}$} (p3);
\draw[->,>=latex] (p3) to node [above] {$s_{-1}m_3$} (p4);
\draw[->,>=latex] (p4) to node [above] {$m_4$} (p5);
\draw[dotted,->,>=latex] (p5) to node [above] {$\ast$} (p6);
\end{tikzpicture}
\end{center}
\caption{The upper line is a subpath of an accepting path~$\pi$ in $\cA$. 
The lower line is the same subpath of~$\pi$ in $\cA'$. Since~$s_{-1}^2 =1$ we have $m_1\cdots m_4= m_1'\cdots m_4'\in \SLZ$.
}\label{fig:SLZGLZ}
\end{figure}
\begin{definition}\label{def:enu}
We say that a semigroup $S$ has an \emph{\enu} if 
	there exist a finite alphabet $\Sig$ and a surjective mapping $\eta: W_S \to S$ such that the following holds:
\begin{enumerate}
\item The subset $W_S\sse \Sig^+$ is decidable (as defined in formal language theory, \cite{HU}). 
\item On input $u,v\in W_S$ we can compute a word $w \in W_S$ such that $\eta(u)\cdot \eta(v)=\eta(w)$.
\end{enumerate}
If $S$ has an \enu and $L\sse S$ is a subset, then we say that the \emph{membership problem for $L$} is decidable, if 
$\oi{\eta(L)} \sse \Sig^+$ is a decidable language.
We say that $S$ has a \emph{decidable word problem}, if on input $u,v\in W_S$ it is decidable whether $\eta(v)=\eta(w)$ in~$S$.

A pair $(S,G)$, where $S$ a semigroup and $G$ is a subgroup of $S$,
has an \emph{\enu} if, in addition to the above, 
$\oi{\eta(G)} \sse W_S$ is decidable; and on input 
$w\in \oi{\eta(G)}$ we can compute some word $\wt w \in W_S$ such that $\eta(\wt w)=\oi{\eta(w)}\in G$. Moreover, if $S=M$ is a monoid, then we view $W_M\sse \Sig^*$ by defining $\eta(1)=1\in M$.

Finally, if $S=G$ is explicitly specified as a group, then we assume that for all 
$w\in W_G$ we can compute some word $\wt w \in W_G$ such that $\eta(\wt w)=\oi{\eta(w)}\in G$.
\end{definition}
The Greek letter $\eta$ used in the definition above stands for \emph{evaluation}. In next proposition, there is also a letter~$\rho$ standing for \emph{representation}. \begin{proposition}\label{prop:alpinvol}
Let $G$ be a group having an \enu in the notation of \prref{def:enu}.
If $\Gam_+\sse W_G$ is a finite subset, and $K\leq G$ is the subgroup generated by $\eta(\Gam_+)$, then we can compute a finite 
set $\Gam\sse W_G$ with an \invol~$\ovv$ and a mapping $\rho:\Gam_+\to \Gam$ such that $\eta(\Gam)$ 
generates~$K$, $\eta(\rho(u)) = \eta(u)$ for all $u\in \Gam_+$, and $\eta(\ov u)=\oi{\eta(u)}$ for all $u\in \Gam$.
\end{proposition}

\begin{proof}
Choosing a linear order on $\Sig$, defines a shortlex-ordering $\leq$ on $\Sig^*$. Therefore, $\leq$ is also a linear order on $W_G\sse \Sig^*$.
Given a word $w\in W_G$, we denote by $\wt w$ the word which is computed on input $w$ and which satisfies $\eta(\wt w)=\oi{\eta(w)}\in G$. (The existence of such an algorithm is part of \prref{def:enu}.)
In the beginning, we let $\rho(u) = u$ and $\ov u=\wt u$, for all $u\in \Gam_+$, and let $\Gam= \rho(\Gam_+) \cup \set{\ov u}{u\in \rho(\Gam_+)}$, but $\Gam$,~$\rho$, and~$\ovv$ will change dynamically. Note that initially $\rho = \id_{\Gam_+}$, $\Gam= \Gam_+ \cup \ov{\Gam_+}$, and $\eta(\Gam)$ generates the subgroup~$K$. Later we change~$\rho$, and we extend~$\ovv$ to an \invol on $\Gam$. During the construction, we will preserve the following invariants: $\eta(\Gam)$ generates~$K$, $\eta(\ov u) = \oi{\eta(u)}$ for all $u \in\rho(\Gam_+)$, and $\eta(\rho(u)) = \eta(u)$ for all $u \in\Gam_+$.

Next, we make~$\ovv$ injective on $\rho(\Gam_+)$. Namely, if it happens that there are $w_1,w_2\in \rho(\Gam_+)$ with $\ov {w_1} = \ov {w_2}$ and $w_1< w_2$, then we remove $w_2$ from $\rho(\Gam_+)$ and replace $w_2$ everywhere by~$w_1$. For example, if we had $\ov u=w_2$ for some $u$, then now we have $\ov u=w_1$.
If we had $\rho(u)=w_2$ for some $u$, then now we have $\rho(u)=w_1$.
\Ip both $\rho(\Gam_+)$ and $\Gam$ become smaller.
Note that this modification preserves the above invariants because $\ov {w_1} = \ov {w_2}$ implies that $\eta(w_1) = \eta(w_2)$.

Since this procedure stops in a finite number of steps, the mappings $u\mapsto \ov u$ and~$\rho$ are computable.
Finally, we extend~$\ovv$ from $\rho(\Gam_+)$ to an \invol of $\Gam$ in a natural way: namely, if $v = \ov u$ for some $u\in \rho(\Gam_+)$, then
we define $\ov v = u$. Clearly, $\eta(\ov u)=\oi{\eta(u)}$ for all $u\in \Gam$.
\end{proof}

\begin{remark}\label{rem:recenu}
Every finitely generated semigroup $G$ has an \enu by choosing a surjective \hom $\eta: \Sig^+ \to G$ where $\Sig$ is finite and letting $W_G=\Sig^+$. For \fg monoids, the decidability of the word problem does neither depend on $\Sig$ nor on the \hom $\eta$. 
In this case, decidability of the word problem 
as defined in \prref{def:enu} coincides verbatim with the standard definition for \fg monoids as used for example in \cite{bo93springer}.
   
Note that one can construct 
a finitely presented semigroup with an undecidable word problem, see Markov \cite{Markov}. It is considered to be the first undecidability result in algebra. The corresponding result for groups is more difficult. It was shown first in independent papers of Novikov and Boone~\cite{nov55,boone59}.

The group $\GL(n,\mathbb{Q})$ has an \enu; and its word problem is decidable. It is also clear that $\GL(n,\mathbb{Q})$ is not finitely generated for $n\geq 1$.
~\hspace*{\fill}$\diamond$\end{remark}
\begin{lemma}\label{lem:finind}
Let $H$ be a finite index subgroup of $G$.  
Then 
\begin{equation}\label{eq:finexHow}
\set{L\sse H}{L\in \RAT(G)} =\set{L\cap H}{L\in \RAT(G)}.
\end{equation}

\end{lemma}

\begin{proof}
 The inclusion $\sse$ is trivial. 
 The other inclusion is clear by \prref{prop:freekle} since $[G:H]<\infty$ implies  $H\in \Rec(G)$.
\end{proof}

\prref{lem:finind} cannot be extended to the case where $H$ has infinite index. For example, the extension fails as soon  $G$ does not have the so-called 
Howson property. The \emph{Howson property} states  that the intersection of two \fg~subgroups is  finitely generated.\footnote{If $G$ is not Howson, consider \fg subgroups 
 $L$ and $H$ such that  $K=L\cap H$ is not \fg\ Hence 
 $K\sse H$ but $K\notin \RAT(G)$; thus \prref{eq:finexHow} fails.}
 
The free groups satisfy the Howson property~\cite{Howson54}.
The following lemma shows that this is not the case for a direct product of nontrivial free groups. It is well-known and follows easily from~\cite{ben69} and standard results in trace theory~\cite{dr95}. For convenience, we provide a proof below. 
\begin{lemma}\label{lem:dr95How}
The direct product $G= F(a,b)\times F(c)$ does not satisfy the Howson property. Here
$F(a,b)$ and $F(c)$  denote  free groups of rank~$2$ and rank~$1$, respectively. 

More precisely, let $H$ be the subgroup of in $G$ which is generated by $(a,c)$ and $(b,1)$, and let $L$ be the subgroup of $G$ generated by $(a,1)$ and $(b,c)$. 
Then $K=H\cap L$ is not rational. (\Ip it is not finitely generated by   \prref{rem:recinG}.)
\end{lemma} 
\begin{proof}
By contradiction assume $K\in \RAT(G)$.
Choose any set of monoid generators of $F(a,b)$ which includes the letters $a$ and $b$. Let $h$ be the canonical inclusion of the free monoid $\os{a,b}^*$ into $F(a,b)$. The family $\RAT(F(a,b))$ is closed under intersection by~\cite{ben69}. Hence, $R= \pi(K) \cap a^*b^*\in \RAT(F(a,b))$. By the second item of \prref{prop:freekle} there is a regular set $K\in \REG(\os{a,b}^*)$ such that $h(K)=R$. A direct calculation shows $\set{a^nb^n}{n\in \N} \sse K\sse \set{w\in \os{a,b}^*}{|w|_a=|w|_b}$. But there is no such regular set~$K$ because otherwise $\set{a^nb^n}{n\in \N}= K \cap a^*b^* \in \REG(\os{a,b}^*)$, which is not regular, see \cite{HU}. 
A contradiction. \end{proof}

\begin{remark}\label{rem:sl34z}
\prref{lem:dr95How} also implies that $\RAT(\SL(4,\Z))$ is not closed under finite intersection
because $\SL(4,\Z)$ contains the product $\SLZ\times \Z$ and
$\SLZ$ contains a free group of rank~$2$.
On the other hand, it is still open whether $\SL(3,\Z)$ satisfies 
the Howson property, see \cite{LongR2011}; and we also do not know whether $\RAT(\SL(3,\Z))$ is closed under finite intersection.
\hspace*{\fill}$\diamond$\end{remark}

Let $M$ be a monoid and $N\leq M$ be a submonoid. 
Following the French school around \Schuetz, we say that $(M,N)$ satisfies the \emph{Fatou property}\footnote{The notation was coined for groups in \cite{BerstelS86MFCS} as an analogue of a result of Fatou who published in 1904 that a rational series of $\Q[x]$ whose coefficients are all integers is a rational series of $\Z[x]$.}
if 
\begin{align}\label{eq:fatou}
\RAT(N)= \set{L\in \RAT(M)}{L\sse N}
\end{align}
Even for \fg commutative monoids the Fatou property 
 does not hold in general. To see this let 
 $M=\N \times \N$. Then $N= (0,0) \cup \set{(m,n)\in M}{m\geq 1}$ 
 is easily seen to be submonoid of $M$ which is not finitely generated (see also \prref{rem:recinG}). 
 Hence $N\notin \Rat(N)$. On the other hand, we have $N\in \Rat(M)$ because
in the additive notation
$\set{(m,n)\in M}{m\geq 1}$ is the linear set $(1,0) +\N(1,0) + \N(0,1)$,
and hence $N$ is a semi-linear subset\footnote{The definition of semi-linear set is in the third item of \prref{exs:class}.} of $M$.

Thus, we need some restrictions either on $M$ or $N$, or both.
In \cite{BerstelS86MFCS,FrougnySS89} 
it is stated that the Fatou property holds for groups by similar arguments as given in \cite{AnisimovS75}. However, the authors do not give any proofs. The first published proof (we are aware of) 
was given by Herbst using the notion of \emph{star height}, see \cite{Herbst91}. An immediate corollary of \prref{thm:silva} is that the Fatou property holds for groups. (In order to have a reference, we state this explicitly in \prref{cor:silvana}.)
Our proof of \prref{thm:silva} uses 
NFAs which is important for our complexity results. Under the assumption that $G$ is \fg and that the index~$[G:H]$ is 
finite the Fatou property for groups has been shown in \cite{Grunschlag1999,Silva02} and, for \fg~virtually free~groups, in \cite{sen96actainf}. To the best of our knowledge, our proof that works directly with NFA's without increasing their sizes was first published  in
the conference paper \cite{DiekertPS20}. We apply it to 
$\SLZ$ and $\mathrm{GL}(2,\mathbb{Q})$. Here, $\mathrm{GL}(2,\mathbb{Q})$ is not finitely generated, and 
the index~$[\mathrm{GL}(2,\mathbb{Q}):\SLZ]$ is infinite.

\begin{theorem}\label{thm:silva}
Let $\cA$ be a $G$-NFA and~$K$ be the subgroup of a group $G$ which is generated by $L(\cA)\sse G$. Then there is a trim~$K$-NFA $\cA'$ which accepts $L(\cA)$  such that the number of states and transitions is bounded by that of $\cA$.  
Moreover, if $G$ has an \enu and if the labels of $\cA$ are given by words in the decidable set $W_G$ as in \prref{def:enu}, then the construction of $\cA'$ is effective.
\end{theorem} 
\begin{proof}
First, we trim the automaton $\cA$.
Therefore, from now on, we assume that every state~$p$ 
(and hence every \tra) is on some accepting path. There is a finite set $\Gam\sse G$ such that for every \tra $p\arc g q$ we have both $g$ and $\oi g$ in $\Gam$. For $g\in \Gam$ we define $\ov g$ by $\ov g= \oi g$. Thus, $\Gam$ is finite set with \invol. The inclusion
$\eta:\Gam\sse G$ induces a \hom $\psi: \Gam^*\to G$ from the free monoid $\Gam^*$ with \invol onto~$K$. Recall that the \invol 
on a word $a_1\cdots a_k$ with $a_i\in \Gam$ is defined by 
$\ov{a_k}\cdots \ov{a_1}$. Thus, $\psi$ respects the \invol.

In case when $G$ has an \enu, we know by assumption that all labels of $\cA$ belong to a decidable set $W_G \sse \Sig^*$ as in \prref{def:enu}. We let $\Gam_+\sse W_G$ be the finite set of labels $u\in W_G$ which 
appear on some \tra $p\arc u q$. By \prref{prop:alpinvol},
there is a computable mapping~$\rho$ from $\Gam_+$ to some finite 
subset $\Gam\sse W_G$ with \invol~$\ovv$ 
such that $\eta(\Gam)$ generates the the same subgroup as $\eta(\Gam_+)$,
$\eta(\ov u)=\oi{\eta(u)}$ for all $u\in \Gam$, and $\eta(\rho(u)) = \eta(u)$ for all $u\in \Gam_+$. Thus, as in the case when $\Gam\sse G$ above, $\eta$ can be extended to a  \hom $\psi: \Gam^*\to G$ from the free monoid $\Gam^*$ to $G$ which respects the \invol. Using~$\rho$, we relabel all \tras in $\cA$ by letters in $\Gam$. 

Therefore we can use a unified notation for both cases $\Gam\sse G$ and~$\Gam\sse W_G$. 
\Ip even for $\Gam\sse G$ we write $\psi(L(\cA))$ rather than $L(\cA)$. That is, we consider $\cA$ as an automaton over the free monoid $\Gam^*$ rather than $G$ since every 
sequence $g_1\lds g_k$ of~$k$ elements in $G$ has a natural evaluation $g_1\cdots g_k$ in $G$ which coincides with $\psi(g_1\cdots g_k)$.
So, in our notation,~$K$ is the subgroup generated by $\psi(L(\cA))$.

Since $\cA$ is trim, for every state~$p$ of $\cA$ there are shortest words~$u_p,v_p\in \Gam^*$ such that~$u_p$ is the label of a path from an initial state to~$p$ and~$v_p$ is the label of a path from~$p$ to a final state. 
Since $K=\langle\psi(L(\cA))\rangle$ we have $\psi(u_p  v_p) \in K$ for all $p\in Q$. 
We also have $\psi(\ov{u_p})= \oi{\psi(u_p)}$ and 
$\psi(v_p)\in  \psi(\ov {u_p})K$. Therefore the left-coset of $\psi(v_p)$ in $G/K$ is unique: it depends on~$p$ and not on the choice of $v_p$. Thus, we can write $\psi(v_p)\in  \psi(r_p)K$ with $r_p=\ov {u_p}$ for $p\notin I\cup F$.  For $p\in I\cup F$, we can choose $r_p=\ov{r_p}=1$, where $1$ denotes the empty word in $\Gam^*$. This choice is possible since for $p\in I$ (resp.,~$p\in F$) we have 
$u_p=1$ (resp.,~$v_p=1$), and hence $\psi(v_p)\in K$. 

Next, we make $\Gam$ possibly larger such that $\Gam$ 
contains two letters $r_p$ and $\ov{r_p}$ for all $p\notin I\cup F$. We define (respectively redefine if necessary) $\eta$ for $r_p$ and $\ov{r_p}$ by 
$\eta(r_p)= \oi{\psi(u_p)}$ and $\eta(\ov{r_p})= \psi(u_p)$.
As above, $\eta$ induces a \hom $\psi: \Gam^*\to G$ respecting the \invol.  

Having defined the coset representatives $r_p$, we transform the NFA $\cA$ into an NFA~$\cB$ as follows. The state 
space of~$\cB$ is defined as the union $Q\cup \ov{Q}$ where $\ov{Q}$ is a disjoint copy of~$Q$. We denote the copy of $p\in Q$ by $\ov p\in \ov{Q}$.

The \tras in~$\cB$ are defined in two steps. In the first step, we 
introduce for each $p\in Q$ an additional outgoing \tra
$p\arc  {r_p} \ov p$ and an additional incoming \tra
$\ov p\arc {\ov{r_p}} p$. Since $\psi(r_p)\psi(\ov{r_p})=1\in G$ for all $p\in Q$, this does not change the accepted language by \prref{lem:silva}. Recall that $r_p=\ov{r_p}=1$ for all $p\in I\cup F$. Thus, an $\eps$-\tra (that is, a \tra with label~$1$) leads from~$p$ to~$\ov p$ and from~$\ov p$ to~$p$ for all $p \in I\cup F$. Therefore we do not change the accepted language by enlarging the sets of initial and final states by $I\cup \ov I$ and $F\cup \ov F$, respectively. 

In the second step, we consider every \tra $p\arc a q\in \del$  with $p,q\in Q$ in some order. 
Since $\phi(u_pv_p)\in K$, $\phi(u_pav_q)\in K$, and $\psi(v_p)K= \psi(r_p)K$, we have
$\psi(av_q)K= \psi(v_p)K = \psi(r_p)K$. We also have $\psi(v_q)K= \psi(r_q)K$, and therefore $K=\psi(\ov{r_p})\psi(a) \psi(v_q)K= \psi(\ov{r_p}ar_q)K$, which is equivalent to $\psi(\ov{r_p}\, a\, {r_q})\in K$.

Hence, defining $h=\ov{r_p}\, a\, {r_q}\in \Gam^*$, we obtain $\psi(h)\in K$. 
Having this, we introduce for~$\cB$
a new \tra $\ov p\arc h \ov q$. See \prref{fig:cB} for a visualization of the NFA~$\cB$.
\begin{figure}[h]
\begin{center}
\begin{tikzpicture}[xscale=2.95,yscale=1.5]
\draw (-1,1) node (lt) {$\ov I \ni \ov \iota$};
\draw (-1,0) node (lb) {$I \ni \iota$};
\draw (3,1) node (rt) {$\ov \tau \in \ov F$};
\draw (3,0) node (rb) {$\tau \in F$};

\draw (0,1) node (pr) {$\ov p$};
\draw (1,1) node (qt) {$\ov q$};
\draw (0,0) node (p) {$p$};
\draw (1,0) node (q) {$q$};
\draw (2,1) node (tr) {$\ov t$};
\draw (2,0) node (t) {$t$};
 
\draw[dotted,->,>=latex] (lt) -- (pr);
\draw[dotted,->,>=latex] (lb) -- (p);
\draw[dotted,->,>=latex] (tr) -- (rt);
\draw[dotted,->,>=latex] (t) -- (rb);

\draw[->,>=latex] (pr) to node [above] {$h= \ov{r_p}\, a\,{r_q}$} (qt);
\draw[->,>=latex] (p) to node [above] {$a$} (q);
\draw[->,>=latex,bend left=10] (lt) to node [right] {$1$} (lb);
\draw[->,>=latex,bend left=10] (lb) to node [left] {$1$} (lt);
\draw[->,>=latex,bend left=10] (rt) to node [right] {$1$} (rb);
\draw[->,>=latex,bend left=10] (rb) to node [left] {$1$} (rt);

\draw[->,>=latex,bend left=10] (pr) to node [right] {$\ov{r_p}$} (p);
\draw[->,>=latex,bend left=10] (p) to node [left] {${r_p}$} (pr);
\draw[->,>=latex,bend left=10] (qt) to node [right] {$\ov{r_q}$} (q);
\draw[->,>=latex,bend left=10] (q) to node [left] {${r_q}$} (qt);
\draw[->,>=latex] (q) to node [above] {$b$} (t);

\draw[->,>=latex] (qt) to node [above] {$h'= \ov{r_q}\, b\,{r_t}$} (tr);
\draw[->,>=latex,bend left=10] (tr) to node [right] {$\ov{r_t}$} (t);
\draw[->,>=latex,bend left=10] (t) to node [left] {${r_t}$} (tr);
\end{tikzpicture}
\end{center}
\caption{The construction of the NFA~$\cB$ yields 
$h=\ov{r_p}\, a {r_q}\in \oi{\psi}(H)$ and $h'=\ov{r_q}\, b {r_t}\in \oi{\psi}(H)$.
}\label{fig:cB}
\end{figure}\\
We claim that $\psi(L(\cA)) = \psi(L(\cB))$. The inclusion $\psi(L(\cA))\sse \psi(L(\cB))$
is trivial. For the other direction we use \prref{lem:silva}: since 
$\psi(h)=\psi(\ov{r_p}\, a {r_q})$ in $G$, we did not change $\psi(L(\cA))$. 

Finally, we define the NFA $\cA'$ by removing from~$\cB$ all states in~$Q$ (together with the incident \tras). \Ip all the remaining \tras are of the 
form $\ov p\arc {h} \ov q$ with $\psi(h)\in K$, the set of initial states is $\ov I$, and the set of final states is $\ov F$. 
We can think of $\cA'$ as a disjoint copy of $\cA$ where 
a \tra $p\arc a q$ with $a\in \Gam$ has been replaced in its copy by the \tra
$\ov p\arc h \ov q$ with label $h=\ov{r_p}\, a\, {r_q}$ such that $\psi(h)\in K$. 
Note that the construction of $\cA'$ is effective if $G$ has an \enu and if the labels of $\cA$ are in the decidable set $W_G$. 

Since we already know that $\psi(L(\cA))=\psi(L(\cB))$, it remains to show
$\psi(L(\cA'))=\psi(L(\cB))$. For this we use a dual construction. Note that if we define $\ov{\ov p}=p$ for all $\ov p \in \ov Q$, then $Q \cup \ov Q$ becomes a set with \invol. Now we perform the same construction as above starting with $\cA'$ (which is the upper line in \prref{fig:cB}) but replacing $p$ with $\ov p$, $r_p$ with $\ov r_p$, etc. \Ip we will have a \tra $p \larc {{r_p}\, h\, \ov{r_q}} q$ between $p$ and $q$ in \prref{fig:cB} instead of $p \larc {a} q$. Let $\cB'$ be the resulting automaton. Since $\psi({r_p}h\ov{r_q}) = \psi(a)$, we conclude that $\psi(L(\cB'))=\psi(L(\cB))$. On the other hand, by \prref{lem:silva}, we have $\psi(L(\cA'))=\psi(L(\cB'))$. This completes the proof of the theorem. 
\end{proof}

\begin{corollary}\label{cor:silvana}
Let $G$ be a group with a subgroup $H$ and $\cA$ be a $G$-NFA with~$n$ states and~$m$ \tras such that $L(\cA)\sse H$. Then there is 
a (trim) $H$-NFA $\cA'$ with at most~$n$ states and at most~$m$ \tras such that $L(\cA')=L(\cA)$. \Ip the groups satisfy  
the \emph{Fatou property} of \prref{eq:fatou}: that is, we have
$\set{L\sse H}{L\in \RAT(G)}=\RAT(H)$.
\end{corollary}
\begin{remark}\label{rem:NyBrodda}
Nyberg-Brodda has recently shown in \cite{NybergBrodda2023} that there
is finitely generated (and \emph{context-free}) monoid $M$ such that its group of units is a rational but not finitely generated. 
Thus, \fg monoids fail to satisfy the Fatou property with respect to subgroups. In his example there is a set of three generators 
$\os{a,b,c}$. The defining relations are $\set{(ab^ic)^2=1}{i \in \N}$. The resulting semi-Thue system is easily seen to be confluent and Noetherian. It follows that  
$M$ is not Dedekind-finite (since
$acac=1$ but $1\neq caca$) and its group of units is the 
rational submonoid $F = (ab^* c)^*$. Thus, $U(M)$ is the free product $F=\ast_{n\in \N} \Z/2\Z$, which is not \fg
\hspace*{\fill}$\diamond$\end{remark}

\begin{corollary}\label{cor:silvaH}
Let $G$ have an \enu and~$H$ be a subgroup such that the membership problem for~$H$ is decidable.
Then, we can decide for a $G$-NFA $\cA$, whose labels are given by words in the set $W_G$ in the notation of \prref{def:enu}, whether $L(\cA)\sse H$. 
\end{corollary}

\begin{proof}
Let~$K$ be the subgroup of $G$ generated by $L(\cA)$. 
We apply \prref{thm:silva} to effectively construct a~$K$-NFA $\cA'$ such that  $L(\cA')=L(\cA)$ and where the \tras in $\cA'$ have labels in $W_G$ such that their image in $G$ generates the subgroup~$K$.
Therefore, $L(\cA)\sse H$ is decidable because 
the membership problem for~$H$ is decidable, and hence we can check whether the labels of \tras of $\cA'$ belong to~$H$.  
\end{proof}

\begin{corollary}\label{cor:finext}
Let $G$ be a \fg group and~$H$ a subgroup of finite index. Then $\Rat(H)$ is a Boolean algebra \IFF  $\Rat(G)$ is a Boolean algebra. Moreover, the membership problem for rational subsets of~$H$ is decidable \IFF it is decidable for~$\Rat(G)$. \end{corollary}
\begin{proof}
It is well-known and easy to see that $G$ is \fg~\IFF~$H$ is \fg Therefore both groups $G$ and~$H$ are \fg \Ip they have \enu{s}, which allows us to apply the effectiveness condition in \prref{thm:silva}.

Assume  that $\Rat(G)$ is a Boolean algebra. Let us show that $\Rat(H)$ is a Boolean algebra, too. Note that 
$\Rat(H)\sse \Rat(G)$; and we have $H\in \Rat(G)$ since~$H$ is finitely generated.
Thus, for every $R\in \Rat(H)$, we have $H\sm R\in \Rat(G)$, and hence $H\sm R\in \Rat(H)$ by \prref{cor:silvana}.  This shows that $\Rat(H)$ is a Boolean algebra. If the
membership for rational sets of $G$ is decidable, then the membership for rational sets of~$H$ is decidable because $\Rat(H)\sse \Rat(G)$.

For the other direction, assume $\Rat(H)$ is a Boolean algebra. In order to show that $\Rat(G)$ is a Boolean algebra let $R\in \RAT(G)$. We have to show that $G\sm R\in \RAT(G)$. 
Since the index $[G:H]$ is finite, there is subgroup $N\leq H$ which is normal in $G$. (Actually,  $N=\bigcap\set{gH\oi g}{g\in G}$ and the intersection is finite since $[G:H]<\infty$.) Let $\phi:G\to G/N$ be the canonical \hom. Then $\phi$ recognizes $H$. 

Let $\os{r_1,\ldots,r_k}\sse G$ be
representatives of left cosets of~$H$, where $k=[G:H]$, such that for each $g\in G$ there is exactly one $r_g$ with 
$g\in r_gH$. 
Thus, $g\notin R$ 
\IFF $r_g^{-1}g\in H\sm r_g^{-1}R$. In other words, $G\sm R = \bigcup_{i=1}^k r_i(H\sm r_i^{-1}R) = \bigcup_{i=1}^k r_i(H\sm (r_i^{-1}N\cap N))$.

By \prref{prop:freekle} we have $r_i^{-1}R \cap H\in \Rat(G)$ because~$H$ is recognizable. 
By \prref{cor:silvana} we have $r_i^{-1}R \cap H\in \Rat(H)$. Since $\Rat(H)$ is a Boolean algebra, $H\sm (r_i^{-1}R \cap H)\in \Rat(H)$, and
we conclude that $G\sm R\in \RAT(H)$.

It remains to show that the membership for $\RAT(G)$ is decidable if the membership for $\RAT(H)$ is decidable. Since $G$ is \fg there 
is some finite generating subset $\Gam\sse G\sm \os 1$ such that $\Gam=\oi{\Gam}$. Thus, every word in $w\in \Gam^*$ has a natural interpretation in the group $G$. 
The \emph{Schreier graph}, also called the \emph{coset graph}, has been defined in \cite{Schreier1927} for~$H$ \wrt~$\Gam$. It is a directed graph
where the set of vertices $V$ is the finite set of all left 
cosets: $V=\set{gH}{g\in G}$. The directed edges are labeled by generators and defined as 
$gH \arc a agH$ for all $a\in \Gam$ and $gH\in V$. Thus, the out-degree of 
each vertex is $|\Gam|$.
We construct the  Schreier graph of $H$ by exhaustive search. 
The construction yields rooted tree $T$ where the nodes $V(T)$ are words in $\Gam^*$. We begin with $T=\os 1$ where $1$ the empty word representing the coset $H$. 
 During the process some 
nodes without  \emph{children} will become a leaf in the final tree. For that we define a subset $L(T)\sse V(T)$ which initially  is empty. The invariant is that all nodes in $L(T)$ are leaves. 

Next, while $V(T)\sm L(T)\neq \es$ we repeat the following loop.
\begin{enumerate}
\item Choose any node $g\in V(T)\sm L(T)$.
\item For each $a\in \Gam$ (in some order) consider
the word $ag \in \Gam^*$, and decide whether $agH=hH$ for some 
$h\in V(T)$. (This is possible because the membership in
$\oi h ag \in H$ is decidable.) 
 If for all $h\in V(T)$ we have $\oi hag \notin H$, then 
the word $ag$ represents the coset $agH$ (which was not represented in $V(T)$ so far) and we add $ag$ to $V(T)$ as a child of $g$.
\item If $g$ is still without any child by the previous step, then $g$ becomes a leaf in the tree $T$. That is, we update $L(T)$ redefining it as $L(T)\cup \os g$. 
\end{enumerate}
Let us show that the algorithm terminates. The first observation is that 
$V(T)$ grows as long as $|V(T)|$ is less than the index of $H$ in $G$.
Thus, the algorithm reaches a point where $|V(T)|=[G:H]$. Having this, all nodes without children become leaves because $\Gam$ is finite. 
At this point the algorithm stops with $V(T)=\os{r_1,\ldots,r_k}\sse G$. The representatives $r_i\in V(T)$ are written as words in $\Gam^*$.

After computing the coset representatives $\os{r_1,\ldots,r_k}\sse G$ using the above procedure, we can decide membership to $R\in \Rat(G)$ using the following equivalence: $g\notin R$ 
	\IFF for some $i\in \{1,\ldots,k\}$, we have $r_i^{-1}g\in H\sm (r_i^{-1}R\cap H)$.
\end{proof}

\begin{remark}\label{rem:ToddCox}
The algorithm in the proof  \prref{cor:finext} yields a coset enumeration, and the algorithm is typically called the \emph{Todd-Coxeter coset-enumeration}. Its original version in~\cite{ToddCoxeter1936} was designed for finding a finite presentation for finite groups, only. For finitely presented groups the coset-enumeration yields an effective construction of the Schreier graph if $[G:H]$ is finite, 
see~\cite{LS01}. 
However, even for finitely presented groups there is no computable upper time bound for the Todd-Coxeter coset-enumeration in general. 
\end{remark}

\section{Smith normal forms and commensurators}\label{sec:SNFC}

It is a classical fact from linear algebra that every matrix 
$m\in \MnQ$ admits a \emph{Smith normal form}. 
For $n=2$, the \SNF of a non-zero $m\in \Q^{2 \times 2}$ is a factorization 
\begin{equation}\label{eq:snf2}
m=r\, e\, \vdmatrix 100q\,f
\end{equation}
such that $r\in \Q$ is a positive rational number, $e,f\in  \SL(2,\Z)$, and $q\in \Z$. Note that we may assume that $r$ is positive because $m\neq 0$ and $\vdmatrix {-1}00{-1}\in \SLZ$. Since $r^2q= \det(m)$, the sign of~$\det(m)$ is determined by the sign of~$q$.
For $q\in \Z$ we fix the notation
\begin{equation}\label{eq:sqref}
s_q=\vdmatrix 100q.
\end{equation}
If we write $m=r\, e\, s_q\,f$ for $m\in \MQ$, then we refer 
to it as the \SNF of $m$ according to (\ref{eq:snf2}) and  (\ref{eq:sqref}).
We use \SNF{s} only when $n=2$.
The computation of \SNF is closely related to Gaussian elimination and relies on $\gcd$-computations. More details are given in \prref{sec:comSNF}.

\subsection{Computation of the \SNF}\label{sec:comSNF}
As mentioned above, a \emph{\SNF} of a non-zero matrix $m$ in $\MQ$ is defined by a factorization $m= r \cdot e\vdmatrix100q f$ where 
$0<r\in \Q$, $e,f\in \SL(2,\Z)$, and $q\in \Z$. Moreover, $r$ and~$q$ are uniquely determined by the matrix~$m$ (but $e$ and $f$ are not unique).
The uniqueness of $r$ and~$q$ can be seen as follows. Let 
$m= r_1 \cdot e_1\vdmatrix100p f_1= r_2 \cdot e_2\vdmatrix100q f_2$ with $0<r_i\in \Q$, $e_i,f_i\in \SL(2,\Z)$ for $i=1,2$, and $p,q\in \Z$.
Multiplying $m$ on the left by $\oi {r_2} \cdot \oi{e_2}$ and on the right by $\oi{f_1}$ yields $\frac{r_1}{r_2} \cdot e\vdmatrix100p = \vdmatrix100 q f$ with $e,f\in\SLZ$. 
Since $0<\frac{r_1}{r_2}$ we can write $\frac{r_1}{r_2}=\frac{s}{t}$, where $s,t$ are positive natural numbers 
such that $\gcd(s,t)=1$. 
Therefore, it is enough to show that 
$\frac st \cdot e\vdmatrix100p = \vdmatrix100 q f$ implies 
$s/t= 1$ and $p=q$. 
Let $e=(e_{ij})$ and $f=(f_{ij})$, then 
\[
\VDmatrix{se_{11}}{spe_{12}}{se_{21}}{spe_{22}}= \VDmatrix{tf_{11}}{tf_{12}}{tqf_{21}}{tqf_{22}}.
\]
Since $\gcd(s,t)=1$, the positive integer $t$ divides $e_{11}$ and
 $e_{21}$. Hence, $t$ divides~$\det(e)=1$. Thus, $t=1$, and by symmetry we also have $s=1$. 
 Therefore, $e\vdmatrix100p= \vdmatrix100 q f$, and hence 
 $\det(\vdmatrix100p)=\det(\vdmatrix100q)$. Clearly, this implies $p=q$.
 
 The following lemma is a special case of a polynomial-time result by Kannan and Bachem~\cite{KannanBachem79}. We include a proof 
 because the result for $2\times 2$ matrices is rather easy to show. Moreover, for $2\times 2$ matrices we obtain a soft cubic time bound whereas~\cite{KannanBachem79} just states polynomial time.\footnote{We did not check whether ``soft cubic time'' is an upper bound for computing the \SNF in higher dimensions, too.}

\begin{lemma}\label{lem:KB}
On input $0\neq m\in \MQ$ with $n=\Abbin {m}$
we can compute $0<r\in \Q$, matrices $e,f\in \SLZ$, and $q\in \Z$ in soft-cubic time $\wt \Oh(n^3)$ such that 
$m = r\cdot e\vdmatrix100q f$.
\end{lemma}
Our proof follows~\cite{KannanBachem79}. 
It relies on the fact that $\gcd$'s can be computed in cubic time. This fact is straightforward, but it is not optimal. For example, Sch\"{o}nhage \cite{Schonhage1971ActaInf} gives a $\Oh(n(\log n)^2(\log\log n))$ algorithm. M\"{o}ller \cite{Moeller2008} gives another 
quasi-linear time algorithm which (according to M\"{o}ller) runs slightly faster than earlier quasi-linear time algorithms.
\begin{proof} 
On input $m$ we calculate some positive integer~$p$ such that $p\cdot m=\vdmatrix abcd$ where $A=\vdmatrix abcd\in \MZ$. For example, we may choose the product over the denominators of all entries in $m$.
Knowing the \SNF of $A$, we obtain the \SNF of $m$ by multiplication with
$\oi p$. Hence, \Wlog, we assume that $m = A$, and let $D=\det(A)$.
 
With the help of matrices $e,f\in \SLZ$ we may assume $a = \Abm {A}\geq 1$.
If $a=1$, then we are done: we have $r=1$ and $q=D$ because
\(
\vdmatrix 10{-c}1\cdot
\vdmatrix 1bcd \cdot \vdmatrix 1{-b}01= \vdmatrix 10{0}{d-bc}
\).
Hence, from now on we assume $a\geq 2$.  
In the first phase we reduce the problem to the case where $A$ is a diagonal matrix. This is true if $b=c=0$. By symmetry, we may assume in the first phase that $a\geq 2$ and $b\neq 0$.\\
\textbf{First phase.}
Let $1\leq g=\gcd(a,b) = pa +qb$ with $0\leq q < a$. This is possible since $(p+b)a +(q-a)b = pa +qb$. 
 Then 
$\vdmatrix p  {-b/g}q{a/g}\in \SL(2,\Z)$, and hence
\[\vdmatrix abcd \cdot \vdmatrix p {-b/g}q{a/g} = 
\vdmatrix{g}{0}{pc+qd}{D/g}.\]
If $\gcd(a,b) = a$, then we choose $p=1$ and $q=0$. 
Otherwise $a/g\geq 2$ and, since $b\neq 0$, we have
$1\leq \abs b< a= \Abm A$. Hence: 
\[\abs{p}= \abs{\frac{g-qb}{a}}\leq \frac{g}{a} + \frac{(a-1)\abs b}{a}
\leq 1/2 + a-1 <a.\]
Thus, after the first step and by left-right symmetry due to transposition of matrices, 
we may assume without restriction
that we actually start with a matrix 
\[A'\in \left\{\vdmatrix g 0 {D'} {D/g},\, \vdmatrix g {D'} 0 {D/g}\right\},\]
where $g | a$ and $0\leq \abs{D'} < 2 {\Abs A}_{\text{max}}^2$.
If $D'=0$, then we stop because the matrix is diagonal which is the aim for this phase.

We now assume that $D'\neq 0$ and $A'= \vdmatrix g 0 {D'} {D/g}$. Let $g'=\gcd(g,D') = pg +qD'$ with $0\leq q<g$. We have $\vdmatrix p q {-D'/g'}{g/g'}\in \SLZ$ and 
\[\vdmatrix p q {-D'/g'}{g/g'}\cdot \vdmatrix g 0 {D'} {D/g}= 
\vdmatrix{g'}{q D/g}{0}{D/g'}.\]
If $g\mid D'$, then $q=0$ and the above matrix is diagonal. So, we stop the first phase. 
Thus, without restriction $0< q<g$, and, \ip$g\neq 1$.
Clearly:
$g'\mid  g \mid a$. Let $D''= q D/g$. Then
\[0\leq \abs{D''} < \abs{D} \leq 2 {\Abs A}_{\text{max}}^2.\]
Since $0< q<g$, we have $g'<g$.
Since each time we have either $g/g' \geq 2$ or $g\mid D'$, we finish after at most $\log {\Abs A}_{\text{max}}$ steps. 
This completes the first phase.\\
\textbf{Second phase.} We continue with a matrix $A'' = \vdmatrix g00{D/g}$ for some
$g\mid a$. If $g\mid D/g$ we are done. Thus, 
\Wlog $D\neq 0$ and  letting
 $d=D/g$ we write 
 \[\vdmatrix g00{d}= 
 \vdmatrix {\gcd(g,d)}{0}{0}{\gcd(g,d)} \cdot \vdmatrix {g/\gcd(g,d)}{0}{0}{d/\gcd(g,d)}.\]
 Let $g'=g/\gcd(g,d)$ and $d'=d/\gcd(g,d)$. Note that $g'\geq 2$ because $g\neq \gcd(g,d)$.
 We add the right column of $\vdmatrix{g'}00{d'}$ to the left one by multiplying with the matrix $\vdmatrix 1011$. We obtain the matrix $\vdmatrix {g'}{0}{d'}{d'}$.
We let $pg'+qd'=1$ with ${0\leq}q<g'$ and $p=(1-qd')/g'$. Hence, $\abs p {\leq \abs{d'} + 1/g' -\abs{d'}/g'  \leq \abs{d'}} \leq \abs D{\leq 2 {\Abs A}_{\text{max}}^2}$. 
Then,  \[
\vdmatrix pq{-{d'}}{g'}\cdot\vdmatrix {g'}0{d'}{d'} = \vdmatrix 1{qd'}0{{g'}d'}.
\]
Subtracting  $q{d'}$ times the left column from the right one
by multiplying with the matrix $\vdmatrix 1{-q{d'}}01$, we obtain 
the desired result. 
\end{proof}

\subsection{Commensurators}\label{sec:commens}
The notion of a commensurator is well established in group theory.  
Let $G$ be a group and $H$ be its subgroup. Then the \emph{commensurator of $H$ in $G$} is defined to be the set of all $g\in G$ such that $H \cap H^g$ has finite index in $H$ and in $H^g$, see for example~\cite[Def.~5.17]{drutuK2018}. 
Here, and in the following, we abbreviate  $g H \oi g$ as $H^g$
which is a standard notation in group theory.
It is a known fact that the commensurator is a subgroup of $G$, see~\cite[Ex.~5.18]{drutuK2018}.\footnote{ Note that in geometric group theory there is a more general notion of an \emph{abstract commensurator}, which is different from what we use here, see \cite[Def.~5.13]{drutuK2018}.}

Now, let $H$ be an arbitrary group. For the sake of brevity, we say that a group $G$ containing $H$ is \emph{a commensurator of} $H$ 
if  for all $g\in G$ the subgroup 
$H\cap H^{g}$ has finite index in~$H$.\footnote{In \cite{krieg1990hecke}, a group $G$ and its subgroup $H$ that satisfy such property are called a \emph{Hecke pair} $(H,G)$.} Note that this also implies that $[H^g : H \cap  H^g] = [H : H^{\oi g} \cap  H]$ is finite for all $g\in G$. Hence $G$ is the commensurator of $H$ in $G$.

If~$H$ has finite index in $G$, then $G$ is a commensurator of~$H$ because the intersection $H\cap H^{g}$ has finite index in $G$ (and hence in $H$) for any $g\in G$.
If $K\leq H$ are subgroups of $G$ and $G$ is a commensurator of $K$, then obviously $H$ is a commensurator of $K$, too.
We also use the following lemma in the proof of \prref{prop:commens}.
\begin{lemma}\label{lem:ind}
Let $K\leq H \leq G$ be a chain of subgroups such that the index $[H:K]$ is finite. Then $G$ is a commensurator of $K$ \IFF $G$ is a commensurator of $H$.
\end{lemma}	
\begin{proof}
Suppose that $G$ is a commensurator of $K$. Then for all $g\in G$ we have:
\begin{align*}
[H:H \cap  H^{g}]\leq [H:K \cap  K^{g}]= [H:K][K:K \cap  K^{g}] <\infty.
\end{align*}
Since  $[H:K]$ is finite, $G$ is a commensurator of $H$.
For the other direction, it is enough to show that for all $g\in G$ we have
$[K:K \cap  K^{g}]\leq [H:H \cap  H^{g}][H:K]$ since, by assumption, both $[H:H \cap  H^{g}]$ and $[H:K]$ are finite.
For that, we start with the following equation
\begin{equation}\label{eq:comm}
\begin{split}
[H:K][K:K \cap  K^{g}] &=	[H:K \cap  K^{g}]\\ &= [H:H \cap  H^{g}][H \cap  H^{g}:K \cap  H^{g}][K \cap  H^{g}:K \cap  K^{g}].
\end{split}
\end{equation}
Next, we use the fact that for all subgroups $N$ and $K$ of $H$, the set of left cosets $N/K\cap N$ embeds into 
$H/K$, and hence $[N : N\cap K] \leq [H:K]$.
The fact implies that
\begin{align*}
[H \cap  H^{g}:K \cap  H^{g}] &= [H \cap  H^{g}: (H \cap  H^{g}) \cap  K] \leq [H:K]\quad \text{and}\\
[K \cap  H^{g}:K \cap  K^{g}] &= [K \cap  H^{g}: (K \cap  H^{g}) \cap  K^{g}]
\leq [H^{g}:K^{g}]= [H:K].
\end{align*}
Substituting these inequalities in~\eqref{eq:comm} and dividing every term by $[H:K]$, we obtain that $[K:K \cap  K^{g}]\leq [H:H \cap  H^{g}][H:K]$, which proves the lemma.
\end{proof}

The statement of the following \prref{prop:commens}
holds for all $n\in \N$.
However, the case $n=2$ has a short proof which is also given below. 
\begin{proposition}\label{prop:commens}The group $\GL(n,\Q)$ is a commensurator of $\SL(n,\Z)$ and of any subgroup $G\leq \GL(n,\Q)$ which contain $\SL(n,\Z)$ as a subgroup of finite index.
\Ip $\GL(n,\Q)$ is a commensurator of both $\SL(n,\Z)$ and $\GL(n,\Z)$.
\end{proposition}

\begin{proof}
In~\cite[Ch.~V]{krieg1990hecke}, it is shown that $\GL(n,\Q)$ is a commensurator of $\GL(n,\Z)$ for all $n\in \N$.
Using \prref{lem:ind}, we see that $\GL(n,\Q)$ is a commensurator of $\SL(n,\Z)$, too. This implies the result.
\end{proof}

For $n=2$ we give a short and direct proof of \prref{prop:commens} based on
the \SNF, which we have defined only for $n=2$.
For $n\geq 3$, such a proof becomes more technical (see \cite[Ch.~V]{krieg1990hecke}).
Our applications only concern $2\times 2$ matrices.
\begin{proof}[Proof of \prref{prop:commens} for $\boldsymbol{n=2}$]
It is enough to show that $\GL(2,\Q)$ is a commensurator of $H = \SL(2,\Z)$.
To see this, 
recall that $s_q= \vdmatrix 100q$ (as in \prref{eq:sqref}), where $q\in \Z$. 
Writing a matrix $g\in \GL(2,\Q)$ in its Smith normal form yields $g=r\, e\, s_q\,f$ with $r\in \Q$ and $e,f\in \SLZ$. Then the index of $g H \oi g\cap
H$
in~$H$ is the same as the index of $s_q H s^{-1}_q \cap H$ in~$H$. 
We have $s^{-1}_q \vdmatrix a b c d s_q= \vdmatrix{a} {qb} {c/q} d$. Hence, a matrix $\vdmatrix a b c d\in H$ belongs to the intersection
$s_q H s^{-1}_q \cap H$ \IFF $c\in q\Z$. Thus, 
$\ker({\bmod\, q})\sse s_q H s^{-1}_q \cap H$, where ${\bmod\, q}:\SL(2,\Z)\to \SL(2,\Z/q \Z)$ is the canonical \hom. Thus, the index of 
$s_q H s^{-1}_q\cap H$ in~$H$ is bounded by the size of the finite group $\SL(2,\Z/q \Z)$.
It follows that $\GL(2,\Q)$ is a commensurator of $\SL(2,\Z)$. \end{proof}

The size of $\SL(2,\Z/q \Z)$ is obviously bounded by some polynomial in~$q$. This would be good enough for our purposes, but not good enough in practical applications. As a matter of fact, 
there is a better and more precise estimate for the index of $s_q H s^{-1}_q\cap H$ in~$H$ which is stated next.
\begin{proposition}\label{prop:indexH}
As above, denote $H=\SL(2,\Z)$. Let $g\in \mathrm{GL}(2,\mathbb{Q})$ and $g=r\, e\, s_q\,f$ be its \SNF with $0<r\in \Q$, $e,f\in\mathrm{GL}(2,\mathbb{Z})$, and $s_q=\vdmatrix 100q$. 
Then
\[
[H:(g H \oi g \cap H)] = [H:(s_q H s^{-1}_q\cap H)] \in \Oh(|q|\log |q|).
\]
\end{proposition}
\begin{proof}We just have seen above that $[H:(g H \oi g \cap H)] = [H:(s_q H s^{-1}_q\cap H)]$, and that $s_q H s^{-1}_q\cap H$ consists of those matrices $\vdmatrix a b c d\in H$ for which $c\in q\Z$. The subgroup $s_q H s^{-1}_q\cap H$ is also denoted as $\Gamma_0(q)$ in the literature. The index of $\Gamma_0(q)$ in~$H$ is equal to $|q|\prod_{p|q}(1+1/p)$, where the product is taken over all prime divisors of~$q$, see~\cite[Ex.~1.2.3(e)]{DiamondS2005}.

We now estimate the above product $\prod_{p|q}(1+1/p)$. Note that
\[
\ln \prod_{p|q}\Big(1+\frac{1}{p}\Big) = \sum_{p|q}\ln \Big(1+\frac{1}{p}\Big)\ \leq\ \sum_{p|q}\frac{1}{p}\ \leq\ \sum_{p\leq |q|}\frac{1}{p} \ \leq\ \ln \ln |q| + C
\]
for some constant $C>0$, where the sums and product are taken over all primes $p$ such that $p\,|\,q$ or $p\leq |q|$, respectively. The last inequality follows from Mertens's Second Theorem, see~\cite[p.~466]{HardyW2008}.
Therefore, $\prod_{p|q}(1+1/p) \leq e^C \ln |q|\in \Oh(\log |q|)$ and $[H:(s_q H s^{-1}_q\cap H)] \in \Oh(|q|\log |q|)$.
\end{proof}

\section{Dichotomy in $\GL(2,\Q)$}\label{sec:algglq}
One of the main results of this paper is \prref{thm:nofinext} stated below, which classifies the \fg subgroups $G$ sitting strictly between $\GLZ$ and $\GLQ$ into two mutually exclusive classes. An important consequence of this dichotomy is that, for such subgroups, $\Rat(G)$ is never closed under intersection, and in particular it is not a
relative Boolean algebra. This is a result of independent interest.
In our proof of the dichotomy,
the \emph{Baumslag-Solitar group} $\BS(p,q)$ where $1=p<q$ shows up.\footnote{The group $\BS(p,q)$ is an HNN-extension (named after Higman, Neumann, and Neumann) of $\Z$ over the subgroups $p\Z$ and $q\Z$ with a ``stable letter'' $t$.} Recall that $\BS(p,q)= \langle a,t\ |\ ta^pt^{-1} = a^q\rangle$ is actually defined for all $p,q\in \Z$, but up to \iso it is enough 
to impose $0\leq p\leq |q|$. As we will see, $\BS(|q|,q)$ for $|q|\geq 2$ contains 
a direct product of a free group of rank two and $\Z$. This is a consequence of 
Bass-Serre theory~\cite{serre80}, see for example \cite{FredenK2007}.

The case $p=0$ is not very interesting since 
$\BS(0,q)$ is isomorphic to the free product $\Z\ast (\Z/q\Z)$. It is fairly easy to see that $\BS(p,q)$ has no free subgroup of finite index unless
$pq=0$, see~\cite{gersten89}.
As a consequence, in both cases of the dichotomy in \prref{thm:nofinext}, the group $\GL(2,\Z)$ has infinite index in $G$ when $\GL(2,\Z)< G \leq \GLQ$.

Actually, we prove more: if $G$ contains a matrix of the form $\vdmatrix {r_1} 00
{r_2}$ with $\abs{r_1} \neq \abs{r_2}$ (which is the second case of the dichotomy), then  $G$ contains some
$\mathop\mathrm{BS}(1,q)$ for $q\geq 2$ which has \textbf{infinite} index in $G$. It is wide open whether the membership for rational
subsets of $G$ can be decided in that second case.

For example, let $p\geq 2$ be a prime, and let $G'$ be generated by $\vdmatrix  0{-1}10$,
$\vdmatrix  1{1}01$, and $\vdmatrix  100p$. In this case 
$\vdmatrix  {p}00{\oi p}$ also belongs to $G'$. Let $\Z[1/p]$ denote the ring $\set{p^nr\in \Q}{n,r \in \Z}$. It is known by
\cite{bm68}  that $\vdmatrix  0{-1}10$, $\vdmatrix  1{1}01$, and $\vdmatrix  {p}00{\oi p}$ generate the special linear group  $\SL(2,\Z[1/p])$ of $2\times 2$ matrices over~$\Z[1/p]$. Hence, $G'$ contains $\SL(2,\Z[1/p])$ as a subgroup. The structure of $\SL(2,\Z[1/p])$ is described in 
\cite[Chapter~II, Sect.~1.4, Cor.~2]{serre80}: it
is an amalgam of two copies of $\SL(2,\Z)$ over a common subgroup of finite index.
It is however unknown how to decide subgroup membership for such amalgams. Moreover, $\vdmatrix{1}00p$ acts by conjugation on
$\SL(2,\Z[1/p])$, and since $\vdmatrix{1}00p$ generates an infinite cyclic group, $G'$ is a semi-direct product of the form $G'=\SL(2,\Z[1/p])\rtimes \Z$. Hence, even if the subgroup membership  for $\SL(2,\Z[1/p])$ were decidable, it could still be undecidable in  $G'$. The situation is more friendly for the subgroup generated by the matrices $\vdmatrix  1{1}01$ and $\vdmatrix  1{0}0p$ because it is the group
$\UT(2,\Z[1/p])\rtimes \Z\cong Z[1/p]\rtimes \Z\cong\mathop\mathrm{BS}(1,p)$, where $\UT(2,\Z[1/p])$ is the group of $2\times2$ upper-unitriangular matrices over $\Z[1/p]$.
The membership problem for rational subsets of $\mathop\mathrm{BS}(1,q)$ is decidable for all $q\geq 2$
by~\cite{CadilhacCZ2020ICALP}. However, it is not clear how to generalize this result to extensions of $\BS(1,q)$ of infinite index.

It is also shown in~\cite[Ex.~3.7]{CadilhacCZ2020ICALP} that 
$\RAT(\BS(1,q))$ is not closed under finite intersection for $q\geq 2$. 
Using \prref{thm:silva}, we show next that this non-closure property holds whenever $0\leq p \leq |q|$ and $|q|\geq 2$. \Ip it covers the ``famous'' Baumslag-Solitar group $\BS(2,3)$ and cases where $q$ is negative.
To the best of our knowledge the following dichotomy theorem for Baumslag-Solitar groups has not been stated explicitly or shown elsewhere. 
\begin{theorem}\label{thm:bsinter}
Let $p,q\in \Z$ and $\BS(p,q)$ be the  Baumslag-Solitar group. Then $\RAT(\BS(p,q))$ is a Boolean algebra \IFF it is closed under finite intersection \IFF $|pq|\leq 1$.
\end{theorem}
\begin{proof} We will use some well-known facts about Baumslag-Solitar groups.
What we need for the proof can be found, for example, in \cite[Sect.~8.4.2]{edam16}
and elsewhere in the literature. For example, as we mentioned above, we assume without restriction that
$0\leq p\leq |q|$. 
We  let $\ov a = \oi a$ and $\ov t = \oi t$. 
The group $\BS(0,q)$ is the free product $\Z\ast (\Z/q\Z)$, hence $\Rat(\BS(0,q))$ is a Boolean algebra by \cite{Sak92,LohSen06}.
Therefore, we only need to consider the case when $p\geq 1$. 
Let $p=|q|=1$; then we know that 
$\RAT(\BS(1,1))$ is a Boolean algebra since 
$\BS(1,1)=\Z \times \Z$, and therefore the rational sets are the semi-linear subsets. 
Also, $\RAT(\BS(1,-1))$ is a Boolean algebra by \prref{cor:finext} since 
it contains $\Z \times \Z$ as a subgroup of index two.
It remains to show that $\RAT(\BS(p,q))$ is not closed under intersection for $|q|\geq 2$ and $1\leq p \leq |q|$.

We treat the case $2\leq p =|q|$ first.
Consider the \fg subgroup $F$ of $\BS(|q|,q)$ which is generated by the two commutators 
$\alp=[a,t]$ and $\bet=[a,t^2]$. Then we obtain a natural 
\epi $\psi$ from the free group $F_{\alp,\bet}$ onto $F$. Now, consider 
a non-trivial freely reduced word $w\neq 1$ in $F_{\alp,\bet}$. Then 
a Britton-reduction (\wrt $a$ and $t$) yields a nontrivial element in $\BS(|q|,q)$ since $|q|\geq 2$. 
For example,  $\alp \bet=at\ov a \ov t \; a t^2 \ov a \ov t^2$ is Britton-reduced, and the Britton reduction of 
$\alp \oi \bet$ yields $at\ov a \ov t \;t^2  a \ov t^2 \ov a= at\ov a  t  a \ov t^2 \ov a$.  Based on this observation, standard arguments with an induction on $|w|$ show that $\psi$ is injective. Therefore $F$ is a free group. 
It is easy to check that $a^p$ commutes with $\alp$ and $\bet$ when $p=|q|$. This implies that
the intersection of the infinite cyclic group $\gen{a^p}$ 
and $F$ is trivial. Thus, $\BS(q,|q|)$ for $|q|\geq 2$ 
contains a direct product isomorphic to $F_2\times \Z$, where $F_2$ is the free group of rank two generated by $\alp,\bet$ and $\Z$ is generated by $a^p$. We have seen in \prref{lem:dr95How} that $F_2\times \Z$ does not satisfy the Howson property. Therefore  $\Rat(\BS(q,|q|))$ is not closed under finite intersection.

In order to finish the proof it remains to consider $\BS(p,|q|)$
where $1\leq p<|q|$.
The proof has a different flavor than the one for $\BS(|q|,q)$ with $|q|\geq 2$. We let 
$A_+=a^0\cup \cdots \cup a^{p-1}$, and $A_- =A_+$ if $q$ is positive and 
$A_-=a^0\cup \cdots \cup a^{1-p}$ if $q$ is negative. We consider the set 
$L=a^*\cap T^* a (\ov t\,\ov t\,)^*$  with $T=tA_- tA_+$. Then $L$ is 
the intersection of two rational sets. We claim that $L$ is not rational.
By contradiction, assume that $L\in \RAT(\BS(p,q))$. Then, 
by \prref{thm:silva}, there is an $a^\Z$-NFA $\cA$ which accepts $L$.
The set $L$ is not empty since $a\in L$. If
$a^k\in L$, then $1\leq k \in \N$ and we can write $a^k\in T^s a t^{-2s}$ for some $s\in \N$.
Thanks to the choice of $A_+$ and $A_-$, we can also state that for each 
$s\in \N$ there is a unique $k_{s}\in \N$ such that $a^{k_s}\in T^s a t^{-2s}$ and $k_{s}\geq (q/p)^{2s}$.
This can be shown by induction on $s$.
More precisely, for each $k_{s}$ there a unique $k'_s$ such that 
$1\leq k_s\leq k'_s \leq k_s + p-1$ and $k'_s \in p\N$. 
Let us define $\ell_s=(q/p)k'_s$. Then we have $a^{\ell_s}\in tA_+T^s a t^{-2s+1}$.
Note that $\ell_s<0\iff q <0$. 
Now, consider the unique  $\ell'_s$  
with $\ell_s\leq \ell'_s \leq \ell_s+p-1$ if $q>0$ (or with 
$\ell_s\geq \ell'_s \geq \ell_s+1-p$ if $q<0$) such that $\ell'_s\in p\Z$.
This leads to the next positive $k_{s+1} \in \N$ such that $k_{s+1}= (q/p)\ell'_s
\geq (q/p)^2 k_s$. 

Putting things together, we have shown that
$L=\set{a^{k_s}}{s \in \N}$ with
$k_{s}\geq (q/p)^{2s}$ for all $s\in \N$. The assumption $L\in \RAT(\BS(p,q))$ implies
$L=\psi(K)$ for some \hom $\psi$ and some regular language $K\sse \os{a, t,\ov t}^*$. 
By the pumping lemma for regular languages (also known as $uvw$-Theorem), we know that for  
all $s$ there is some $r\neq s$ with $|k_s-k_r|\in \Oh(1)$. It is a contradiction with the above lower bound on $k_s$.
\end{proof}  

Another ingredient to show the dichotomy is the next proposition and its corollary.
\begin{proposition}\label{prop:ZA}
Let $G$, $H$, and $A$ be groups, where the center $Z(H)$ is trivial and $A$ is Abelian. 
If $\phi:H\to G\times A$ is an injective \homo, then the induced \hom $\phi_1:H\to G$ is injective where $\phi_1(h)=g$ for $\phi(h)=(g,a)$.
\end{proposition}
\begin{proof}
It is enough to show that $\phi_1(h)=1$ implies $h=1$. To see this, take $h\in H$ with $\phi(h)=(1,a)$. Then $(1,a)$ is in the center of $G\times A$ and therefore in the center of $\phi(H) \cong H$. Therefore 
$h\in Z(H)$ which is trivial. Hence, $\phi(h)=(1,a)$ implies $h=1$ and we are done. 
\end{proof}

The following corollary holds for all $\BS(p,q)$ where $1\leq p<|q|$ because the center of those $\BS(p,q)$ is trivial. More generally, the center of an HNN-extension $\text{HNN}(H,t,\phi)$ with an \iso $\phi: A \to B$ is trivial if $A \neq H$ and $\set{a\in Z(A)}{\phi(a)=a}=\os 1$. 
However, we need \prref{cor:ZA} only for $\BS(1,q)$ with $q\geq 2$: in this case the proof is less technical.
\begin{corollary}\label{cor:ZA}
Let $G$ and $A$ be groups where $A$ is Abelian. If $1<|q|$ and the Baumslag-Solitar group $\BS(1,q)$ appears as a subgroup in the $G\times A$, then $\BS(1,q)$ appears in~$G$. 
\end{corollary}

\begin{proof}
The group $\BS(1,q)$ is isomorphic to the  semi-direct product
$ \Z[1/q]\rtimes \Z
$. The elements of $\Z[1/q]\rtimes \Z$ are pairs 
$(r,m)$ where $r=pq^e$ with $p,e,m\in \Z$ and the multiplication
$(r,m)\cdot(s,n)= (r+q^m s, m+n)$. A direct verification shows 
that the center of $\Z[1/q]\rtimes \Z$ is trivial. 
Thus,  \prref{prop:ZA} yields the result.
\end{proof}

\begin{theorem}\label{thm:nofinext}
Let $G$ be a \fg~group such that  $\GL(2,\Z)< G \leq\GL(2,\Q)$. 
Then there are two mutually exclusive cases. 
\begin{enumerate}
  \item $G$ is isomorphic to $\GL(2,\Z)\times \Z^k$ for some $k\geq 1$.
  \item  $G$ contains a subgroup which is an extension of infinite index of $\BS(1,q)$ for some $q\geq 2$.
\end{enumerate}
Furthermore, in both cases of the dichotomy, $\RAT(G)$ is not closed under finite intersection.
\end{theorem}
\begin{proof}
 We distinguish two cases. 
 In the first case, we suppose that  $G$ is generated by $\GLZ$ and finitely many elements from the center 
 $Z(G)$. Since  $\GLZ$ is a subgroup of $G$, we see that 
 $Z(G)\leq \set{\vdmatrix r00r}{r\in \Q^*}$. 
 Moreover, since 
 $-1\in \GLZ<G$, the group $G$ is generated by $\GLZ$ and a nontrivial \fg~subgroup $Z\leq
\set{\vdmatrix r00r}{r\in \Q^* \wedge r>0}$.
 Hence $Z\cong \Z^k$ 
 for some $k\geq 1$ because $\set{r\in \Q^*}{r>0}$ is torsion free and 
 $Z$ is finitely generated and Abelian.
Since $\GLZ$ contains a free group of rank $2$ and $k\geq 1$, \prref{lem:dr95How} tells us that $\Rat(G)$ is not closed under finite intersection. 
  
Assume we are not in the first case. Then consider
any finite generating set of $G$ and write the generators in their Smith normal form
$re\vdmatrix 10 0 {q} f$ with $0<r\in \Q$, $e,f\in \GL(2,\Z)$ and
$q\in \Z$. Since 
 $\vdmatrix 1 0 0 {-1}\in \GL(2,\Z)<G$, the
generators can be chosen from $\GL(2,\Z)$ and matrices of the form  $\vdmatrix r 0 0 {rq}$ with $0<r\in \Q$ and
$0\neq q\in \N$. Note that there is at least one generator $s=\vdmatrix r
0 0 {rq}$
where $r>0$ and $2\leq q\in \N$, because otherwise we are in the first case.

As usual, we define $\BS(1,q)= \langle a,t\ |\ tat^{-1} = a^q\rangle$ in standard group generators $a$ and~$t$.
Let $b=\vdmatrix 1011$ and $\phi:\BS(1,q)\to G$ be a \homo such that $\phi(a)=b$ and $\phi(t)=s$. It is well-defined since $s \vdmatrix 1011  \oi s = \vdmatrix 1011^q$. Let  $\BS=\phi(\BS(1,q))$. 
We claim that $\phi$ is an \iso between $\BS(1,q)$ and $\BS$. To see the claim 
we observe that every element 
$g\in \BS(1,q)$ can be written in the form $t^kb^xt^n$ where 
$k, x, n$ are integers.  Suppose $g=t^kb^xt^n$ and 
$\phi(g)= 1$. Then $\vdmatrix 10x1=\phi(b^x)= \phi(t^{-n-k}) = \vdmatrix r00{rq}^{-n-k}$ is a diagonal matrix and $x=0$. Hence, $g=t^{k+n}$ and 
 $\phi(g)=s^{k+n}=1$. This implies $k+n=0$, and  $\phi$ is injective. Hence, the claim. 
 
Next, we show that BS has infinite index in $G$. Consider any $g\in \BS\cap \SL(2,\Z)$. As above, consider  $f=s^{k}\vdmatrix 1011^xs^m$ with $x,k,m\in \Z$. Since
 by assumption~$\det(f)=1$, we obtain $m=-k$ and hence $f = \vdmatrix 10{q^kx}1 \in {\vdmatrix 1011}^\Z$. Therefore
 $\SL(2,\Z) \cap \BS$ is the infinite cyclic group generated by $\vdmatrix 1011$. It has infinite
 index in $\SL(2,\Z)$. It follows that $G$ contains an extension of
 $\BS(1,q)$ of infinite index.

Finally, let us show that $\GL(2,\Z)\times \Z^k$ cannot contain 
$\BS(1,q)$ for $k\geq 0$. Otherwise, there is no dichotomy. For the sake of contradiction assume the contrary. By \prref{prop:ZA} this implies $\BS(1,q) \leq \GL(2,\Z)$. We have seen in \prref{sec:rba} that 
$\Rat(\GL(2,\Z))$ is a Boolean algebra because $\GL(2,\Z)$ is a \fg~and virtually-free. This implies that for the \fg~subgroup $\BS(1,q)$, the set $\RAT(\BS(1,q))$ is a Boolean algebra. \Ip it is closed under finite intersection. This is a contradiction to \prref{thm:bsinter}. 
\end{proof}
\begin{theorem}\label{thm:undec}
Let $G$ be isomorphic to $\GL(2,\Z)\times \Z^k$ with $k\geq 1$.
Then, on input  $L,R\in \RAT(G)$  it is undecidable whether $L=R$. 
However, on input $g\in G$ and $R\in \RAT(G)$ it is decidable whether $g\in R$.
\end{theorem}
\begin{proof}
By \prref{rem:glzvf}, we know that $\GL(2,\Z)$ has a free subgroup~$F_2$ of rank two and index~$24$.
\Ip $G$ contains 
 the  free partially commutative 
 monoid $M=\os{a,b}^*\times \os{c}^*$ with $a\neq b$. It was proved by Aalbersberg and Hoogeboom in~\cite{ah89} that the equality problem is undecidable for  $\RAT(M)$. 
 
For the decidability, we use a result by Lohrey and Steinberg~\cite{LohSte08}. They showed that 
the membership problem for $\RAT(F_2\times \Z^k)$ is decidable. 
Since  $F_2\times \Z^k$ has finite index in $G$, the membership problem for rational subsets in
 $G$ is decidable by \prref{cor:finext}.
\end{proof}
\section{Flat rational sets}\label{sec:FRAT}
In this section we introduce the notion of \emph{flat rational set}
for a semigroup $M$ and a subset $T$. 
If~$S = \gen T$ is a subsemigroup of~$M$ generated by $T$, then we can extend positive decidability results for $\RAT(S)$ to the larger family $\FRAT M S$. When $\RAT(M)$ is an effective Boolean algebra, then all the decision problems studied here are decidable. However,
for a group $G$ sitting between $\GLZ$ and $\GLQ$, the family
 $\RAT(G)$ is never a Boolean algebra unless
$G=\GLZ$, see \prref{thm:nofinext}. The main result of this section is \prref{thm:lea}. It shows that the membership
problem and (even stronger) the emptiness problem for Boolean combinations of flat rational sets are decidable for $\FRAT{\GL(2,\Q)}{\GL(2,\Z)}$.

The following definition is given for a semigroup $M$ and a 
subset $T\sse M$. The main interest is when $M$ is a monoid
and $T$ generates a submonoid $S=\gen T$. Below we also define when 
an $M$-NFA is flat over $T$. In this case $T$ is a subset of labels of its \tras.
\begin{definition}\label{def:fratmon}
We say that $L\sse M$ is \emph{flat rational} over a subset $T$
if $L$ is a finite union of languages of the form 
$L_0g_1L_1 \cdots g_t L_t$ where all  $L_i\in \RAT(\gen T)$ and $g_i\in M$.
\end{definition}
The family of flat rational subsets over~$T$ is denoted by $\FRAT M T$. If~$S=\gen{T}$, that is $T$ generates the subsemigroup~$S$ of $M$, then \prref{def:fratmon} implies
$\FRAT M T=\FRAT M S$.

In order to specify a set~$L$ in $\FRAT M S$ for a subsemigroup $S$ we can also use an $M$-NFA with a syntactic restriction as in \prref{def:flatNFA} with $T=S$. In this case, as soon as the membership to~$S$ is decidable, we can check whether an  $M$-NFA is flat over~$S$, and if it is, then we know that the accepted language belongs to $\FRAT M S$.
\begin{definition}\label{def:flatNFA}
	Let $T\sse M$. An $M$-NFA $\cA=(Q,\del,I,F)$
	is called \emph{flat over $T$} if no \tra having a label outside $T$ lies on a directed cycle. 
\end{definition}
\begin{remark}\label{rem:fratmon}
As we mentioned in the introduction, the notion of $\FRAT M S$ \emph{is a special case} of a \emph{polynomial closure} 
$\mathop{\text{Pol}}(M,\cL)$ introduced by Sch\"utzenberger in~\cite{sch76}: more precisely, in our special case we have $\cL=\Rat(S)$, where $S$ is a subsemigroup of~$M$.\footnote{The results in~\cite{sch76} characterize star-free (or aperiodic) 
languages as the polynomial closure over a language class by using  prefix codes of bounded-synchronization delay.
More results in this direction are in \cite{Schutzenberger1974b} and \cite{DiekertWalter17}.}
There is also a related notion of flatness in the context of finite control systems, see \cite{FinkelP19} and its references.\footnote{In control theory the definition says that every control-state belongs to at most one loop. 
}
\hspace*{\fill}$\diamond$
\end{remark}
The next theorem is a generalization of \prref{thm:silva}.  
\begin{theorem}\label{thm:genfratmon}
Let $M$ be a monoid such that all right-invertible elements are invertible\footnote{This means that $M$ is \dedfin, see \prref{rem:erata} for a short discussion of this notion.} and $H$  a subgroup of $U(M)$.
Then the family $\FRAT M H $ is the least family $\cR$ of subsets of $M$ satisfying the following conditions: 
\begin{itemize}
  \item $\cR$ contains   all finite subsets of $M$,
  \item $\cR$ is closed under finite union and concatenation,
  \item $\cR$ is closed under taking the Kleene-star over subsets of $H$  which belong to~$\cR$.
 \end{itemize}
\Ip this implies that $\set{L\sse H}{L\in\FRAT M H} =\RAT(H)$.
\end{theorem}
\begin{proof}
Clearly, $\RAT(H)\sse \cR$ and hence, all flat rational sets over $H$ are contained in $\cR$. To prove inclusion in the
 other direction, we need to show that the family of flat rational subsets of $M$ over $H$ (i) contains
 all finite subsets of $M$, (ii) is closed under finite union and concatenation, and (iii) is closed
 under taking the Kleene-star over subsets of $H$. The first two conditions are obvious. We show
 (iii) in two steps.  Let $L$ be a flat rational set over $H$ such that $L\sse H$.  First we show that $L\in \RAT(G)$, where $G=U(M)$ is the group of units of $M$. 
 Since $L\in \RAT(M)$, there is some $M$-NFA $\cA$ accepting $L$. After trimming,  we may assume without restriction that every \tra is used on some accepting path. Let $g$ be any label of a \tra. Then, thanks to trimming, there are $u,v\in M$ with $ugv\in L\sse H$. Hence, there is some $w\in H$ such that   $ugvw =1\in G$. Therefore, 
 $u$ has a right-inverse. Since $M$ is \dedfin, we have  $u\in G$ and $\oi u ugvwu =gvwu  =1$. It follows that $g$ has a right-inverse; and therefore $g\in G$. This shows the first step:
 $L\in \Rat(G)$.  
 
In the second step we apply \prref{thm:silva}. It  shows $L\in \Rat(H)$. 
Hence, $L^*\in \Rat(H)$, which concludes the proof of (iii). So, $\FRAT{M}{H}$ is closed under all three closure properties. It also shows 
$\set{L\sse H}{L\in\FRAT M H}=\RAT(H)$. 
 \end{proof}

\begin{remark}\label{rem:erata}
In the literature a
monoid~$M$ is called \emph{\dedfin} if all right-invertible elements are invertible. That is, $ab=1$ implies  $ba=1$ for all $a,b\in M$.  The notation appears for example in \cite{faith2003} and \cite[Def.~2.3.2]{antoniou2019book}. The class of \dedfin monoids is closed under taking submonoids. It includes all finite monoids, all cancellative monoids and hence, all groups. If~$F$ is a field, then $F^{n\times n}$ is \dedfin
because a matrix in $F^{n\times n}$ is invertible \IFF  its determinant is not zero. More results about \dedfin monoids are 
in the classical textbook~\cite{cp67}.
In our conference paper~\cite{DiekertPS20} the assertion of \prref{thm:genfratmon} was stated without 
the hypothesis that $M$ is \dedfin.
However, in our applications we only considered those monoids.
Further results in~\cite{DiekertPS20} were not affected by the missing hypothesis. The example in \prref{rem:NyBrodda} given by Nyberg-Brodda shows that \prref{thm:genfratmon} does not hold in general if  $M$ is not \dedfin. 
\hspace*{\fill}$\diamond$\end{remark}
\begin{theorem}\label{thm:lea}
Let $G$ be a group with an \enu, and $H\leq G$ be a subgroup such that the following conditions hold:
\begin{itemize}
\item The family $\RAT(H)$ is an effective relative Boolean algebra.
\item The group $G$ is a commensurator of $H$, and  on input $g\in G$, we can compute the index of $H_g= gH\oi {g} \cap H$ in $H$.
\item The membership problem for $H$ is decidable.
\end{itemize} 
Then  $\FRAT G H$ forms an effective relative Boolean algebra. \Ip given a
finite Boolean combination\footnote{Complementation in the Boolean combination is taken with respect to $G$.}  $B$ of flat rational sets of $G$ over $H$, we can decide the emptiness of $B$. 
\end{theorem}
Note that we do not require $\Rat(H)$ to be a Boolean algebra. In fact, it is a Boolean algebra \IFF $H\in \Rat(H)$ \IFF  $H$ is finitely generated.
Before giving the proof of \prref{thm:lea} let us first state one of its  consequences. 

\begin{corollary}\label{cor:leas}
Let $B\sse \GL(2,\Q)$ be a finite Boolean combination of flat rational sets of $\GL(2,\Q)$
over $\GL(2,\Z)$, then we can decide the emptiness of $B$. 
\end{corollary}
\begin{proof}
By \prref{rem:glzvf} the group $\GL(2,\Z)$ is a finitely generated  virtually free group. Hence, $\RAT(\GL(2,\Z))$ is an
  effective Boolean algebra by~\cite{Silva02}. The group $\GL(2,\Q)$ is infinitely generated, but obviously the group of matrices with rational entries has an \enu  in which the membership for $\GLZ$ is decidable.
In \prref{sec:SNFC}, we showed that $\GL(2,\Q)$ is a commensurator
  of its subgroup $\GL(2,\Z)$. The index of $\GL(2,\Z)_g= g\GL(2,\Z)\oi {g} \cap \GL(2,\Z)$ in $\GL(2,\Z)$ is bounded by $|\GL(2,\Z/q\Z)|$ if 
$g= r e\vdmatrix 100qf$ is the \SNF of $g$ (see
 \prref{prop:indexH}).
  Thus, all hypotheses of \prref{thm:lea} hold.
\end{proof}

For the proof of \prref{thm:lea} we will need the following lemma.
Recall the notation \(
H_g= gH\oi {g} \cap H= \set{h\in H}{\oi g h g\in H}
\) for $H\leq G$. Since we also defined $H^g$ as $gH\oi {g}$, we have 
$H_g= H^{g} \cap H$.

\begin{lemma}\label{lem:frida}
Let $G$ be a group and $H$ be a subgroup, $L\in \RAT(H)$, and $g\in G$.
Then under the assumptions of \prref{thm:lea} we can compute 
an $H$-NFA accepting $\oi g (L \cap H_g)g$. 
\end{lemma}

\begin{proof}
 Since $H_g=gH\oi g \cap H$ is of finite index in $H$, we can compute 
 an NFA $\cA'$ accepting  $L'= L\cap H_g\in \RAT(H_g)$ by \prref{lem:finind}. The labels 
of \tra{s} are in $H_g$. We have $\oi g H_g g \sse H$. Hence it is enough to change every label $h$ of \tra{s} in $\cA'$ to $\oi g h g$. 
 This gives the NFA $\cA$ for $\oi g (L \cap H_g)g$ over $H$. 
\end{proof}

\begin{proof}[Proof of \prref{thm:lea}]
Let $g\in G$ and $K\in \RAT(H)$. First, we claim that 
we can rewrite $Kg\in \RAT(G)$ as a finite union of languages
$g'K'$ with $g'\in G$ and $K'\in \RAT(H)$. Let us show the claim.

The rewriting process for $Kg$ begins with a computation of a set $U_g\sse H$ of left coset representatives of $H_g$ such that 
 $H=\bigcup\set{uH_g}{u\in U_g}$. This is possible because, by assumption, the membership for $H$ is decidable; 
and hence, the membership for $gH\oi g$ and for $H_g= gH\oi g \cap H$ is decidable, too. Moreover, by the second assumption, we can compute the index $k=|H:H_g|$.
Thus we can enumerate the elements of~$H$ until
	we find~$k$ elements that belong to~$k$ different left cosets of~$H_g$.
	Checking if two elements belong to the same coset is decidable since
	the membership for~$H_g$ can be decided. 
Thus, 
  \begin{align*}
Kg&= \bigcup\set{K\cap uH_g}{u\in U_g}g
=
\bigcup\set{ug\, \oi g (\oi uK \cap H_g)g}{u\in U_g}
\\&=\bigcup\set{g' \oi g (g{g'}^{-1} K \cap H_g)g}{g'\in U_g g}.
\end{align*}
Using \prref{lem:frida} we obtain~$\oi g (g{g'}^{-1} K \cap H_g)g=K'\in \RAT(H)$. 
 This shows the claim.

Note that since membership for~$H$ is decidable,
we can effectively enumerate a set~$S$ of all distinct representatives of the right cosets of~$H$,
and moreover for each~$g\in G$ find a representative~$g'\in S$ such that~$g\in g'H$. 

Let $L\in\FRAT{G}{H}$. Hence $L$ is a finite union of languages
$L_0g_1L_1 \cdots g_tL_t$ where all~$L_i\in \RAT(H)$. Using the claim, we can write~$L$ as a finite
union of languages~$gK$ with~$g\in G$ and~$K\in \RAT(H)$. By the above observation, we have
$g=g'h$ for some~$h\in H$ which can be effectively found. Hence we can write~$gK=g'(hK)$, where~$hK\in \RAT(H)$.
Therefore, every flat rational set~$L$ can be written as a union $L=\bigcup_{i=1}^n g_iK_i$, where
$g_i\in S$ and $K_i\in \RAT(H)$. Since $gK_1\cup gK_2=g(K_1\cup K_2)$, we may assume  that all $g_i$
in the expression $L=\bigcup_{i=1}^n g_iK_i$ are different.

Now let $L$ and $R$ be two flat rational sets. By the above argument we may assume that
$
  L=\bigcup_{i=1}^n a_iL_i$ and $ R=\bigcup_{j=1}^m b_jR_j,
$
where $a_i,b_j\in S$ and $L_i,R_j\in \RAT(H)$. Then we have $L\sm R = \bigcup_{i=1}^n \big(a_iL_i
\sm \bigcup_{j=1}^m b_jR_j \big)$.
Note that if $a_i\notin \{b_1,\dots,b_m\}$, then $a_iL_i \sm \bigcup_{j=1}^m b_jR_j= a_iL_i$, but if
$a_i=b_j$ for some $j$ then $a_iL_i \sm \bigcup_{j=1}^m b_jR_j= a_i(L_i\sm R_j)$. Since $\RAT(H)$ is
an effective relative Boolean algebra, we can compute the rational expression for $L_i\sm R_j$ in $H$. Hence
we can compute the flat rational expression for $L\sm R$.

As a consequence, given any language $B$ as a Boolean combination of flat rational sets, we find a flat rational expression for $B$. Every flat rational expression is a rational expression (over $G$). Deciding emptiness of a rational expression in a monoid with an \enu is trivial.
\end{proof}
\label{sec:GH}

For the remainder of this section we let $M$ be a monoid, 
$G$ be a subgroup of its group of units, and $H$ be a finite index subgroup of $G$. 

Since $\Rat(H)\sse \Rat(G)$, the membership problem 
of $\FRAT{M}{H}$ is a special case 
of the membership problem 
of $\FRAT{M}{G}$. The aim is to prove the converse: the membership problem of $\FRAT{M}{G}$ is reducible to the membership problem of $\FRAT{M}{H}$.

\begin{theorem}\label{thm:cAm}
Let $G$ be a subgroup of the group of units in $M$ and $H\leq G$ be its finite index subgroup. Then we have $\FRAT{M}{G}=\FRAT{M}{H}$
and, for every $M$-NFA $\cA$ which is flat over $G$, there exists an $M$-NFA~$\cB$ which is flat over $H$ such that $\abs \cB$ is polynomial in $\abs \cA$ and such that $L(\cA)=L(\cB)$.

Moreover, suppose that  $[G:H]$ is known and that the monoid $M$ 
has an \enu as in \prref{def:enu}. If 
both the membership problems for $H$ and for $G$ are decidable,
then the construction of the NFA~$\cB$ is effective.
\end{theorem}

The main ingredient of the following proof of \prref{thm:cAm} is the application of \prref{thm:silva}. 
\begin{proof}
Clearly, it is enough to show that $\FRAT{M}{G}\sse\FRAT{M}{H}$.
\WLOG, we assume that the input is specified by 
a trim $M$-NFA $\cA=(Q,\del,q_{\text{in}},q_{\text{fin}})$, which is flat over $G$, such that $q_{\text{in}}$
is the unique initial state 
without any incoming \tra and 
$q_{\text{fin}}$ the unique final state 
without any outgoing \tra. Moreover, $q_{\text{in}}\neq q_{\text{fin}}$. By adding, if necessary, $\eps$-self-loops\footnote{Recall that an $\eps$-\tra in an $M$-NFA is a \tra $p\arc 1 q$, where $1$ is the neutral element of $M$. 
} we may assume that all other states have incoming and outgoing \tra{s}.

For $i=1\lds t$, let $\cA_i=(Q_i,\del_i,I_i,F_i)$ be the set of (disjoint) subautomata of $\cA$ which are induced by the strongly connected components of $\cA$ with a nonempty set of \tra{s}. 
Thus, $q_{\text{in}},q_{\text{fin}}$ are the only states which do not appear in any $\cA_i$.
The initial states $I_i$ (resp., the final states $F_i$) are defined as those states of $\cA_i$ that have incoming (resp., outgoing) \tras in $\cA$ which do not belong to $\cA_i$.
By \prref{def:flatNFA}, each $\cA_i$ is a $G$-NFA. 
Let $1\in R\sse G$ be a finite set of right coset representatives for 
$H$ in $G$. That is, $G$ is the disjoint union 
$G=\bigcup_{f\in R} Hf$ with $1\in R$. 

For each $1\leq i \leq t$ and $f\in R$, there is 
a trim $G$-NFA $\cA_{i,f}=(Q_{i,f},\del_{i,f},I_{i,f},F_{i,f})$ of polynomial size in $\abs \cA \cdot [G:H]$ such that $Q_{i,f}=Q_i\times R$ and
 $L(\cA_{i,f})= L(\cA_{i})\cap Hf$. Note that we have
$|\del_{i,f}|\leq |\del_{i}|$ because for each $i\in \os{1\lds t}$ and
$(p,a,q)\in \del$ there is at most one \tra
$(p,r_p)\arc a (q,r_q)\in \del_{i,f}$, where $r_p$ and $r_q$ are the 
right-cosets given by any path from any state in $I_{i,f}$ to $p$ and $q$, respectively.
This can be shown by using the same idea as in the proof of \prref{thm:silva}.
Hence, $\sum_{1\leq i \leq t}|\del_{i,f}|\leq |\del|$. Moreover, 
we can construct $\cA_{i,f}$ in such a way that $|I_{i,f}|\leq |I_{i}|$ and $|F_{i,f}|\leq |F_{i}|$.

If $M$ has an \enu and the membership problems for $H$ and $G$ are decidable, then the construction of 
each $\cA_{i,f}$ is effective: By exhaustive search we can find right-coset representatives for pairwise different cosets until $[G:H]$ of them are found. 

Introduce a new final state  $p_{i,f}$, and for each 
 $p\in F_{i,f}$  a new \tra $p \arc{\oi f} p_{i,f}$. This leads to a new $G$-NFA $\cA'_{i,f}=(Q'_{i,f},\del'_{i,f},I_{i,f},\os{p_{i,f}})$ such that
 $L(\cA'_{i,f})= L(\cA_{i,f})\oi f \sse H$.
 Since $L(\cA'_{i,f})\in \RAT(G)$, we may apply \prref{thm:silva}. After renaming, we obtain an $H$-NFA $\cB_{i,f}
 =(Q'_{i,f},\del''_{i,f},I_{i,f},\os{p_{i,f}})$ such that 
$L(\cB_{i,f})=L(\cA'_{i,f})$.

To finish the construction of $\cB$, consider a disjoint union of NFAs
\begin{equation}\label{eq:cBB}
\cB= \os{q_{\text{in}},q_{\text{fin}}}\cup \bigcup_{1\leq i\leq {t}, f\in R} \cB_{i,f}
\end{equation}
Thus, $q_{\text{in}}$ and $q_{\text{fin}}$ are reintroduced for the same purpose: $q_{\text{in}}$ becomes the unique initial state and 
$q_{\text{fin}}$ becomes the unique final state. 

For all $f\in F$, we let  $Q_{0,f}=\os{q_{\text{in}}}$  and $Q_{t+1,f}=\os{q_{\text{fin}}}$. 
One after another, consider all pairs 
$(i,j)$ where $0\leq i,j\leq {t+1}$ and $i\neq j$. Then introduce 
for every \tra $p_i\arc{m_{i,j}}q_j\in \del$ with $p_i\in Q_i$ and 
$q_j\in Q_j$ and every $f\in R$, a new \tra 
$p_{i,f}\arc{f \,m_{i,j}}q_{j,f}$ in $\cB$ for every $q_{j,f}\in I_{j,f}$, where 
$p_{i,f}$ is the unique final state in $\cB_{i,f}$. 
This completes the construction of~$\cB$.
\end{proof}

\begin{corollary}\label{cor:Pred}
We have $\FRAT \MQ \GLZ=\FRAT \MQ H$ for every finite index subgroup $H$ of $\GLZ$. 
Moreover, there is  
a polynomial time reduction of the membership problem for $\FRAT \MQ \GLZ$ to the membership problem for $\FRAT \MQ H$. For the reduction we assume that matrices in $\MQ$ are encoded as $4$-tuples of rational numbers written as quotients of binary integers.

More precisely, 
there is a polynomial $p(n)$ such that the following task can be computed 
in $\DTIME(p(n))$:
the input is a $\MQ$-NFA $\cA$, which is flat over $\GLZ$.  The input size is $\Abbin{\cA}$, and the output is a  
$\MQ$-NFA~$\cB$ with  $\Abbin{\cB}\leq p(\Abbin{\cA})$ which is flat over $H$ and satisfies \(L(\cB)=L(\cA)
\). \end{corollary}

\begin{proof}Again, it is enough to show that $\FRAT \MQ  \GLZ\sse\FRAT \MQ H$. It is also obvious that all effectiveness assumptions stated in \prref{thm:cAm} are satisfied for $\MQ$, $G=\GL(2,\Z)$, and $H \leq \GL(2,\Z)$ because $H\in \Rec(\GL(2,\Z))$. 
Since $H$ is not part of the input, we assume that the index $[\GLZ:H]$ and a set $R$ of right-coset representatives is given to us in advance\footnote{In case when $H$ is given by a finite set of generators in $\GLZ$, we can compute in a preprocessing phase
the index $[\GLZ:H]$ and a set of right-coset representatives.},
and we can write 
$\GL(2,\Z)$ as a disjoint union over right-cosets $\GLZ=\bigcup_{r\in R} Hr$.
Following the proof of \prref{thm:cAm} step by step, we see that the algorithm runs in polynomial time because addition, multiplication, and division of binary integers is possible 
in polynomial time. 
Thus, the proof of the corollary is the same as that of \prref{thm:cAm}
by plugging in concrete complexities. \end{proof}

In the special case of $H=\SLZ$, we have $\GLZ=H \cup \vdmatrix{1}00{-1}H$. So, we can replace the application of \prref{thm:silva} inside the proof of \prref{cor:Pred} by the simpler construction of 
\prref{prop:GLSL} (that was illustrating a special case for the use of \prref{thm:silva}).

\section{The membership problem for $\FRAT{\GLQ}{S}$ with $\GLZ\sse S$}\label{sec:membfrat}

The aim of \prref{sec:membfrat} is to prove Theorems~\ref{thm:red}
and~\ref{thm:nsmflat}. With respect to decidability \prref{thm:nsmflat} is stronger than \prref{thm:red} but the known upper bounds on the complexities are different. 
\begin{theorem}\label{thm:red}
On input
$g\in \GLQ$ and a $\GLQ$-NFA 
 $\cA$ that is flat over $\GLZ$, where the  input size is $n=\Abbin g + \Abbin \cA$, it is decidable whether $g\in L(\cA)$
in singly exponential time 
$\DEXPTIME=\DTIME(2^{n^{\Oh(1)}})$.\end{theorem}

\begin{theorem}\label{thm:nsmflat}
On input $g\in \GLQ$ and a $\GLQ$-NFA 
 $\cA$ that is flat over the monoid $\GLZ\cup \set{h \in \GL(2,\Q)}{|\det (h)|> 1}$, where the  input size is $n=\Abbin g + \Abbin \cA$, it is decidable whether $g\in L(\cA)$
in doubly exponential time $\DTIME(2^{2^{n^{\Oh(1)}}})$.
\end{theorem}

The proof of \prref{thm:red} is given in \prref{sec:memfrat} and 
the proof of \prref{thm:nsmflat} is in \prref{sec:memDET}, which is a reduction to the assertion in \prref{thm:red}.
The main difficulty is to show decidability of the membership problem for $\FRAT{\GLQ}{\GLZ}$. 
The complexity follows by a careful, but straightforward, analysis of the decidability proof.

\subsection{The membership problem for $\RAT(\GLZ)$}\label{sec:memslz}
In this subsection we consider a special instance of \prref{thm:red}, where the input is an NFA~$\cA$ such that all labels of \tras are in $\GLZ$, and the problem is to decide whether~$1\in L(\cA)$. 
A special case of this problem was studied in \cite{BellHP2023} by Bell et al. Their main result states that the membership problem for subsemigroups of $\GLZ$ is $\NP$-complete. The proof in~\cite{BellHP2023} is technically demanding and quite elaborate.

In \prref{thm:ratslz}, we show a pseudo-polynomial time complexity\footnote{The complexity of a problem involving integers is called \emph{pseudo-polynomial} if it is polynomial time when integers are given in unary representation.} for deciding whether $1\in L(\cA)$. Our proof is rather simple and avoids compression techniques from~\cite{BellHP2023}. It also keeps the paper self-contained at this point.
We are mainly interested in $\DTIME$-complexities, and  $\NP$-completeness means that there is little hope to find a sub-exponential deterministic 
decision algorithm.
Note that another instance of this problem, the subgroup membership problem in $\GLZ$, was shown to be decidable in polynomial time by Lohrey in~\cite{Lohrey23Tocs}.

\begin{theorem}\label{thm:ratslz}
The following problem can be decided in
 $\DTIME(\Abm{\cA}^{\;\;\Oh(1)})$. 

\noindent{} INPUT: A $\GLZ$-NFA $\cA$ whose unary input size is $\Abm {\cA}$.

\noindent{} QUESTION: $1\in L(\cA)$?
\end{theorem}

\begin{proof}
The commutator subgroup of $\SLZ$ is 
a free subgroup of rank~$2$, and it has index~$24$ in $\GLZ$ by~\cite{Newman62} as we discussed in \prref{rem:glzvf}. 
By \prref{cor:Pred} we can reduce in polynomial time the
problem of deciding $1\in L(\cA)$ to the special instance where
all matrices are in the free subgroup~$F\leq \SLZ\leq \GLZ$. The ambient group 
$\SLZ$ is generated by the matrices $S=\vdmatrix{0}{1}{-1}{0}$ of order~$4$
and $R=\vdmatrix{0}{-1}{1}{1}$ of order~$6$. This is a well known 
classical result, see, for example, \cite[Ch.~8.12]{edam16}.
The free subgroup~$F$ has a finite (and symmetric) generating set 
$\Sig=\oi \Sig\sse \SLZ$ (of size at most~$48$) such that each generator in $\Sig$ can be written as a product over the matrices~$S$ and~$R$ of constant length. 

The inclusion $\Sig\sse F$ induces a canonical 
\hom $\psi$ of $\Sig^*$ onto~$F$. 
In another  polynomial time reduction with respect to the unary input size $\Abm{\cA}$, we replace matrices in~$F$ by words over $\Sig$. Using the ideas of Gurevich and Schupp in~\cite{GurevichS07} for the 
projective linear group $\PSL(2,\Z)$, it is possible to replace a matrix in~$F$ of unary size $\Abm{F}$ by a
word over $\Sig$ of length $\Oh(\Abm{F})$. This is also explained, for example, in \cite[Ch.~8.12]{edam16}.
Next, using more \tras, we can assume that each \tra is labeled with a letter in $\Sig$.  The number of additional \tra{s} is in $\Oh(\Abm{\cA})$. It is this step which would exponentially blow-up the NFA if we used binary representation of integers.  For example, we have $\Abbin{\vdmatrix 1n01}\in \Oh(\log n)$, but  $\vdmatrix 1n01$  written as a word $w_n\in \Sig^*$ has $\abs{w_n}\in \Omega(n)$.

Formally, we obtain $\Sig$-NFA~$\cB$ with  $\psi(L(\cB))= L(\cA)$.
More details are in  
\cite[Prop.~15.4]{DiekertElder2020ijac}. Having constructed~$\cB$,
it remains to decide  whether $1\in \psi(L(\cB))$. For that we use a construction of Benois in~\cite{ben69}. Her aim was
to show that $\Rat(F)$ is closed under comp\-le\-men\-tation. For that 
she transforms first an $\Sig$-NFA~$\cB$ into another
$\Sig$-NFA $\cB'$ such that first, $L(\cB)$ only accepts freely reduced words (these are words without any factor $a\oi a$ for $a\in \Sig$) and second,   
$\psi(L(\cB'))= \psi(L(\cB))$. 
Let us explain why  her transformation of~$\cB$ into $\cB'$ can be performed in polynomial time in $\abs{\cB}$. It uses a so-called ``flooding algorithm'' where it is temporarily allowed to use labels in $\Sig\cup \os 1$. 
For $\cB=(Q,\del,I,F)$ and $p,q\in Q$, we let $\cB[p,q]= (Q,\del,\os p ,\os q)$.
As long there is a letter $a\in \Sig$ with $a\oi a\in  L(\cB[p,q])$, we 
introduce an $\eps$-\tra $p\arc 1 q$ into~$\cB$, unless $1\in L(\cB[p,q])$. Next, we remove all $\eps$-\tras by standard methods. 
So, each time~$\cB$ is changed the number of pairs $(p,q)$ with 
$1\in L(\cB[p,q])$ increases. Therefore the flooding stops after at most $|Q|^2$ rounds. Once it is finished, we see that if $ua\oi a v$ is accepted, then $uv$ is accepted, too. Since $\Sig$ is fixed, 
the time complexity of the entire transformation is polynomial in 
$|Q|+|\del|$. The construction begins and ends with NFAs without $\eps$-\tras. Since $L(\cB')= L(\cB)$ and 
$1\in L(\cB')$ \IFF at least one initial state is final, we are done.\footnote{The algorithm of Benois works in a more general setting, for example, see \cite[Sec.~8.9]{edam16}.}
\end{proof}

\begin{remark}\label{rem:lohrey}
If we started with an input where matrices are written in binary, then the proof of \prref{thm:ratslz}
shows decidability in  $\DEXPTIME$. 
\hspace*{\fill}$\diamond$\end{remark}

\subsection{Proof of \prref{thm:red}}\label{sec:memfrat}
The proof of \prref{thm:red} begins with an input matrix $g\in \GLQ$ and a $\GLQ$-NFA  $\cA$  which is flat over $\GLZ$.  
Since the input $g$ is nonsingular, we can assume that $g=\vdmatrix 1001$.
By \prref{cor:Pred}, we transform the NFA in polynomial time to a $\GLQ$-NFA which is flat over $\SLZ$. Thus, without restriction, the input $\GLQ$-NFA   $\cA$ is flat over $\SLZ$. The problem 
is to decide whether the identity matrix is accepted by $\cA$.

If $1\in L(\cA)$, then there is an accepting path such that the \tra{s} outside $\SLZ$ are used  $t$ times, where $t$ is less than the number of strongly connected components of $\cA$. 
(Otherwise, the NFA $\cA$ were not flat over $\SLZ$.) Since the contribution of every such \tra to $\Abbin{\cA}$ is at least $2$, we have $2t < 1 + \abs Q + 2t \leq \Abbin{\cA}$ and hence $t<\Abbin{\cA}/2$.
Thus, we can nondeterministically guess in polynomial time an initial state~$q_0$ and a sequence of $t$ \tra{s} 
$q_{j-1} \arc {g_j} q_j$ for $1\leq j \leq t$ such that all other \tra{s} (which are used on that path) are labeled with matrices from $\SLZ$. We may assume that $t\geq 1$ because for $t=1$ at least one initial state is final,  and then we have a proof for $1\in L(\cA)$.

Using the above guess, we compute in polynomial
time~$t$ subautomata
 $\cA_j$ of $\cA$ for $1\leq j \leq t$ such that
\[ 1\in L(\cA) \iff 1 \in g_1L(\cA_1)g_2L(\cA_2) \cdots g_tL(\cA_t).
\] 
Note that $\sum_j \Abbin{g_j}
+\Abbin{\cA_j}$ is bounded by a polynomial in $\Abbin{g}+\Abbin{\cA}$.
Next, we write each matrix $g_j$ in its \SNF as $g_j=r_je_j\vdmatrix 100{q_j} f_j$, where $0<r_j\in \Q$, $e_j,f_j\in \SLZ$, and $0\neq q_j\in \Z$.
Let $r= \prod_{j=1}^{t}r_j$ and $q= \prod_{j=1}^{t}q_j$, then 
\[
1 \in g_1L(\cA_1)g_2L(\cA_2) \cdots g_tL(\cA_t)
\]
implies $r^2q=1$.  
Thus, $0<1/r\in \N$ and $0<q\in \N$. For $m=1/r$ we obtain:
\begin{equation}\label{eq:rdivm}
g\in L(\cA)\! \iff \! \vdmatrix {m}00{m} \in e_1\vdmatrix 100{q_1}f_1L(\cA_1) e_2 \vdmatrix 100{q_2}f_2 L(\cA_2)\cdots  e_t \vdmatrix 100{q_t}f_t L(\cA_t).
\end{equation}
\begin{definition}\label{def:UqLq}
Let $0\neq q\in \N$. Then we define two subgroups: 
\begin{align*}
H_{L,q} &=\set{\vdmatrix abcd\in \SL(2,\Z)}{b\equiv 0 \bmod q}\quad \text{and}\\
H_{U,q} &=\set{\vdmatrix abcd\in \SL(2,\Z)}{c\equiv 0 \bmod q}.
\end{align*}
\end{definition}
The images of $H_{L,q}$ and $H_{U,q}$ 
${}\bmod q$
are the subgroups of \emph{lower} and \emph{upper} triangular matrices in $\SL(2,\Z/q\Z)$, which explains the choice of letters $L$ and $U$.
\begin{lemma}\label{lem:Heq}
The subgroups $H_{L,q}$ and $H_{U,q}$ of $\SLZ$ are conjugate in $\GLQ$: 
\begin{equation}\label{eq:cqbq}
\oi {\vdmatrix 100q} H_{U,q}  \vdmatrix 100q= H_{L,q}.
\end{equation}
Moreover their indices in $\SLZ$ are in $\Oh(q \log q)$. \Ip they are of finite index and therefore recognizable subsets in $\SLZ$ and $\GLZ$. 
\end{lemma}
\begin{proof}
\prref{eq:cqbq} is straightforward since
${\vdmatrix 100{1/q}} \vdmatrix abcd  \vdmatrix 100q= \vdmatrix a{bq}{c/q}d$. \prref{prop:indexH} shows the estimation of the index.
\end{proof}
\begin{lemma}\label{lem:nix}
Let $0\neq q\in \Z$ and
$\gcd(b,d)=1$. Then there are integers $1\leq x,y < \abs q$ such that $\gcd(x,y)=1$ and $xb +yd\equiv 0\bmod q$. 
\end{lemma}

\begin{proof}
For $\abs q= 1$, the numbers $x=y=1$ are coprime; and they satisfy $xb +yd\equiv 0\bmod q$ because all integers are congruent modulo $1$. 
Hence, we may assume $2 \leq \abs q$.
 
Let $P_1$ be the set of primes~$p$ such that $\gcd(p,d)=1$ and 
$P_2$ be the set of primes~$p$ such that $p \mid d$. 
Write $q=q_1\cdot q_2$ such that $q_i$ uses primes from $P_i$, only. For every prime~$p$ we have
\begin{align}\label{eq:hugo1}
p\in P_1 \implies \gcd(p,d)= 1& \\ \label{eq:hugo2}
p\in P_2 \implies \gcd(p,b)= 1& \text{, because }\gcd(b,d)=1.
\end{align}
Hence, $d$ is invertible in $\Z/q_1\Z$
and $b$ is invertible in $\Z/q_2\Z$. Therefore we can solve the following congruences.  
\begin{align}\label{eq:con1}
x_1\equiv 1 \bmod q_1 \text{ and }& y_1\equiv -b\oi d  \bmod q_1 \\ \label{eq:con2}
y_2\equiv 1 \bmod q_2 \text{ and }& x_2\equiv -d\oi b \bmod q_2 
\end{align}
Since $\gcd(q_1,q_2)=1$ we obtain by the Chinese remainder theorem $x,y$ with $1\leq x,y<|q|$ such that
\begin{align}\label{eq:conn1}
x\equiv 1 \bmod q_1 \text{ and }& x\equiv -d\oi b  \bmod q_2 \\ \label{eq:conn2}
y\equiv 1 \bmod q_2 \text{ and }& y \equiv -b\oi d  \bmod q_1 
\end{align}
The congruences in (\ref{eq:conn1}) and (\ref{eq:conn2}) tell us that these $x,y$ with $1\leq x,y<|q|$ satisfy 
\begin{align}\label{eq:sunny}
xb +yd \equiv 0 \bmod q.
\end{align}
Indeed, the congruence in (\ref{eq:sunny}) holds ${}\bmod q_1$ and 
${}\bmod q_2$, hence it holds ${}\bmod q$.

We claim that $\gcd(x,y,q)=1$. To see this, let $p\mid q$.
Then $p\mid q_i$ for exactly one $i\in \os{1,2}$. Say $i=1$, then $x\equiv 1 \bmod q_1$ implies $x\equiv 1 \bmod p$ because $p\mid q_1$. Hence, $\gcd(x,q_1)=1$. For $i=2$ we obtain $\gcd(y,q_2)=1$ and therefore $\gcd(x,y,q)=1$ since $q=q_1q_2$, which shows the claim.

It is still possible that there is a prime~$p$ such that $p\mid x$ and $p\mid y$. However, since
$\gcd(x,y,q)=1$ such a prime~$p$ is invertible in $\Z/q\Z$. Thus,
\begin{align}\label{eq:sunn}
\frac{x}{p}b +\frac{y}{p}d \equiv 0 \bmod q.
\end{align}
The property $\gcd(\frac{x}{p},\frac{y}{p},q)=1$ is inherited. 
So we can make $x$ and $y$ smaller. 
Repeating this process a finite number of times, we obtain desired $x$ and $y$ such that $\gcd(x,y)=1$.
\end{proof} 

We use the following well-known fact based on the extended Euclidian algorithm. 
\begin{lemma}\label{lem:gcd}
	Given two $n$-bits integers $a$ and $b$, we can compute in deterministic polynomial time in $n$ 	integers $x$ and $y$ such that $ax+by=\gcd(a,b)$ with $|x|, |y|\leq \max\{|a|,|b|\}$.
\end{lemma}
Actually, using the fact that multiplication and division of $n$-bits integers is possible in soft-linear time, we can give a soft-quadratic time bound
for \prref{lem:gcd}.

\begin{lemma}\label{lem:guessxy}
Let $q\in \Z$ and  $g= \vdmatrix abcd\in \SL(2,\Z)$ be given in binary encoding. Then for $T\in \os{L,U}$ there is a matrix $M=\vdmatrix xyzw\in \SL(2,\Z)$, where $\Abm M\leq q$, such that $g\in M \cdot H_{T,q}$. \Ip since $\Abbin M\leq \log(q)$, we can guess the matrix $M$ nondeterministically and verify in polynomial time that $\oi M \vdmatrix abcd \in H_{T,q}$.
\end{lemma}

\begin{proof}
Since $\vdmatrix abcd\in \SL(2,\Z)$, the entries $b$ and $d$ are coprime. 
By \prref{lem:nix}  there are coprime $x$ and $y$ such that firstly, 
$xb +yd\equiv 0 \bmod q$ and secondly, $1\leq x,y \leq \abs q$.
We can guess these $x,y$. Next, we apply \prref{lem:gcd} 
to obtain $w,z$ in polynomial time such that 
firstly, 
$xw - yz = 1$ and secondly, $\abs w, \abs z \leq \max\os{x,y} \leq \abs q$. We obtain $\vdmatrix xyzw \vdmatrix abcd \in H_{L,q}$. 
This shows the result for $T=L$. The result for $T=U$ is symmetric. 
\end{proof}

Recall that we have reduced via an $\NP$-reduction the problem
of deciding $g\in L(\cA)$ to the problem of deciding in the notation of (\ref{eq:rdivm}) the following membership problem:
\begin{align*}\vdmatrix {m}00{m} \in e_1\vdmatrix 100{q_1}f_1L(\cA_1) e_2 \vdmatrix 100{q_2}f_2 L(\cA_2)\cdots  e_t \vdmatrix 100{q_t}f_t L(\cA_t).
\end{align*}
Firstly, we will do the following preprocessing steps. We conjugate the above equation with~$e_1$ to move it to the end of the expression.
By making, if necessary, $t$ larger and adding dummy NFAs $\cA_i$ of constant size with $L(\cA_i)=\os 1$ we can assume that $t\leq \Abbin{\cA}$ with  $2\leq t\in 2^\N$.  
In the new notation, we let
$m_t=m$ and construct (in polynomial time) NFAs $\cA_{j,t}$, for $1\leq j\leq t$, such that $L(\cA_{j,t})= f_jL(\cA_{j})e_{j+1}$ for $j<t$, and $L(\cA_{t,t})= f_tL(\cA_{t})e_1$.
For convenience, we assume without restriction that 
each $\cA_{j,t}$ is trim with a single initial state 
$p_{j,t}$ without incoming \tra and a single outgoing \tra
$p_{j,t}\arc{f_j}p'_{j,t}$.

The last step finishes the preprocessing phase,  
and the problem becomes to decide whether
\begin{equation}\label{eq:fprepro}
\vdmatrix {m_t}{0}{0}{m_t}\in \vdmatrix 100{q_{1,t}}L(\cA_{1,t})\cdots  \vdmatrix 100{q_{t,t}}L(\cA_{t,t}).
\end{equation}
We now perform at most $\log_2 t \leq \log_2\Abbin{\cA}$ rounds. 
In the~$k$-th round we will have $s=t/2^{k-1}$. 
Each round starts with the problem:
\begin{equation}\label{eq:rond}
 \vdmatrix {m_s}00{m_s}\in \vdmatrix 100{q_{1,s}}R_{1,s} \cdots \vdmatrix 100{q_{s,s}}R_{s,s}
\end{equation}
where 
$0<m_s\in \N$ and $R_{i,s}= L(\cA_{i,s})$ such that for all $(i,s)$ we have $0\neq q_{i,s}\in \Z$ and $\cA_{i,s}$ is an $\SL(2,\Z)$-NFA. In first round, we have $s=t$ and we start with the problem in (\ref{eq:fprepro}). Each round will halve the number 
$s$ until either $s$ becomes $1$ or we know that $g\notin L(\cA)$, for example because $m_s^2\neq \prod_{i=1}^{s} q_{i,s}$.  In such a case, we stop. 
In the~$k$-th round, we perform the following steps {}from \textbf{1} to~\textbf{10}. 
\begin{enumerate}
\item For the sake of simplifying the notation in (\ref{eq:rond}), we rename $m_s$, $R_{i,s}$, and $L(\cA_{i,s})$ as $m$, $R_{i}$, and $L(\cA_{i})$, respectively. Thus, the problem in the~$k$-th round 
becomes to decide whether the following holds 
\begin{equation}\label{eq:newQ}
\vdmatrix m00{m}\in \vdmatrix100{q_1} R_{1} \cdots  \vdmatrix100{q_s}R_{s}.
\end{equation}
Without restriction we have $m^2 = \prod_{i=1}^{s} q_{i,s}$ because 
$R_i\sse \SLZ$ for all $1\leq i \leq s$. 
\item  
Assuming that the assertion in (\ref{eq:newQ}) holds,  
we guess for all even $2\leq i \leq s$  matrices $M_i\in\SLZ$ such that $\Abm{M_i}\leq q_i$ and $R_{i-1}\cap M_iH_{U,q_i}\neq \es$.
This is possible because if (\ref{eq:newQ}) holds, then such $M_i$'s \emph{exist} and \prref{lem:guessxy} gives an upper bound on $\Abm{M_i}$.
\item For all odd $i$ between  $1$ and $s$, we rename in  $\cA_{i}$ the label $f_i$ of the unique outgoing
\tra from the initial state $p_i$ to some $p'_i$ with $M_{i+1}^{-1}f_i$. 
We obtain an $\SLZ$-NFA $\wt \cA_{i}$ such that 
$L(\wt \cA_{i})= M^{-1}_{i+1} L(\cA_{i})$ for all odd $i$.  We do not touch the $\cA_{i}$'s for even $i$.  
\item By \prref{lem:Heq}, we know that the index of $H_{U,q_i}$ in $\SLZ$ is in $\Oh(\abs{q_i}\log \abs{q_i})$. \Ip 
$H_{U,q_i}$ is a recognizable subset of $\SLZ$. This implies that
$\oi{M_i} R_{i-1} \cap H_{U,q_i}$  is rational in $\SLZ$.
More precisely, for all even~$i$, 
using \prref{cor:Pred} we construct in polynomial time an NFA $\cB_i$
such that firstly, the NFA $\cB_i$ accepts 
$\oi{M_i} R_{i-1} \cap H_{U,q_i}$ and secondly, all labels of the \tra{s} are in $H_{U,q_i}$. It is also easy to see that 
the construction keeps the invariant that  $\cB_i$  has a unique initial state with a single outgoing \tra but no incoming \tra. 
\item  For every even $i$, we write
\begin{align*}(R_{i-1}\cap M_iH_{U,q_i})\vdmatrix 100{q_i} &=  M_i(\oi{M_i} R_{i-1} \cap H_{U,q_i})\vdmatrix 100{q_i}\\ 
&= 
M_i\vdmatrix 100{q_i} \big(\vdmatrix 100{1/q_i}(\oi{M_i} R_{i-1} \cap H_{U,q_i})\vdmatrix 100{q_i}\big)\\
&=M_i\vdmatrix 100{q_i} \big(\vdmatrix 100{1/q_i}L(\cB_{i})\vdmatrix 100{q_i}\big).
\end{align*}
\item Define $K_{i}= \vdmatrix 100{1/q_i} L(\cB_{i}) \vdmatrix 100{q_i}$. The NFA for accepting $K_{i}$ is the NFA $\cB_{i}$ where every label $h$ of a \tra is replaced 
by $\vdmatrix 100{1/q_i} h \vdmatrix 100{q_i}$. Since $h\in H_{U,q_i}$, the new labels belong to the subgroup $
H_{L,q_i}$ of $\SLZ$.
\item Define $R_{i}'= K_{i}\cdot R_{i}$ and let $g_i'=\vdmatrix 100{q_{i-1}}M_i \vdmatrix 100{q_i}$ for all even $i$.
For each $g_i'$, compute its \SNF $g_i'= r_i'e'_i \vdmatrix 100{q_i'}f_{i}'$. Thanks to \prref{lem:KB},
it is possible to do in time polynomial in $n= \Abbin g +\Abbin{\cA}$ from inputs $q_{i-1}$, $q_i$ and $M_i$
since $\Abbin{q_{i-1}}$, $\Abbin{q_{i}}$, and $\Abbin{M_i}$ are all bounded by a polynomial in $n$.
\item  Similar to the preprocessing phase which led to (\ref{eq:fprepro}), we push the positive rationals $r_i'$, for each even $i$, to the left by multiplying both sides in (\ref{eq:newQ}) with $1/r_i'$. This yields a new positive natural number $m_{s/2}$ on the left side. Otherwise we have $g\notin R$. 
Since each  $1/r_i'$ is pushed to the left, the new
$g_i'$ is equal to  $e'_i \vdmatrix 100{q_i'}f_{i}'$. 
\item We conjugate~(\ref{eq:newQ}) with $e'_2$ to move it to the end of the expression. 
For even $i$, define $R_i''= f_{i}'R_i'e'_{i+2}$, where $e'_{{s}+2}=e'_2$. 
Overall, we have to verify:
\[
\vdmatrix {m_{s/2}}00{m_{s/2}} \in  \vdmatrix 100{q_2'}R''_2 \cdots  \vdmatrix 100{q'_{s}}R''_{s}.
\]
Note that the concatenation uses only even indices. 
We must have $m_{s/2}^2= \prod_i q'_{2i}$ since, otherwise, we have $g\notin L(\cA)$.
Finally, we let $q_{i,s/2} = q'_{2i}$ and $\cA_{i,s/2}$ be the NFA for $R''_{2i}$.
This finishes one round of the reduction. 
\item If $s/2 = 1$, then we must have $q_{1,1}=m_{1}^2$. Otherwise, we have $g\notin R$. If $s/2 >1$, then we go back to step \textbf{1} 
with the new problem, where the new value of $s$ becomes $s/2$.   
\end{enumerate}
If the procedure above terminates with $s/2=1$, then we end up with the problem of deciding
$\vdmatrix {m_{1}}00{m_{1}} \in  \vdmatrix 100{m_1^2} R_{1,1}$. Due to uniqueness of the \SNF, we must have $m_1=1$ and hence the problem reduces to deciding whether $\vdmatrix 1001 \in R_{1,1}$, which can be done using \prref{thm:ratslz}.

It remains to analyze the time complexity of the procedure
for the input size $n= \Abbin g +\Abbin{\cA}$.
Note that the $\NP$-reduction in the preprocessing step can be replaced by a $\DTIME(2^{n^{\Oh(1)}})$-reduction.
The main procedure stops after at most $\log_2 t \leq \log_2 n$ rounds.
After each round the largest 
value $|q_i|$ is bounded by $m^2$ using the inequality in (\ref{eq:newQ}), and hence  $|q_i|\in 2^{\Oh(n)}$. 
At every stage, in step \textbf{4}, we rely on  product automata construction with an automaton of size $\Oh(|q_i|\log |q_i|)$. This requires $|q_i|^{\Oh(n)}n^{\Oh(1)}$ time.
We also need to compute Smith normal forms in step \textbf{7}, which can be done in time polynomial in $n$.
In step~\textbf{2}, we used a nondeterministic guesses. In a deterministic simulation, we need to run through $\Oh(q_{max}^{4n})$ possibilities, where $q_{max} = \max \{\abs {q_i}\}$.
Hence, the reduction runs in $\DTIME(2^{n^{\Oh(1)}})=\DEXPTIME$ time.

Finally, if we reach $s=1$, then we apply \prref{thm:ratslz}
which eventually decides whether $g\in L(\cA)$ by checking 
in time $\DTIME(2^{n^{\Oh(1)}})=\DEXPTIME$ whether
$\vdmatrix 1001\in L(\cA_{1,1})$.
This concludes the proof. 
\subsection{Proof of \prref{thm:nsmflat}}\label{sec:memDET}
Recall that the statement of \prref{thm:nsmflat} says that on input $g\in \GLQ$ and a $\GLQ$-NFA 
 $\cA$ that is flat over the monoid $S= \GLZ\cup \set{h \in \GL(2,\Q)}{|\det (h)|> 1}$, it is decidable whether $g\in L(\cA)$ in time $\DTIME(2^{2^{n^{\Oh(1)}}})$,
where input size $n$ is defined as $n=\Abbin g + \Abbin \cA$. 
Clearly, $g\in L(\cA)\iff 1\in \oi g L(\cA)$.
Actually, since $\cA$ is flat over~$S$, we can construct 
in polynomial time $S$-NFAs $\cA_i$ and matrices $f_i\in \GLQ\sm S$ for $1\leq i\leq \ell \in \Oh(n)$ such that 
\begin{align}\label{eq:polyred0}
 \vdmatrix 1001\in L(\cA_0)f_1L(\cA_1) \cdots f_\ell L(\cA_\ell)
 &\iff g\in L(\cA).
\end{align}
Let $1\leq m\in \N$ be the greatest common divisor of the denominators of entries in $f_i$ for all $1\leq i\leq \ell$,
which can be computed in polynomial time.
Multiplying both side in (\ref{eq:polyred0}) with~$m$, we obtain:
\begin{align}\label{eq:polyred00}
 \vdmatrix m00m\in L(\cA_0)g_1L(\cA_1) \cdots g_\ell L(\cA_\ell)
 &\iff g\in L(\cA),
\end{align}
where all the $g_i$'s have integer entries. \Ip $g_i\in S$ for all $1\leq i\leq \ell$.
Thus, in polynomial time we find an~$S$-NFA $\cA'$ having the property
\begin{align}\label{eq:polyred000}
 \vdmatrix m00m\in L(\cA')
 &\iff g\in L(\cA).
\end{align}
 Since we can replace the input size $n$ by any polynomial in $n$, we assume for simplicity and without restriction that 
 $\Abbin m\leq n$ and $\Abbin{h}\leq n$ whenever $h$ appears as a label of a \tra in $\cA'$. This implies $m\leq 2^n$ and
 $|\det(h)|\geq 1+2^{-2n}$ whenever $|\det(h)|>1$. 
Assume $\vdmatrix m00m \in L(\cA')$ and let $t$ be the maximal number of times a \tra is used on an accepting path which is labeled by $h$, where~$|\det(h)|>1$. 
 Since $(1+2^{-2n})^{k}> 1+k2^{-2n}$ for all~$k$, we obtain
$t\leq 2^{2n}(m^2-1)\leq 2^{4n}$ because $m\leq 2^n$.
Next, we nondeterministically guess $t\leq 2^{4n}$ transitions
labeled by $h_i$ with~$|\det(h_i)|>1$ and $\GLZ$-subautomata $\cA'_i$ of $\cA'$ for $1\leq i \leq t$ such that
\begin{equation}\label{eq:guess}
	g\in L(\cA) \iff \vdmatrix m00m \in L(\cA')\iff \vdmatrix m00m \in L(\cA'_0)h_1L(\cA'_1) \cdots h_tL(\cA'_t).
\end{equation}
Let $L= L(\cA'_0)h_1L(\cA'_1) \cdots h_tL(\cA'_t)$. Then we have $L\in\FRAT\GLQ \GLZ$ and the language $L$ can be represented 
by some NFA~$\cB$, flat over $\GLZ$, which can be constructed in deterministic time $2^{{\Abbin{\cA}}^{\Oh(1)}}$.
By \prref{thm:red}, we can decide $\vdmatrix m00m \in L(\cB)$ in deterministic time $2^{{\Abbin{\cB}}^{\Oh(1)}}$. 
Altogether, we obtain a deterministic doubly exponential time algorithm to decide $g\in L(\cA)$ as stated in \prref{thm:nsmflat}.
\section{Singular target matrices}\label{sec:ginger}
The aim of this section is to prove the following two theorems. 
\begin{theorem}[The mortality problem]\label{thm:mort}
Given as input a $\MQ$-NFA 
 $\cA$ which is flat over the monoid generated by $\GLZ\cup \Q \cup \os{s_0}$, it is decidable whether $0\in L(\cA)$
in singly exponential time $\DTIME(2^{n^{\Oh(1)}})$ with respect to the  input size $n=\Abbin \cA$.
\end{theorem}
\begin{theorem}\label{thm:sing}
Given as inputs a matrix $g\in \MQ$ and a $\MQ$-NFA 
 $\cA$ which is flat over the monoid generated by $\GLZ\cup \set{r\in \Q}{r>1} \cup \os{0,s_0}$, 
it is decidable whether $g\in L(\cA)$
in doubly exponential time $\DTIME(2^{2^{n^{\Oh(1)}}})$ with respect to the  input size $n=\Abbin g + \Abbin \cA$.
\end{theorem}
The proofs of these theorems are given in \prref{sec:mort} and 
\prref{sec:sing}, respectively.
\subsection{Preliminary calculations}\label{sec:beforeflood}
In this section, we will use the following definitions.

\begin{definition}\label{def:Mija}
For $a\in \Z$ we define
\[
M_{i,j}(a) = \set{\vdmatrix{g_{11}}{g_{12}}{g_{21}}{g_{22}}\in \GLQ\cap \MZ}{g_{ij} = a}.
\]
For $0\neq a\in \Z$ we define
\[
M(a,0) = \set{\vdmatrix{g_{11}}{g_{12}}{g_{21}}{g_{22}}\in \GLQ\cap \MZ}{g_{11} = a \text{ and } g_{21} = 0}.
\]
\end{definition}
In other words, $M_{i,j}(a)$ is the set of nonsingular $2\times 2$ integer matrices where the entry $i,j$ is equal to $a$; and  $M(a,0)$ is the subset of upper triangular matrices in  $M_{1,1}(a)$.
\begin{problem}\label{prob:m11}\
\noindent{} INPUT: An integer $a\in \Z$ and a $(\GLQ\cap \MZ)$-NFA~$\cB$ which is flat over $\GLZ$ and where the  input size is $\Abbin a + \Abbin \cB$.

\noindent{} QUESTION: $M_{i,j}(a)\cap L(\cB)\neq \es$?
\end{problem}
\begin{problem}\label{prob:ma0}\
\noindent{} INPUT: An integer $0\neq a\in \Z$ and a $(\GLQ\cap \MZ)$-NFA~$\cB$ which is flat over $\GLZ$ and where the  input size is $\Abbin a + \Abbin \cB$.

\noindent{} QUESTION: $M(a,0)\cap L(\cB)\neq \es$?
\end{problem}
\begin{problem}\label{prob:glgl}\
\noindent INPUT: $g\in \GLQ$ and a $\GLQ$-NFA 
$\cB$ that is flat over $\GLZ$, where the  input size is $\Abbin g + \Abbin \cA$.

\noindent{} QUESTION: $g\in L(\cB)$?
\end{problem}
Recall that \prref{prob:glgl} is decidable in singly exponential time 
$\DTIME(2^{N^{\Oh(1)}})$ for $N=\Abbin g + \Abbin \cA$ by \prref{thm:red}. The reason to use the letter $N$ here instead of~$n$ is that we will apply \prref{thm:red} later for singular matrices where $N$ is exponential is another parameter~$n$.
\begin{lemma}\label{lem:freddy}
There are $\NP$-reductions of Problems~\ref{prob:m11} and \ref{prob:ma0} to \prref{prob:glgl}.
\Ip we can solve both problems in $\DEXPTIME$ by  \prref{thm:red}.
\end{lemma}

\begin{proof}
We begin with \prref{prob:m11}.
Let $N=\Abbin a + \Abbin {\cB}$ be the input size, and $L(\cB) = g_1L(\cB_1)g_2L(\cB_2) \cdots g_tL(\cB_t)$, where $L(\cB_j)\sse \GLZ$.

Note that $\GLZ$ contains the matrix $\vdmatrix 0110$ such that multiplying any matrix $m\in \MQ$ with $\vdmatrix 0110$ on the left (resp., on the right) swaps the rows (resp., columns) of~$m$. Hence, without restriction, we can assume that $i=j=1$, and the problem is to decide whether $M_{1,1}(a)\cap L(\cB)\neq \es$.

Since~$\cB$ is flat over $\GLZ$, it follows that $\Abbin D$
is bounded by some 
polynomial in~$N$ for every $\vdmatrix abcd\in L(\cB)$, where $N$ is the input size.
If $a=0$, then we have $D=-bc$,
and so we can guess $b$ and $c$. Note that
\[
\vdmatrix 0bcd \vdmatrix 1x01 = \vdmatrix 0bc{d+cx}.
\]
Hence we can guess $0\leq d' \leq \abs c$ such that
\begin{equation*}
\vdmatrix 0bcd \in  L(\cB) \iff \vdmatrix 0bc{d'}\in  L(\cB)\vdmatrix 1101^{\Z}, \end{equation*}
where the question ``$\vdmatrix 0bc{d'}\in  L(\cB)\vdmatrix 1101^{\Z}$\,?'' is an instance of \prref{prob:glgl}.
Here, and in the following, $\vdmatrix 1101^{\Z}$ is a shortcut for $\vdmatrix 1101^{\Z} = \vdmatrix 1101^{*}\cup \vdmatrix 1{-1}01^{*}$.

Thus, we assume $a\neq 0$ for the rest of the proof.
For all $x,y\in \Z$, a straightforward calculation shows: 
\[
\vdmatrix 10y1 \vdmatrix abcd \vdmatrix 1x01 = \vdmatrix a{b+ax}{c+ay}{d+cx+by+axy}.
\]
As a consequence,  there are 
integers ${b'},{c'},{d'}$ with $0\leq \abs{b'},\abs{c'}\leq \abs{a}$
such that 
\begin{equation}\label{eq:leon}
\vdmatrix abcd \in  L(\cB) \iff \vdmatrix a{b'}{c'}{d'}\in  \vdmatrix 1011^{\Z}L(\cB)\vdmatrix 1101^{\Z}. 
\end{equation}
Since $a\neq 0$ and $ad'=(D +{b'}{c'})$, the binary sizes of the integers $b'$, $c'$, and $d'$ are polynomially bounded in $n$. Thus, we can guess the integers ${b'},{c'}$ among exponentially many candidates and compute $d'$. The right-hand side in (\ref{eq:leon})
is again an instance of \prref{prob:glgl}, and we are done with \prref{prob:m11}.

It remains to show an $\NP$-reduction for \prref{prob:ma0}. 
Recall that the problem is to decide whether there exist $b,d\in \Z$ such that $\vdmatrix ab0d \in L(\cB)$. 

Again, since~$\cB$ is flat over $\GLZ$, it follows that $\Abbin D$
is bounded by some 
polynomial in~$N$. Since $D=ad$ and $|D|= |g_1\cdots g_t|$,
where the $g_i$'s are the nonsingular integer matrices defined above, we know that
$d=\pm g_1\cdots g_t/a\in \Z$. So there are only two options for $d$, and we can compute both possibilities in polynomial time if $D$ and all $g_i$'s are written in binary. 
Note that
\[
\vdmatrix ab0d \vdmatrix 1x01 = \vdmatrix a{b+ax}0d.
\]
Hence we can guess $0\leq b' \leq \abs a$ such that
\[
\vdmatrix ab0d \in  L(\cB) \iff \vdmatrix a{b'}0d\in  L(\cB)\vdmatrix 1101^{\Z},
\]
where the question ``$\vdmatrix a{b'}0d\in  L(\cB)\vdmatrix 1101^{\Z}$\,?'' is again an instance of \prref{prob:glgl}.
\end{proof}
\subsection{The flooding procedure}\label{sec:inflood}
Recall that a zero-\tra is a \tra whose label is the zero matrix,
and $s_0$ denotes the matrix $\vdmatrix 1000$.
In the following, a \sztra means a \tra with label $a\cdot s_0= \vdmatrix a000$ where $0\neq a\in \Z$. (The notation is justified in our context since every $\MZ$-matrix of rank  one is in $\SLZ \vdmatrix a000 \SLZ$ with $0\neq a\in \Z$.)
 
The aim of this section we prove \prref{lem:final}, which will be used to show Theorems \ref{thm:mort} and~\ref{thm:sing}.
The key ingredient of \prref{lem:final} is the procedure $\textsc{Flooding}(m,\cA)$; and we begin with an informal description. It has two parameters: 
a natural number $m\in \N$ and a $\MZ$-NFA $\cA$ of input size $\Abbin{\cA}=n$ which is flat over $\GLZ\cup \Z s_0$.
We rewrite in $\DTIME\big(n^{\Oh(1)}\big)$ every label in its \SNF. This makes it possible to assume that the procedure is called only if $\cA$ is a $\MZ$-NFA where each label of a nonzero \tra is either in $\GLQ \cap \MZ$ or a rank-$1$ matrix. 

The idea
 of the ``flooding''  is to introduce more \sztras that can be used as shortcuts for accepting paths
without changing the accepted language $L(\cA)$.\footnote{A similar idea was also used in \prref{sec:memslz} and, as mentioned there, goes back to~\cite{ben69}.} Actually, there are only three cases in the proof of \prref{lem:final}.
Firstly, for $m=0$, the procedure stops as soon as a \tra with the
zero-matrix as a label appears. This is the witness that the zero-matrix is accepted, and hence we stop. 
For $m\neq 0$, we first remove all zero-\tras and we never introduce any 
zero-\tra. 
In the remaining two cases, the procedure either exits with the correct output that $ms_0\notin L(\cA)$ or,
in the third cases, it transforms $\cA$ into an NFA~$\cB$  with $L(\cB) = L(\cA)$ such that if $ms_0\in L(\cB)$, then
it is accepted by a path in $\cB$ which uses a \sztra exactly once. Clearly, the third case is impossible for $m=0$. 
The formal description is in \prref{fig:flo}.

\begin{figure}[ht]
\begin{algorithmic}[1]\small
\Procedure{Flooding}{$m,\cA$}
\State {Run the trimming procedure and denote by~$\cB$ its output.}
\State {If $\cB$ contains a zero-\tra and $m=0$, then output $\vdmatrix 0000 \in L(\cA)$ and exit.}
\State {Otherwise, remove all zero-\tras in $\cB$ and trim it again.}
\Repeat\Comment{\textit{In the beginning of each round, there are no zero-\tras and~$\cB$ is trim.}}

\If{there is no \tra with label $a\cdot s_0$ such that $a\mid m$}
\State {\textrm{Output} $\vdmatrix m000 \notin L(\cA)$ \textrm{and} exit the procedure}.
\EndIf \State  Create a list $\cL$ containing all
triples $(t,t',a)$, 
where $0\neq a\in \Z$ and $t,t'$ are

\hskip\algorithmicindent \sztras $t=(q\Arc{\!a'\cdot s_0} p)$ with $t'=(p'\Arc{a''s_0} q')$ 
such that

\hskip\algorithmicindent 
the product $a'aa''$ divides~$m$.
 \Comment{\textit{Otherwise, $a'aa''$ cannot be used as a label.}}

\ForAll{$(t,t',a)\in \cL$}
\State Let $\cB_{[p,p']}$ be a sub-automaton of $\cB$ containing all
\tras  with labels from 

\hskip\algorithmicindent\ \ $\GLQ \cap \MZ$ in which $p$ is the unique initial and $p'$ is the unique final state.

\Comment{\textit{The procedure behaves differently for $m=0$ and $m\neq0$.}}
\If {$m=0$ and $M_{1,1}(0)\cap L(\cB_{[p,p']})\neq \es$}

\State \textrm{Output} $\vdmatrix 0000 \in L(\cB)$  \textrm{and} exit the procedure.

\Comment{\textit{We are done because~$\cB$ is trim.}}
\ElsIf {$m\neq 0$  and $M_{1,1}(a)\cap L(\cB_{[p,p']})\neq \es$}
\State Introduce a \tra $q\Arc {\!\!a'\!aa''\cdot {s_0}} q'$
(unless it is already present in~$\cB$).
\EndIf \Comment{\textit{See \prref{fig:flood} for an illustration.}}
\EndFor
\Until{the NFA~$\cB$ stabilizes during the body of the repeat-loop in lines 3:--15:}
\EndProcedure
\end{algorithmic}
\caption{The code of the flooding procedure.}
\label{fig:flo}
\end{figure}
\begin{figure}[ht]
\begin{center}
	\begin{tikzpicture}[xscale=2,yscale=1.5]
	\draw (-1,0) node (pi) {$q$};
	\draw (3,0) node (pf) {$q'$};
	\draw (0,0) node (p0) {$p$};
		\draw (2,0) node (p2) {$p'$};
\draw (0.9,1.15) node (pk) {$a'aa''\!\cdot\! s_0$};
\draw[->,>=latex] (pi) to node [above] {$a'\!\cdot\! s_0$} (p0);
\draw[->,>=latex] (p2) to node [above] {$a''\!\cdot\! s_0$} (pf);
\draw[dotted,->,>=latex] (p0) to node [above] {$g_1\cdots g_\ell \in M_{1,1}(a)$} (p2);
\draw[->,>=latex,bend left=45] (pi) to node [left] {} (pf);		
	\end{tikzpicture}
\end{center}
\caption{The flooding procedure introduces a \sztra with label $a'aa''\!\cdot\! s_0$}
\label{fig:flood}
\end{figure}

\begin{lemma}\label{lem:final}
Let $m\in \N$ and $\cA$ be a $\MZ$-NFA  which is flat over 
$\GLZ\cup \Z s_0$. 
Then, for $m=0$ we have $\vdmatrix 0000 \in L(\cA)$ \IFF the procedure $\text{\sc Flooding}(0,\cA)$ in \prref{fig:flo} stops with that output.

For  $m\neq 0$, the procedure works as follows. It either stops and correctly outputs that $ms_0\notin L(\cA)$ or it 
terminates with a trim $\MZ$-NFA which is flat over
$\GLZ\cup \Z s_0$ and has the following two properties:
\begin{enumerate}
\item We have $L(\cB) = L(\cA)$.
\item If $ms_0\in L(\cA)$, then $ms_0$ is accepted by some path where a \sztra is used exactly once. 
\end{enumerate}

Moreover, the flooding procedure can be implemented in $\DTIME(2^{N^{\Oh(1)}})$, where $N = \Abbin{m} + \Abbin{\cA}$.
\end{lemma}
\begin{proof}
Using \prref{lem:silva}, it is easy to see that $L(\cB)=L(\cA)$ is an invariant throughout the procedure. 
If we see a zero-\tra
after the initial trimming then we give the correct answer for $m=0$ in $\DTIME(N^{\Oh(1)})$ and we are done in this case. Thus, we may assume that 
we enter the repeat-loop at least once with a trim $\MZ$-NFA $\cB$ without
zero-\tras. If all labels of \tras in $\cB$ are invertible, we cannot accept any singular matrix. Thus, if $ms_0\in L(\cA)$, then on every iteration of the loop, $\cB$ must contain 
a \tra with label $as_0$ such that $a\mid m$ because $m$ and all labels are integer matrices. Thus, the procedure gives the correct answer 
$\vdmatrix m000 \notin L(\cA)$
whenever it exits in the body of the outer repeat-loop with a negative answer. 
Inside the inner for-loop, the procedure can stop and exit with another correct answer 
$\vdmatrix m000 \in L(\cA)$ for $m=0$. 

The considerations above handle all possible exits, and each time the answer is correct. 
It remains to deal with the case when there are no such exits at all. 
The termination is clear because the flooding must stop eventually.
We claim that every iteration of the repeat-loop shortens every accepting path for $ms_0\in L(\cB)$.

Suppose $m\neq0$ and $ms_0\in L(\cB)$.  Let $\ell$ be the minimal number of \sztras 
used on an accepting path for $m$. For $\ell=1$ there is nothing to do. Hence we may assume
$\ell\geq 2$. Let $q \arc {a's_0} p$ be the first and $p' \arc {a''s_0} q'$ be the second  \sztra on that path. Then $L(B_{[p,p']})$ contains some 
nonsingular matrix $\vdmatrix abcd$, and therefore we have 
$a'aa''s_0\in L(B_{[q,q']})$. We must have $a\neq 0$ and $a\mid m$ because $m\neq0$. After that the flooding procedure can proceed by adding a new \tra $q \larc {a'aa''s_0} q'$, which does not change $L(B)$ because it is just a short cut of an existing path with label $\vdmatrix {a'aa''}000$.
 It is indeed new because $\ell$ was chosen to be minimal, and with $q \larc {a'aa''s_0} q'$ we can find another path which uses less \sztras to accept $ms_0$ than before.
Therefore, for $m\neq 0$, the flooding leads to an accepting path which uses a \sztra exactly once.
Note that this argument also shows that for $m=0$ the procedure must eventually find a witness 
for $\vdmatrix 0000\in L(\cB)$ and exit because the zero-matrix cannot be accepted by a path that uses exactly one \sztra.
This shows the correctness of 
the algorithm in \prref{fig:flo}.

In order to finish the proof of \prref{lem:final}, it remains to analyze its complexity. This  is done as follows. Firstly, the procedure trims $\cA$ andproduces an output $\cB$. Trimming does not change 
the accepted language and does not increase the number of states or \tras.  After that the procedure does not change the state set of~$\cB$ anymore. The number of states is less than~$N$. The set of pairs $(q,q')$ in~$\cB$ with an outgoing (resp.,~incoming) \sztra is not changed, and hence their number is less than $N^2$.
 The number of divisors $a$ of $m$ is at most $\log (m)\leq \log(2^N)$. 
 Hence, it is polynomial in~$N$. Therefore, the number of repeat loops is bounded by a polynomial in~$N$. Constructing the list 
 $\cL$ can be performed in $\DTIME(N^{\Oh(1)})$. Thus, it remains to show that the inner for-all-loop can by implemented to run in
 $\DTIME(2^{N^{\Oh(1)}})$. Within each loop we have to solve an instance of \prref{prob:m11}. We can answer \prref{prob:m11} in $\DTIME(2^{N^{\Oh(1)}})$ by 
 \prref{lem:freddy}. Therefore,
 \prref{lem:final} is proved.
\end{proof}

\subsection{Deciding the mortality problem: proof of \prref{thm:mort}}\label{sec:mort}
Recall that \prref{thm:mort} says that, given as input a
$\MQ$-NFA~$\cA$ which is flat over the monoid generated by $\GLZ\cup \Q \cup \os {s_0}$ of size 
$n=\Abbin{\cA}$, it is decidable in singly exponential time $\DTIME(2^{n^{\Oh(1)}})$  
whether $\vdmatrix 0000 \in L(\cA)$. Without restriction, we assume that
$\cA$ is trim and does not have zero-\tras.

In a preprocessing phase, we compute in polynomial time for every \tra $p\arc{h}p'$
the \SNF of its label as $h=e\vdmatrix r00{rq}f$, where $e,f\in \GLZ$, $q\in \Z$, and $0<r\in \Q$.  After that we replace the label~$h$ by $e\vdmatrix 100{q}f$ which does not change the property whether~$\vdmatrix 0000$ is accepted. Splitting each \tra into at most three \tras
we obtain an NFA~$\cA'$ of size~$N$ which is polynomial in~$n$ such that every \tra has its label in $\GLZ \cup \set{s_q}{q\in \Z}$.
The NFA $\cA'$ is a $\MZ$-NFA which is flat over the set $\GLZ \cup\os{s_0}$.
Since $N=\Abbin{\cA'}$ is polynomial in~$n$, we rename $\cA'$ as~$\cA$ and assume that $n=N$.

After this preprocessing, we run the procedure 
{\sc Flooding}{($m,\cA$)} which, assuming~$n=N$, 
stops in time $2^{n^{\Oh(1)}}$. Recall that if the procedure did not exit 
with the answer $\vdmatrix 0000\in L(\cA)$, then we must have $\vdmatrix 0000\notin L(\cA)$ because otherwise
 it would accept 
$\vdmatrix 0000$ using a path where a rank-$1$ matrix appears at most once and all other labels are invertible matrices, which is impossible.
Therefore, \prref{thm:mort} is shown. \qed

\subsection{Proof of \prref{thm:sing}}\label{sec:sing}
\prref{thm:sing} states that given as inputs a matrix $g\in \MQ$ and a $\MQ$-NFA 
 $\cA$ which is flat over the monoid generated by $\GLZ\cup \set{r\in \Q}{r>1} \cup \os{0,s_0}$, 
it is decidable whether $g\in L(\cA)$
in $\DTIME(2^{2^{n^{\Oh(1)}}})$ with respect to $n=\Abbin g + \Abbin \cA$. 
Thanks to \prref{thm:nsmflat} and \prref{thm:mort}, it is enough 
to prove \prref{thm:sing} when the input 
$g$ is singular but not zero. As usual, we may assume that~$\cA$ is a trim~$\MQ$-NFA without any zero-\tra and which is flat over the monoid generated by 
$\GLZ\cup \set{r\in \Q}{r>1}\cup \os {s_0}$. Note that the assertion of
\prref{thm:sing} does not change if we replace~$n$ by some $n'\in n^{\Oh(1)}$. This allows us to rename $n'$ as $n$ whenever convenient.

As in the proof of \prref{thm:mort}, we start with a preprocessing phase. 
We begin by computing in polynomial time the \SNF of the target matrix $g=e_g\vdmatrix {r_g}000 f_g$ with $0< r_g\in \Q$.  
Multiplying~$g$ by the denominator of $r_g$ and changing $\cA$ by adding to it new initial and final \tras with labels $\oi e_g$ and $\oi f_g$, respectively, we assume without restriction that $g=\vdmatrix {m_g}000={m_g}s_0$
with $1<{m_g}\in \N$ and that the modified NFA is still called $\cA$ with $n=\Abbin{\cA}$.
Next, 
we compute for each \tra $p\arc h p'$ the \SNF of~$h$ as $h=er\vdmatrix 100q f=ers_q f$ with~$e,f\in \GLZ$, $q\in \Z$, and $1<r\in \Q$. We also split the \tra $p\larc {srs_q f}p'$ into at most 3 \tras such that all labels are either in $\GLZ$ or of the form $rs_q$ with $0<r\in \Q$ and $q\in \Z$. Since the \SNF was computed in polynomial time, we can write $r$ as a fraction $r=n_r/m_r$ where~$n_r$ and~$m_r$
are positive natural numbers in $2^{n^{\Oh(1)}}$.

Again, we assume that
the NFA is still called $\cA$ with $n=\Abbin{\cA}$. Since $\cA$ is flat over the monoid generated by $\GLZ\cup \set{r\in \Q}{r>1}\cup \os {s_0}$, there are at most $n$ \tras with a label
$rs_q\notin \GLZ$ where $0< r\leq 1$.
Multiplying $g$ and the labels of these \tras 
with appropriate positive integers in $2^{n^{\Oh(1)}}$, we may assume that $2\leq r\in \N$ for all these \tras and the target matrix is changed to $g'=K{m_g}s_0$ with $K\in 2^{n^{\Oh(1)}}$. To simplify the notation, we rename $g'$ as $g$ and assume that the new automaton is called $\cA$.

This finishes the first phase in the preprocessing.
At this point we have the following situation: the target matrix $g$ is of the form $m s_0$ with $1\leq m\in \N$. The labels of $\cA$ are in either in $\GLZ \cup \os{s_0}$ or of the form  $rs_q$ with $q\in \Z$ and $1+2^{-n} \leq r\in \Q$, where $n=\Abbin{\cA}$.

For the second phase of the preprocessing, we define
a subset $\cT$ of \tras:
\begin{equation}\label{eq:Ttra}
\cT=\set{p\arc{h}p'}
{\exists k,\ell\in \N\,\exists q\in \Z:\;h = (k/\ell) s_q\text{ and }k/\ell>1}.
\end{equation}
Since we have $n=\Abbin{\cA}$, the label $h$ of every \tra in $\cA$ satisfies $\Abbin{h}<n$. Hence, if the label is $rs_q=\vdmatrix {k/\ell}{0}{0}{kq/\ell}$ with $\gcd(k,\ell)=1$, then $\Abbin{k}+\Abbin{\ell} < n$. Thus, $k< 2^n$ and $\ell<2^n$. Moreover, we also have $q<2^{n}$ according to the definitions in \prref{sec:inputsize}.

Suppose that $ms_0\in L(\cA)$. Then there is an accepting path using~$t$ transitions $\tau_j\in \cT$ such that all other \tras on that path are labeled by nonzero integer matrices. Recall that every $\tau_j$ has a label $r_js_{q_j}$ with $r_j \geq 1+2^{-n}$. Since all other matrices on the chosen accepting path have integer entries and $r_j$'s commute with all matrices, we obtain that $(1+ 2^{-n})^t\leq m < 2^n$. Since $1+t2^{-n}\leq (1+ 2^{-n})^t$, we obtain
$t2^{-n}\leq 2^n$, which means that $t\leq 2^{2n}$.

Next, we perform the following $\NTIME\big(2^{\Oh(n)}\big)$-reduction 
which defines a $(\MQ\sm\os{0})$-NFA $\cA'$ 
by guessing a sequence of~$t$ \tras $\tau_j\in \cT$ with label $r_js_{q_j}$ 
and $t+1$ subautomata $\cA_j$ of $\cA$ where all labels of \tras in $\cA_j$ belong to 
$\MZ$
such that:
\begin{align}\label{eq:tear}
g\in L(\cA)&\iff  ms_0\in L(\cA')= L(\cA_0)r_1s_{q_1}L(\cA_1) \cdots r_ts_{q_t}L(\cA_t).
\end{align}
Since each $\cA_i$ is a subautomaton of $\cA$ which does not have \tra from $\cT$, it must be flat over $\GLZ\cup \os{s_0}$. We also have $\Abbin{\cA_i}\leq n$. Recall that we have calculated each $r_j$ as a fraction $r_j=k_j/\ell_j$ where $k_j$, $\ell_j$ are nonzero natural numbers with $k_j,\ell_j<2^n$. 
Thus, in $\DTIME\big(2^{\Oh(n)}\big)$ we can construct a $\MZ$-NFA $\cA''$ such that 
(\ref{eq:tear}) becomes 
\begin{align}\label{eq:tears}
g\in L(\cA)&\iff  m(\prod_{j=1}^t \ell_j)s_0\in L(\cA'')= L(\cA_0)k_1s_{q_1}L(\cA_1) \cdots k_ts_{q_t}L(\cA_t).
\end{align}
Note that we have $\prod_{j=1}^t \ell_j\leq 2^{n2^{2n}}$. Thus, 
$\prod_{j=1}^t \ell_j$ can be very large number which needs $2^{n^{\Oh(1)}}$ bits in binary notation.
We conclude 
that we have  $\Abbin{\cA''}\leq N$ where $N\in \N$ is some computable number in $2^{\Oh(n)}$. The automaton $\cA''$ is large, but it is a $\MZ$-NFA which is flat over 
$\GLZ\cup \os{s_0}$. Hence we can call the procedure
{\sc Flooding}($\ell,\cA''$) according to \prref{fig:flo} where
$\ell =m \prod_{j=1}^t \ell_j\in 2^{N^{\Oh(1)}}$.

Since we assumed that $ms_0\in L(\cA)$, we can guess the automaton $\cA'$ in (\ref{eq:tear}) correctly
and assume that  $\ell s_0\in L(\cA'')$. The output 
of {\sc Flooding}($\ell,\cA''$) is a
$\MZ$-NFA $\cB$ which is flat over
$\GLZ\cup \Z s_0$
such that $\ell s_0\in L(\cB)$ \IFF $\ell s_0$ is accepted by some path which uses
a \sztra $\tau$ exactly once.

We guess $\tau=as_0$ and remove all other
\sztra{s} from $\cB$, which yields a sub-automaton $\cB'$ of $\cB$. Note that if $\ell s_0\in L(\cB')$, then $a$ must divide $\ell$. Hence we can assume without restriction that $a=1$. Next, 
we guess two sub-automata  $\cC_1$ and $\cC_2$ of  $\cB'$ such that 
$\cC_1$ and $\cC_2$ are both $\MZ\cap \GLQ$-NFA which are flat over 
$\GLZ$ and we have
\begin{equation}\label{eq:nzerosing}
g\in L(\cA) \iff \ell s_0 \in L(\cC_1)s_0L(\cC_2) = L(\cC_1)s_0\cdot s_0 L(\cC_2).
\end{equation}
Clearly, the assertion in (\ref{eq:nzerosing}) holds \IFF for $j\in \os{1,2}$ 
there are invertible $\MZ$ matrices $\vdmatrix {a_j}{b_j}{c_j}{d_j}\in L(\cC_j)$ with 
$\vdmatrix \ell000= \vdmatrix {a_1}{0}{c_1}{0} \vdmatrix {a_2}{b_2}{0}{0}$. 
The last equality holds \IFF $a_1a_2=\ell$, $c_1a_2=0$, and $b_2a_1=0$.
Since~$\ell\neq 0$, we conclude~$c_1=0$ and~$b_2=0$. 
There are only~$\log(\ell) \in N^{\Oh(1)}$ possibilities to write 
$a_1a_2=\ell$ in nonzero integers~$a_1$ and~$a_2$. Hence we guess them 
and the assertion in (\ref{eq:nzerosing}) is equivalent to the conjunction of the
following two assertions: 
\begin{align}
\label{eq:onesing}
\exists {b_1},{d_1}\in \Z: 
\vdmatrix {a_1}{b_1}{0}{d_1} \in L(\cC_1)\\
\label{eq:twosing}
\exists {c_2},{d_2}\in \Z:
\vdmatrix {a_2}{0}{c_2}{d_2}\in L(\cC_2)
\end{align}
Using transpositions of matrices, the assertion in (\ref{eq:twosing}) is equivalent 
to $\exists {c_2},{d_2}\in \Z:
\vdmatrix {a_2}{c_2}{0}{d_2}\in L(\cC_2^T)$, where $\cC_2^T$ is obtained 
by reversing the direction of all \tras, interchanging initial and final states, 
and by replacing every label $\vdmatrix {a}{b}{c}{d}$
by its transposition $\vdmatrix {a}{b}{c}{d}^T = \vdmatrix {a}{c}{b}{d}$.
Thus, after this observation, we only need to decide 
the assertion in (\ref{eq:onesing}). This is an instance of 
\prref{prob:ma0} which can be decided in $\DTIME(2^{N^{\Oh(1)}})\sse \DTIME(2^{2{^{n^{\Oh(1)}}}})$ by \prref{lem:freddy}.
This concludes the proof of \prref{thm:sing}.

\section{Conclusion and open problems}\label{sec:conclusion} 
The decidability of membership problems in group theory has a long history going back to the work of Dehn (and others) at the beginning of the 20th century. Of particular interest are the membership problems for 
$\GL(n,\Z)$ and $\GL(n,\Q)$ but as soon as~$n\geq 3$ various natural decision problems become undecidable, whereas the corresponding problems remain open for~$\GL(2,\Q)$.

The contributions of the paper are as follows. On a conceptual level, 
we draw the attention to the family of flat rational sets $\FRAT M S$ of a semigroup $M$ with respect to a subsemigroup~$S$. By definition, $\FRAT M S$ contains  $\RAT(S)$, and it is a subfamily of $\RAT(M)$. For us, the most interesting case is when $S=H$ is a group.\footnote{Recall that $\FRAT M S$ is polynomial closure of $\RAT(S)$ in the terminology of \cite{sch76}. However, we are not aware if his concept was used for decision problems in group theory elsewhere.} In this case $\FRAT M H$ has an inductive definition without reference to a particular presentation of~$M$ or~$H$, see \prref{thm:genfratmon}. This is a rather strong result. It has a remote analogue for finite semigroups when \Schuetz \cite{sch76} characterized
aperiodic semigroups by allowing the star over certain prefix
codes of bounded synchronization delay.

Another main contribution is the 
dichotomy stated in 
\prref{thm:nofinext}. It shows that if a subgroup~$G$ of $\GL(2,\Q)$ contains $\GL(2,\Z)$ and, in addition, a diagonal but not central matrix like 
$\vdmatrix a00d$ with~$\abs a\neq  \abs d$, then~$G$ contains a Baumslag-Solitar group $\BS(1,q)$ with $q\geq 2$ which has infinite index in~$G$. As a consequence, there is no hyperbolic subgroup in $\GL(2,\Q)$ which has $\GL(2,\Z)$ as a proper subgroup. \Ip \wrt inclusion, $\GL(2,\Z)$ is a maximal virtually free group and also a maximal hyperbolic group in $\GL(2,\Q)$. 

We have the following natural 
hierarchy of decision problems in terms of their increasing complexity: 
\begin{itemize}
\item The membership problem for \fg~subgroups.
\item The membership problem for \fg~subsemigroups.
\item The membership problem for rational subsets.
\item Inclusion of rational subsets.
\end{itemize}
For $\GL(2,\Z)$, the inclusion and hence the equality of rational subsets is decidable
because  the family $\Rat(\GL(2,\Z))$ is an effective Boolean algebra. The dichotomy implies that for any subgroup~$G$ in $\GLQ$, which is larger than
$\GL(2,\Z)$, either membership for rational subsets is decidable but 
equality of rational subsets is undecidable or, in the other case, we do not know (when the paper is written) whether membership for \fg~subgroups of~$G$ is decidable. 
These facts were the main motivation to define the notion of a flat rational sets. It pushes the positive decidability results for $\GL(2,\Z)$ further to the relative Boolean algebra 
$\FRAT{\GL(2,\Q)}{\GL(2,\Z)}$ (and beyond if we include nonsingular matrices). Using several structural results for flat rational sets, we proved our main positive decidability results in \prref{thm:nsmflat} for nonsingular matrices and in \prref{thm:mort} and \prref{thm:sing} for singular matrices.

\subsubsection*{Open problems}Potential directions for future research include the following items. 
\begin{itemize}
\item Find other applications of flat rational sets to natural membership problems. For example, when considering $\GL(2,k)$ 
where~$k$ is either an algebraic field over~$\Q$ or a function field in one variable over a finite field.
\item 
Let~$G$ be the subgroup of $\GLQ$ which is generated by $\GLZ$ and $\vdmatrix 100p$ where~$p$ is prime. Is the subgroup membership problem for~$G$ decidable?
\item Several statements of our paper contain complexity bounds but we do not know whether they are sharp. 
For example, \prref{prob:glgl} is~$\NP$-hard, but a proof for~$\NP$-completeness is still missing to date, although a recent work \cite{BellHP2023} might suggest a positive answer.
\item Is the mortality problem decidable for rational subsets of~$\MQ$? 
This problem is equivalent to the following question: given a 
$\GLQ\cap \MZ$-NFA~$\cA$, do there exist $b,c,d\in \Z$ such that
$\vdmatrix 0bcd\in L(\cA)$.
\end{itemize}

\bibliographystyle{siamplain}
\bibliography{refs}
\end{document}